\documentclass[a4paper,USenglish,cleveref]{lipics-v2021}

\hideLIPIcs
\nolinenumbers

\usepackage[numbers,sort&compress]{natbib}
\makeatletter
\def\NAT@spacechar{~}
\makeatother

\usepackage[ruled]{algorithm}

\usepackage[skip=0pt]{subcaption}
\usepackage{sidecap}

\usepackage[inline]{enumitem}

\newlist{thmenum}{enumerate}{1}
\setlist[thmenum]{label=\roman*), ref=\thetheorem\,(\roman*)}
\crefalias{thmenumi}{theorem}

\usepackage{array}

\usepackage{pgfplots}
\usepgfplotslibrary{units,groupplots}
\pgfplotsset{compat=newest}

\definecolor{mlblue}{RGB}{170,190,224}
\definecolor{mlbrown}{RGB}{237,162,86}
\definecolor{mlgreen}{RGB}{124,178,83}
\definecolor{mlred}{RGB}{211,94,96}
\definecolor{mlgray}{RGB}{110,115,115}
\definecolor{mlpurple}{RGB}{188,164,202}
\definecolor{mlox}{RGB}{171,104,87}
\definecolor{mlmoss}{RGB}{204,194,16}

\definecolor{mdblue}{RGB}{57,106,177}
\definecolor{mdbrown}{RGB}{230,134,58}
\definecolor{mdgreen}{RGB}{59,147,78}
\definecolor{mdred}{RGB}{204,37,41}
\definecolor{mdgray}{RGB}{83,81,84}
\definecolor{mdpurple}{RGB}{107,76,154}
\definecolor{mdox}{RGB}{146,36,40}
\definecolor{mdmoss}{RGB}{148,139,61}

\pgfplotscreateplotcyclelist{cons}{%
  {mdblue,every mark/.append style={fill=mlblue},mark=*},%
  {mdred,every mark/.append style={fill=mlred},mark=*},%
  {mdbrown,every mark/.append style={fill=mlbrown},mark=square*},%
  {mdgray,every mark/.append style={fill=mlgray},mark=diamond*},%
  {mdpurple,every mark/.append style={fill=mlpurple},mark=diamond*}%
}
\pgfplotscreateplotcyclelist{query}{%
  {mdred,every mark/.append style={fill=mlred},mark=*},%
  {mdbrown,every mark/.append style={fill=mlbrown},mark=square*},%
  {mdgreen,every mark/.append style={fill=mlgreen},mark=square*},%
  {mdgray,every mark/.append style={fill=mlgray},mark=diamond*}%
}

\definecolor{plotcolor0}{RGB}{255,0,0}
\definecolor{plotcolor1}{RGB}{0,0,255}
\definecolor{plotcolor2}{RGB}{64,170,0}

\pgfplotscreateplotcyclelist{mylist}{%
  plotcolor0, every mark/.append style={solid,fill opacity=0.6,scale=1}, mark=diamond* \\%
  plotcolor1, every mark/.append style={solid,fill opacity=0.6,scale=0.75}, mark=square* \\%
  plotcolor2, every mark/.append style={solid,fill opacity=0.6,scale=1}, mark=* \\%
  plotcolor0, every mark/.append style={solid,fill opacity=0.2,scale=1}, mark=diamond*, dashed \\%
  plotcolor1, every mark/.append style={solid,fill opacity=0.2,scale=0.75}, mark=square*, dashed \\%
  plotcolor2, every mark/.append style={solid,fill opacity=0.2,scale=1}, mark=*, dashed \\%
  plotcolor0, every mark/.append style={solid,scale=1}, mark=diamond, densely dotted \\%
  plotcolor1, every mark/.append style={solid,scale=0.75}, mark=square, densely dotted \\%
  plotcolor2, every mark/.append style={solid,scale=1}, mark=o, densely dotted \\%
}

\makeatletter
\newif\ifhidelinks@hidelinks
\newcommand{\hidelinks}{%
  \hidelinks@hidelinkstrue
  \let\hidelinks@ifHy@colorlinks@status\ifHy@colorlinks
  \let\hidelinks@ifHy@ocgcolorlinks@status\ifHy@ocgcolorlinks
  \let\hidelinks@ifHy@frenchlinks@status\ifHy@frenchlinks
  \let\hidelinks@Hy@colorlink\Hy@colorlink
  \let\hidelinks@Hy@endcolorlink\Hy@endcolorlink
  \let\hidelinks@@pdfborder\@pdfborder
  \let\hidelinks@@pdfborderstyle\@pdfborderstyle
  \hypersetup{hidelinks}%
}
\newcommand{\restorelinks}{%
  \ifhidelinks@hidelinks
    \hidelinks@hidelinksfalse
    \let\ifHy@colorlinks\hidelinks@ifHy@colorlinks@status
    \let\ifHy@ocgcolorlinks\hidelinks@ifHy@ocgcolorlinks@status
    \let\ifHy@frenchlinks\hidelinks@ifHy@frenchlinks@status
    \let\Hy@colorlink\hidelinks@Hy@colorlink
    \let\Hy@endcolorlink\hidelinks@Hy@endcolorlink
    \let\@pdfborder\hidelinks@@pdfborder
    \let\@pdfborderstyle\hidelinks@@pdfborderstyle
  \fi
}
\makeatother

\makeatletter
\pgfplotsset{
    groupplot xlabel/.initial={},
    every groupplot x label/.style={
        at={($({\pgfplots@group@name\space c1r\pgfplots@group@rows.west}|-{\pgfplots@group@name\space c1r\pgfplots@group@rows.outer south})!0.5!({\pgfplots@group@name\space c\pgfplots@group@columns r\pgfplots@group@rows.east}|-{\pgfplots@group@name\space c\pgfplots@group@columns r\pgfplots@group@rows.outer south})$)},
        anchor=north,
    },
    groupplot ylabel/.initial={},
    every groupplot y label/.style={
            rotate=90,
        at={($({\pgfplots@group@name\space c1r1.north}-|{\pgfplots@group@name\space c1r1.outer
west})!0.5!({\pgfplots@group@name\space c1r\pgfplots@group@rows.south}-|{\pgfplots@group@name\space c1r\pgfplots@group@rows.outer west})$)},
        anchor=south
    },
    execute at end groupplot/.code={%
      \node [/pgfplots/every groupplot x label]
{\pgfkeysvalueof{/pgfplots/groupplot xlabel}};
      \node [/pgfplots/every groupplot y label]
{\pgfkeysvalueof{/pgfplots/groupplot ylabel}};
    }
}
\def\endpgfplots@environment@groupplot{%
    \endpgfplots@environment@opt%
    \pgfkeys{/pgfplots/execute at end groupplot}%
    \endgroup%
}
\makeatother

\newenvironment{customlegend}[1][]{%
  \begingroup\csname pgfplots@init@cleared@structures\endcsname\pgfplotsset{#1}%
}{\csname pgfplots@createlegend\endcsname\endgroup}%
\def\addlegendimage{\csname pgfplots@addlegendimage\endcsname}

\usepackage{placeins}

\usepackage{xspace}
\usepackage{booktabs}

\newcommand{\Tstart}{\alpha}
\newcommand{\Tword}{\beta}
\newcommand{\eps}{\ensuremath{\varepsilon}}
\newcommand{\ie}{\emph{i.e.,}\xspace}
\newcommand{\eg}{\emph{e.g.,}\xspace}
\newcommand{\cf}{\emph{cf.}\xspace}

\newcommand{\wmin}{w_{\min}}
\newcommand{\wmax}{w_{\max}}
\newcommand{\Tsel}{T_{\textit{sel}}}
\newcommand{\rand}[1][]{\FuncSty{rand}(#1)}

\newdimen\slantmathcorr
\def\oversl#1{%
\setbox0=\hbox{$#1$}
\slantmathcorr=\wd0
\hskip 0.2\slantmathcorr \overline{\hbox to 0.8\wd0{%
\vphantom{\hbox{$#1$}}}}
\hskip-\wd0\hbox{$#1$}
}

\def\undersl#1{%
\setbox0=\hbox{$#1$}
\slantmathcorr=\wd0
\underline{\hbox to 0.8\wd0{%
\vphantom{\hbox{$#1$}}}}
\hskip-0.8\wd0\hbox{$#1$}
}

\newcommand{\imin}{\undersl{i}}
\newcommand{\imax}{\oversl{i}}
\newcommand{\jmin}{\undersl{j}}
\newcommand{\jmax}{\oversl{j}}
\newcommand{\kmin}{\underline{k}}
\newcommand{\kmax}{\overline{k}}

\newcommand{\set}[1]{\left\{ #1\right\}}
\newcommand{\setsmall}[1]{\{ #1\}}
\newcommand{\setGilt}[2]{\set{#1: #2}}
\newcommand{\seq}[1]{\left\langle #1\right\rangle}
\newcommand{\seqGilt}[2]{\seq{#1: #2}}

\newcommand{\Def}{:=}
\newcommand{\Is}{:=}
\newcommand{\Oh}[1]{\mathcal{O}\!\left( #1\right)}
\newcommand{\Ohsmall}[1]{\mathcal{O}(#1)}
\newcommand{\Ohsmash}[1]{\smash{\Oh{#1}}}
\newcommand{\oh}[1]{\mathrm{o}\!\left( #1\right)}
\newcommand{\Th}[1]{\Theta\!\left( #1\right)}
\newcommand{\Thsmall}[1]{\Theta(#1)}
\newcommand{\Om}[1]{\Omega\left(#1\right)}
\newcommand{\Omsmall}[1]{\Omega(#1)}

\newcommand{\prob}[1]{{\mathbf{P}}\!\left[#1\right]}
\newcommand{\expect}[1]{{\mathbf{E}}\!\left[#1\right]}
\newcommand{\ceil}[1]{\left\lceil #1\right\rceil}
\newcommand{\floor}[1]{\left\lfloor #1\right\rfloor}
\newcommand{\card}[1]{\left\vert{#1}\right\vert}
\newcommand{\punkt}{\enspace .}

\newcommand{\real}{\mathbb{R}}
\newcommand{\rplus}{\mathbb{R}_+}
\newcommand{\nat}{\mathbb{N}}
\newcommand{\Increment}{\raisebox{.12ex}{\hbox{\tt ++}}}

\usepackage[algo2e, ruled, vlined, linesnumbered, norelsize]{algorithm2e}
\DontPrintSemicolon
\SetAlCapHSkip{0pt}
\IncMargin{-\parindent}
\SetProgSty{}
\SetFuncSty{}
\SetArgSty{}
\crefname{algocf}{Algorithm}{Algorithms}

\SetKwComment{tcp}{---\,}{}%
\SetCommentSty{textit}
\newcommand{\Comment}[1]{\tcp*{#1}}
\newcommand{\Commenti}[1]{\tcp*[f]{#1}}

\SetKwProg{Function}{Function}{}{}
\SetKwProg{Procedure}{Procedure}{}{}
\SetKwFor{ParFor}{for}{dopar}{}
\SetKwBlock{Loop}{loop}{}
\SetKw{To}{to}
\SetKw{And}{and}
\SetKw{Break}{break}

\newcommand{\POne}{\makebox{WRS--1}\xspace}
\newcommand{\PWith}{WRS--R\xspace}
\newcommand{\PWithout}{WRS--N\xspace}
\newcommand{\PSubset}{WRS--P\xspace}
\newcommand{\PRes}{WRS--B\xspace}
\newcommand{\PPerm}{WRP\xspace}
\newcommand{\BPOne}{\textbf{\POne}\xspace}
\newcommand{\BPWith}{\textbf{\PWith}\xspace}
\newcommand{\BPWithout}{\textbf{\PWithout}\xspace}
\newcommand{\BPSubset}{\textbf{\PSubset}\xspace}
\newcommand{\BPRes}{\textbf{\PRes}\xspace}
\newcommand{\BPPerm}{\textbf{\PPerm}\xspace}

\newcommand{\twolvl}{\textsf{2lvl}\xspace}
\newcommand{\twoclassic}{\textsf{2lvl-classic}\xspace}
\newcommand{\twosweep}{\textsf{2lvl-sweep}\xspace}
\newcommand{\osens}{\textsf{OS}\xspace}
\newcommand{\osnd}{\textsf{OS-ND}\xspace}
\newcommand{\psa}{\textsf{PSA}\xspace}
\newcommand{\psag}{\textsf{PSA+}\xspace}

\hypersetup{
  colorlinks=true,
  pdftitle={Parallel Weighted Random Sampling},
  pdfauthor={Lorenz Hübschle-Schneider and Peter Sanders},
  pdfsubject={}
}

\title{Parallel Weighted Random Sampling}

\author{Lorenz Hübschle-Schneider}{Karlsruhe Institute of Technology, Germany}{huebschle@kit.edu}{}{}
\author{Peter Sanders}{Karlsruhe Institute of Technology, Germany}{sanders@kit.edu}{}{}
\authorrunning{L. Hübschle-Schneider and P. Sanders}
\acknowledgements{This work was performed on the supercomputer ForHLR funded by the Ministry of
  Science, Research and the Arts Baden-Württemberg and by the Federal Ministry
  of Education and Research.}

\Copyright{Lorenz Hübschle-Schneider and Peter Sanders}

\ccsdesc[500]{Theory of computation~Sketching and sampling}
\ccsdesc[500]{Theory of computation~Parallel algorithms}
\ccsdesc[300]{Theory of computation~Data structures design and analysis}

\keywords{categorical distribution, multinoulli distribution, parallel algorithm, alias method, PRAM, communication efficient algorithm, subset sampling, reservoir sampling}

\relatedversion{Preliminary versions of this paper were presented at the 27th
  European Symposium on Algorithms (ESA 2019) \cite{HubSan19c} and the 32nd ACM
  Symposium on Parallelism in Algorithms and Architectures (SPAA 2020)
  \cite{HubSan20}.  We also draw on material from the first author's
  dissertation \cite{lorenzdiss}.}

\supplement{The code and scripts used for our experiments are available under
  the GPLv3 at \url{https://github.com/lorenzhs/wrs} and
  \url{https://github.com/lorenzhs/reservoir}.}

\begin{document}

\maketitle

\begin{abstract}
  Data structures for efficient sampling from a set of
  weighted items are an important building block of many applications.
  However, few parallel solutions are known.  We close many of these
  gaps both for shared-memory and distributed-memory machines.  We
  give efficient, fast, and practicable parallel algorithms for
  building data structures that support sampling single items (alias
  tables, compressed data structures).  This also yields a simplified
  and more space-efficient sequential algorithm for alias table
  construction. Our approaches to sampling $k$ out of $n$ items
  with/without replacement and to subset (Poisson) sampling are
  \emph{output-sensitive}, i.e., the sampling algorithms use work
  linear in the number of \emph{different} samples. This is also
  interesting in the sequential case.  Weighted random permutation can
  be done by sorting appropriate random deviates.  We show that this
  is possible with linear work using a nonlinear transformation of
  these deviates.  Finally, we give a communication-efficient, highly
  scalable approach to (weighted and unweighted) reservoir
  sampling. This algorithm is based on a fully distributed model of
  streaming algorithms that might be of independent interest.
  Experiments for alias tables and sampling with replacement show near
  linear speedups both for construction and queries using up to 158
  threads of shared-memory machines.  An experimental evaluation of
  distributed weighted reservoir sampling on up to 256 nodes (5120
  cores) also shows good speedups.
\end{abstract}

\section{Introduction}\label{intro}

Weighted random sampling (WRS) asks for sampling items (elements) from a set
such that the probability of sampling item $i$ is proportional to a
given weight $w_i$.  Several variants of this fundamental
computational task appear in a wide range of applications in
statistics and computer science, \eg for computer simulations, data
analysis, database systems, and online ad auctions (see, \eg
Refs. \cite{motwani1995randomized,olken1995survey}). These methods build
on (pseudo)random number generators that provide uniformly distributed numbers (\eg
Ref. \cite{MatNis98}).

Continually growing data volumes (``Big Data'') imply that
the input sets and even the sample itself can become large.
Since actually processing the sample is often fast,
sampling algorithms can easily become a performance bottleneck.
As a result of persistently stagnating clock speeds,
\emph{parallel algorithms} are required to extract performance gains from
advances in hardware.  This includes \emph{shared-memory} algorithms that
exploit current multi-core processors as well as \emph{distributed-memory} algorithms
that split the work across multiple processors without incurring too much overhead
for communication.

However, there has been surprisingly little work on parallel weighted
sampling.  This paper closes many of these gaps.
 Below, the input set $A$ is
assumed to be $1..n$ without loss of generality (throughout this paper we use $i..j$ as a
shorthand for $\set{i,\ldots,j}$).
Item $i$ has weight $w_i$ and let $W \Def \sum_{i=1}^{n}{w_i}$ be the sum of all
weights.  The item's \emph{relative weight} is its share of the total weight,
$w_i/W$.
We consider the following weighted sampling problems:
\begin{description}
\item[\BPOne:] Weighted sampling of \emph{one item} $X$ from a categorical (or multinoulli) distribution
  (equivalent to \PWith and \PWithout for $k=1$), \ie $\prob{X=i}=w_i/W$.
\item[\BPWith:] Sample $k$ items from $A$ \emph{with replacement}, \ie the samples are independent and
  for each sample $X$, $\prob{X=i}=w_i/W$. Let $s \leq k$ denote the number of \emph{different}
  items in the sample $S$. Note that we may have $s\ll k$ for skewed input distributions.
\item[\BPWithout:] Sample $k$ pairwise unequal items $s_1\neq\cdots\neq s_k$
  \emph{without replacement} such that for any sample index $j \in 1..k$ and
  item $i\notin\set{s_\ell\mid \ell<j}$, \kern0.3em
  $\prob{s_j=i}=w_i/(W-\sum_{\ell<j}{w_{s_\ell}})$.%
  \footnote{Some publications use a different definition of \PWithout where
    items are sampled with probability $kw_i/W$. This requires special consideration of items with $w_i>W/k$, \eg Ref. \cite[Example 2]{efraimidis2015stream}.\label{fn:probdef}}
\item[\BPPerm:] Permute the items with the same process as for \PWithout using $k=n$.
\item[\BPSubset:] Sample a \emph{subset} $S \subseteq A$ where $\prob{i\in S}=w_i\leq 1$.
  Also known as Poisson sampling.
\item[\BPRes:] \emph{Batched reservoir sampling}. Repeatedly solve
  \PWithout when batches of new items arrive. Only
  the current sample and batch may be stored. Let $b$ denote the batch size.
\end{description}

\newcommand{\isort}[3]{\ensuremath{\mathrm{isort}^{#3}_{#1}(#2)}}
\newcommand{\iso}[2]{\isort{#1}{#2}{*}}
\newcommand{\mc}[2]{\multicolumn{2}{#1}{#2}}
\newcommand{\mcna}{--- & ---}

\begin{table}[t]
  \caption{Result overview (expected asymptotic running times).  Input size~$n$,
    output size $s$, sample size~$k$, startup latency of point-to-point
    communication $\Tstart$, time for communicating one machine word $\Tword$,
    log-weight ratio $u=\log U=\log \wmax/\wmin$, mini-batches of $b$ items per
    PE.  The complexity of sorting $n$ integers with keys from $0..x$ is
    $\mathrm{isort}_x(n)$ ($\mathrm{isort}^*=$ parallel, $\mathrm{isort}^1=$
    sequential).  Distributed-memory results assume randomly distributed inputs
    for ease of presentation in this table.}
  \centering
    \begin{tabular}{lrllll<{\hspace{2mm}}rll}
      & \multicolumn{3}{c<{\hspace{2mm}}}{Shared memory} & \multicolumn{3}{c}{Distributed mem. w. random data distrib.}\\
      Problem  & §             & Preprocessing  & Query &  §  & Preprocessing                          & Query     \\
      \cmidrule(r){1-1}\cmidrule(r{2mm}){2-4}\cmidrule(r){5-7}
      \POne     & \ref{se}     & $\frac np+\log p$  & $1$          & \ref{dist}       & $\frac np + \Tstart \log p$   & $\Tstart$ \\
      \PWith    & \ref{kwith}  & \iso{u}{n} & $\frac sp + \log n$  & \ref{kwith:dist} & $\isort{u}{\frac np}{1}+\Tstart\log p$   & $\frac sp + \log p$ \\
      \PWithout & \ref{kno}    & \iso{u}{n} & $\frac kp + \log n$  & \ref{ss:nr_dist} & $\isort{u}{\frac np}{1}+\Tstart\log p$   & $\frac kp + \Tstart\log^2 n$ \\
      \PWithout &              &            &                      & \ref{ss:nr_dist} & $\frac np + \Tword u+\Tstart\log p$   & $\frac kp + \Tstart\log n$ \\
      \PPerm    & \ref{perm}   & ---        & \iso{n(u+\log n)}{n} & \ref{perm}       & ---                                   & \iso{n(u+\log n)}{n}\\
      \PSubset  & \ref{subset} & $\frac np + \log p$ & $\frac sp +\log p$ & \ref{subset:dist} & $\frac np + \log p$  & $\frac{s}{p}+\log p$     \\
      \PRes     &              &            &                      & \ref{res} & \mc{r}{\hspace*{2mm}---\hspace*{2mm}\hfill$(b+1)\log(b+k)+\!\Tstart\log^2kp$}     \\
    \end{tabular}
  \label{tbl:results}
\end{table}

Table~\ref{tbl:results} summarizes our results.  When applicable, our
algorithms build a data structure once which is later used to support
fast sampling queries. Basically, the results say that we can solve
most problems with (optimal) linear work and logarithmic (or even
constant) latency. More complicated formulas stem from the fact that
some of the results reduce subproblems to sorting sets of integers.
The exact bounds for integer sorting depend on the machine model and
the range of the sorted keys.  To quantify that, we define
$u\Def \log U$ where\footnote{Throughout this paper, we use $\log{x}$
  to denote $\log_2{x}$, and $\ln{x}$ to denote the natural logarithm
  $\log_e(x)$.}
 $U \Def \wmax/\wmin\Def \max_i w_i/\min_i w_i$ and use
black box formulas explained in Section~\ref{ss:models}.

We are not aware of previous competitive parallel algorithms or more efficient
sequential algorithms.  The distributed algorithms are
refinements of the shared-memory algorithms with the goal to reduce
communication costs compared to a direct distributed implementation.
As a consequence, each PE mostly works
on its local data (the owner-computes approach). Communication -- if
at all -- is only performed to coordinate the PEs and is sublinear in
the local work except for extreme corner cases. The owner-computes
approach introduces the complication that differences in local work
introduce additional parameters into the analysis that characterize
the local work in different situations.  The distributed part of Table~\ref{tbl:results}
therefore covers the case when items are randomly assigned to PEs.
This simplifies the exposition and is often possible in practice.

\paragraph*{Outline and Contributions.}
First, in \cref{prel}, we review the models of computation used
in this paper as well as known techniques we are building on.
We discuss additional related work  in \cref{rel}.
In \cref{s:alias}, we consider problem \POne.
We first give an improved sequential algorithm for constructing
\emph{alias tables} -- the most widely used data structure for problem
\POne that allows sampling in constant time.
Then we parallelize that algorithm for shared and distributed memory.
We also present parallel construction for a more space efficient variant \cite{BLK13}.

Sampling $k$ items with replacement (problem \PWith) seems to be
trivially parallelizable with an alias table.  However, using a
distributed alias table does not lead to a communication-efficient
distributed algorithm, and we can generally do better than repeatedly
sampling a single item for skewed input distributions where the number
of \emph{distinct} output items $s$ can be much smaller than~$k$.
\Cref{kwith} develops an \emph{output-sensitive} algorithm, \ie its
work is linear with respect to the output size. The approach combines
a table-based approach \cite{BLK13} with divide-and-conquer trees
\cite{SandersLHSD18}.  This is interesting both as a parallel and as a
sequential algorithm.

\Cref{kno} employs the algorithm for problem \PWith to
solve problem \PWithout.  The main difficulty here is to estimate the
right number of samples with replacement to obtain a sufficient number
of distinct samples. Then an algorithm for \PWithout without preprocessing
is used to reduce the ``weighted oversample'' to the desired exact output size.

The weighted permutation problem \PPerm can be
reduced to sorting (see \cref{prel:exp}). We show in
\cref{perm} that this is actually possible with linear
work by appropriately defining the (random) sorting keys so that we
can use integer sorting with a small number of different keys. Since
previous linear-time algorithms are fairly complicated~\cite{lang2014riffle},
this may also be interesting for a sequential
algorithm. Indeed, a similar approach might also work for other
problems where sorting can be a bottleneck, \eg smoothed analysis of
approximate weighted matching \cite{MauSan07}.

For subset sampling (problem \PSubset), we parallelize the approach of
\citet{bringmann2017subset} in \cref{subset}.
Once more, the preprocessing requires integer sorting.  However, only
$\Oh{\log n}$ different keys are needed so that linear work sorting
works with logarithmic latency even deterministically on an EREW PRAM.
In \cref{res}, we approach \PRes by introducing a scalable model for
distributed streaming computations. We then adapt the sequential
streaming algorithm of \citet{efraimidis2006reservoir} to this model
by applying distributed priority queues \cite{HubSan16topk}.

\Cref{exp} gives a detailed experimental evaluation of our algorithms for \POne, \PWith, and \PRes. \Cref{concl} summarizes the results and discusses possible future
directions.

\section{Preliminaries}\label{prel}

\subsection{Models and Basic Tools for Parallel Computation}\label{ss:models}

We give the complexity of algorithms using the usual asymptotic notation with
$\Oh{\cdot}$ for upper bounds, $\Om{\cdot}$ for lower bounds, and $\Th{\cdot}$
for simultaneous upper and lower bounds, and $o(\cdot)$ for lower order terms -- see, \eg Ref. \cite{MehSanPar}.

We strive to present our parallel algorithms in a model-agnostic way,
\ie we largely describe them in terms of standard operations such as
prefix sums for which efficient parallel algorithms are known on
various models of computation.  We analyze the algorithms for two
simple models of computation, the PRAM (Parallel Random Access Machine) and a distributed message-passing
model.  In each case, $p$ denotes the number of
processing elements (PEs).  Most of our algorithms achieve
running time polylogarithmic in the input size for a sufficiently large number of
PEs. This is a classical goal in parallel algorithm theory, and we
believe that it is now becoming practically important with the advent
of massively parallel computing and fine-grained parallelism in GPGPUs (general purpose graphics processing units).

For shared-memory algorithms, we use the CREW PRAM model (Concurrent
Read Exclusive Write PRAM) \cite{Jaj92}.  We
will use the concepts of \emph{work} and \emph{span} of a
computation to analyze these algorithms.
The work is the total number of clock cycles invested by all~PEs.
The span is the time needed by a parallel algorithm with an unbounded number of~PEs.
In this paper, the work is usually linear in the input size $n$ and the span logarithmic in~$n$.
This is a textual way to say that the running time is $\Oh{n/p+\log p}$
for inputs of size~$n$ and $p$ PEs.

For distributed-memory computations, we use point-to-point
communication between PEs where passing a message of length $\ell$
takes time $\Tstart+\ell\Tword$.
We assume $1\leq \Tword\leq\Tstart$ where ``1'' is the unit of time for internal computation in the RAM model (basically counting machine instructions).
We will use that \emph{prefix sums} and \emph{(all-)reductions}
can be computed in time
$\Oh{\Tword\ell+\Tstart\log p}$ for vectors of size~$\ell$ \cite[Chapters~13.1--13.3]{MehSanPar}. The \emph{all-gather} operation collects a value
from each PE and delivers all values to all PEs. It can be implemented
to run in time $\Oh{\Tword p+\Tstart\log p}$ \cite[Chapter~13.5]{MehSanPar}.
We will particularly strive to obtain
\emph{communication-efficient algorithms}~\cite{SSM13} where
total communication cost is sublinear in the local computation time
and local input size.
Some of our algorithms are even \emph{communication free} (\eg Section~\ref{kwith:dist}).

We need one basic toolbox operation where the concrete machine model
has some impact on the complexity. Sorting $n$ items with integer
keys from $1..K$ can be done with linear work in many relevant
cases. Sequentially, this is possible if $K$ is polynomial in
$n$ (radix sort).  Radix sort can be parallelized even on a
distributed-memory machine with linear work and span $n^{\eps}$ for
any constant $\eps>0$.  Logarithmic span is possible for
$K=\Oh{\log^c n}$ for any constant $c$, even on an EREW PRAM
\cite[Lemma~3.1]{raja1989sort}.  For a CRCW PRAM, linear work
and logarithmic span can be achieved with high probability when
$K=\Oh{n\log^c n}$ \cite{raja1989sort}\footnote{The paper gives the constraint $K=n$ in
its Theorem~3.1 but the generalization is important for us in \cref{perm} and
already implicitly proven in the paper: simply increase the number of most
significant bits sorted by \emph{Fine\_Sort} from $3 \log \log n$ to
$(3+c+\oh{1}) \log \log n$, which does not change the algorithm's asymptotic
running time~\cite[Lemma~3.3]{raja1989sort}.}.  Resorting to comparison-based
algorithms, we get work $\Oh{n\log n}$ and $\Oh{\log n}$ span on an EREW
PRAM~\cite{Col88}.

\subsection{Bucket-Based Sampling}\label{prel:alias}

The basic idea behind several solutions to problem~\POne is to build a
table of $m\in\Th{n}$ buckets where each bucket represents a total
weight of $\Oh{W/m}$.  Sampling then chooses a bucket uniformly at
random and uses the information stored in the bucket to determine the
actual item \cite{devroye1986,olken1995survey}.  For example, if item weights differ only by a constant
factor $a$, we can use $n$ buckets of width $aW/m$ and simply store
one item per bucket.
Then we can use rejection sampling to obtain constant expected query time.
We use this method in \cref{kwith}. In \cref{ss:succinct} we describe another method
based on rejection sampling that uses $\leq 2n$ buckets of width $W/n$,
splitting each item between $\ceil{nw_i/W}$ buckets.

Sampling in worst-case constant time with only a single memory probe is possible using
Walker's alias table method \cite{walker1977alias}, and its improved
construction due to \citet{vose1991alias}.
An alias table consists of $m\Is n$ buckets of capacity $W/n$ each, where bucket
$b[i]$ represents some part $w_i'$ of the weight of item $i$. We have
$w_i'=w_i$ for \emph{light} items with weight $w_i\leq W/n$.
The remaining weight $w_i-w_i'$ of the \emph{heavy} items with $w_i>W/n$ is
distributed to the remaining capacity $W/n-w_i'$ of the buckets such that
each bucket only represents at most one other item (the \emph{alias} $a_i$).
Algorithm~\ref{alg:vose} gives high-level pseudocode for the approach proposed
by Vose, storing $w_i'$ and $a_i$ as $b[i].w$ and $b[i].a$, respectively.
The items are first classified into light and heavy items.
Each heavy
item is distributed over light items until its residual weight drops below $W/n$.
It is then treated in the same way as a light item.
To sample an item from an alias table, pick a
bucket index $r$ uniformly at random, toss a biased coin that comes up
heads with probability $w_r' \cdot n/W$, and return item $r$ for heads, or its alias $a_r$
for tails.

\begin{algorithm2e}[h]
  \caption{Classical construction of alias tables similar to Vose's approach
    \cite{vose1991alias}.}\label{alg:vose}
  \KwIn{$\seq{w_1,\ldots,w_n} \in \real^n$ the weights of the $n$ input items}
  \KwOut{$b$, an alias table consisting of $n$ pairs $(w,a)$ of (partial) weight
    $w$ and alias $a$}
  \Function{\FuncSty{voseAliasTable}$(\seq{w_1,\ldots,w_n})$}{
    $W \Def \sum_{i=1}^{n}{w_i}$ \Comment{total weight}
    $h \Def \setGilt{i\in 1..n}{w_i>W/n}$ as stack \Comment{heavy items}
    $\ell \Def \setGilt{i\in 1..n}{w_i \leq W/n}$ as stack \Comment{light items}
    $b \Def \seq{(w_1,\bot), (w_2,\bot), \ldots, (w_n,\bot)}$ \Comment{init
      table with item weights, dummy aliases}
    \While(\Commenti{consume heavy items}){$h \neq \emptyset$}{
      $j \Def h.\FuncSty{pop()}$ \Comment{get a heavy item}
      \While(\Commenti{still heavy}){$b[j].w > W/n$}{
        $i \Def \ell.\FuncSty{pop()}$ \Comment{get a light item}
        $b[i].a \Def j$ \Comment{fill bucket $b[i]$ with a piece of item $j$}
        $b[j].w \Def (b[j].w + b[i].w) - W/n$ \Comment{avoid cancellation in $b[j].w\!-\!(W/n\!-\!b[i].w)$}
      }
      $\ell.\FuncSty{push}(j)$ \Comment{item $j$ is light now}
    }
    \Return $b$
  }
\end{algorithm2e}

\newcommand{\selsize}{g}
\newcommand{\selrank}{r}

\subsection{Weighted Sampling using Exponential Variates}\label{prel:exp}

It is well-known that a uniform sample without replacement of size $k$ out of
$n$ items $1..n$ can be obtained by associating with each item a uniform variate
from the interval $(0,1]$, $v_i \Def \rand$, and selecting the $k$ items with
the smallest associated variates $v_i$ (\eg
Refs.~\cite{sunter1977list,fan62res}).  This method can be generalized to
generate a \emph{weighted} sample without replacement by raising uniform
variates to the power of the inverse of the items' weights, that is,
$v_i \Def \rand^{1/w_i}$ and selecting the $k$ items with the \emph{largest}
associated values $v_i$
\cite{efraimidis1999par,efraimidis2006reservoir,efraimidis2015stream}.
Equivalently, one can generate \emph{exponentially distributed random variates} $v_i \Def
-\ln(\rand)/w_i$ and select the $k$ items with the
\emph{smallest} associated $v_i$ \cite{arratia2002expclock}
(``\emph{exponential clocks method}''), which amounts to a simple
$x \mapsto -\ln(x)$ mapping of the above, but is numerically more stable and
easier to generate in a vectorized fashion.

\subsection{Divide-and-Conquer Sampling}\label{prel:daq}

Uniform sampling with and without replacement can be done using a
divide-and-conquer algorithm~\cite{SandersLHSD18}.   To sample $k$
out of $n$ items uniformly and \emph{with replacement}, split the set into
two subsets with $n'$ (left) and $n-n'$ (right) items, respectively. Then the
number of items $k'$ to be sampled from the left has a binomial distribution
($k$ trials with success probability~$n'/n$).
We can generate $k'$ accordingly and then recursively sample $k'$
items from the left and $k-k'$ items from the right.  This can be used
for a communication-free parallel sampling algorithm. We have a tree
with $p$ leaves. Each leaf represents a subproblem of size about $n/p$
-- one for each PE. Each PE descends this tree to the leaf assigned to
it (time $\Oh{\log p}$) and then generates the resulting number of
samples (time $\Oh{k/p+\log p}$ with high probability).  Different PEs
have to draw the same random variates for the same interior node of
the tree.  This can be achieved by seeding a pseudo-random number
generator with an ID of this node.

\subsection{Selection from Sorted Sequences}\label{prel:sel}

Our reservoir sampling algorithm relies on the basic operation of
selecting the $\selrank$ smallest elements from a set.  More concretely, we use
the setting where the elements are distributed over the
PEs but sorted locally. Which selection algorithm to use depends
on the data and requirements.  Let $\selsize$ be the maximum input size at any
PE and $\selrank$ be the rank of the desired item.

\begin{enumerate}[nosep,wide,label=(\arabic*)]
\item \label{selapprox}%
  If $\selrank$ is allowed to be in some range
  $\underline{\selrank}..\overline{\selrank}$ with
  $\overline{\selrank}-\underline{\selrank}\in\Omsmall{\overline{\selrank}/d}$
  for some integer $d$, it is possible to find an item with rank between
  $\underline{\selrank}$ and $\overline{\selrank}$ in expected time
  $\Oh{d \log \selsize + \Tword d + \Tstart \log p}$ \cite[Lemma 3]{HubSan16topk}.

\item \label{selgeneral}%
  If $\selrank$ is fixed, we can still use amsSelect, supplying it with exact
  bounds $\overline{\selrank}=\underline{\selrank}=\selrank$.  This requires
  expected time $\Oh{\log^2 \selrank + \Tstart \log \selrank \log p}$
  \cite[Theorem 2.7]{lorenzdiss}.

\item \label{selrandom} We can also adapt a selection algorithm for unsorted
  inputs, and with high probability, time
  $\Ohsmall{\log^2 \selsize / \log p + \Tword \min(\selsize,
    \sqrt{p}(1+\log_p\selsize))+\Tstart\log(\selsize p)}$ will suffice
  \cite[Chapter 3.1.5]{lorenzdiss}.  For randomly distributed items, this
  reduces to time $\Oh{\log^2 \selsize/\log p + \Tstart \log(\selsize p)}$ with
  high probability.
\end{enumerate}

\section{Related Work}\label{rel}

This paper extends the conference paper \cite{HubSan19c} and the brief announcement \cite{HubSan20}.
It also draws some material from the first author's dissertation \cite{lorenzdiss}.

For the more general field of sampling with unequal probabilities, we refer to
\citet{brewer1983unequal} and \citet{tille2006sampling}, who provide a
comprehensive overview.  For the most part, the methods surveyed therein focus
on estimating the total of a measure that is assumed to correlate with the
items' weights.  Here, we focus on the literature on (efficiently) computing
various kinds of samples from categorical distributions.

\subsection{Sampling one Item (Problem \POne)}

Extensive work has been done on generating discrete random variates from a
fixed categorical distribution
\cite{walker1977alias,vose1991alias,marsaglia2004fixedbase,devroye1986,bringmann2017subset}.
All these approaches use preprocessing to construct a data structure that
subsequently supports very fast (constant time) sampling of a single item.
The best-known approach is the alias method
\cite{walker1977alias,vose1991alias}, which computes in linear time a data
structure that allows solving problem \POne in constant time, \cf
\cref{prel:alias}.
\citet{BLK13} explain how to achieve expected time $r$ using only $\Oh{n/r}$
bits of space beyond the input distribution itself.
Data structures for \POne can also be dynamized \cite{hagerup1993dynamic,vitter2013theory,berenbrink2020dynalias}
to support weight updates and insertions for integer weights.

\subsection{Sampling Without Replacement (Problems \PWithout and \PPerm)}\label{rel:no}

The exponential clocks method of \cref{prel:exp} (see also
Refs.~\cite{efraimidis1999par,efraimidis2006reservoir,efraimidis2015stream,arratia2002expclock,cohen2007bottomk})
is a simple $\Oh{n}$ algorithm for sampling without replacement.  This approach also
lends itself towards use in streaming settings (\emph{reservoir sampling}) and
can be combined with a skip value distribution to reduce the number of required
random variates from $\Oh{n}$ to $\Oh{k \log \frac n k}$ in
expectation~\cite{efraimidis2006reservoir}.  \citet{cohen2007bottomk} mention that this
approach can also be used in the model of \emph{bottom-$k$ sketches}.
A related algorithm for \PWithout with given
inclusion probabilities instead of relative weights is described by \citet{chao1982reservoir}.
\citet{braverman2015float} present another sequential algorithm using a
reduction to sampling \emph{with} replacement (\emph{cascade sampling}).

More efficient algorithms for \PWithout repeatedly
sample an item and remove it from the distribution using a
dynamic data structure \cite{wong1980norep,olken1995survey,hagerup1993dynamic,vitter2013theory}.
With the most efficient such algorithms \cite{hagerup1993dynamic,vitter2013theory} we achieve
time $\Oh{k}$, albeit at the price of an inherently sequential and rather complicated algorithm
that might have considerable hidden constant factors.

It is also possible to combine techniques for sampling \emph{with} replacement
with a rejection method.  However, the performance of these methods depends
heavily on $U$, the ratio between the largest and smallest weight in the input,
as the rejection probability rises steeply once the heaviest items are removed.
\citet{lang2014riffle} gives an analysis and experimental evaluation of such
methods for the case of $k=n$ (\cf ``Permutation'' below).
A recent practical evaluation of approaches that lend themselves towards
efficient implementation is due to \citet{muller2016eval}.

The \emph{permutation} problem \PPerm can be seen as special case of
sampling without replacement with $k=n$. In particular, sorting the
exponential variates from \cref{prel:exp} computes such a permutation.
\citet{lang2014riffle} compares different approaches to problem
\PPerm.

\subsection{Poisson Sampling / Subset Sampling (Problem \PSubset)}

Poisson sampling \cite{hajek1964poisson} is often used in statistical
estimation, \eg Ref.~\cite{uscensusasm}.
It also has applications in graph generation, such as generating Chung-Lu random
graphs \cite{chung2003randgraph}.
Refer to \citet{tille2006sampling} for an overview of methods used for
statistical purposes.  \citet{bringmann2017subset} present an efficient
sequential algorithm.  In computer science literature, Poisson sampling is
sometimes also called \emph{subset sampling}, with equivalent definitions.

\subsection{Parallel Sampling}
Uniform (pseudo)random number generators such as Mersenne Twisters \cite{MatNis98}
can generate random numbers in parallel by using different seeds on each PE.
We use this as a basis for our algorithms.
Beyond that, there is surprisingly little work on parallel sampling.  Even uniform
unweighted sampling had many loose ends until recently
\cite{SandersLHSD18}.  Parallel uniformly random permutations are
covered in~Refs.~\cite{Hag91,San98c}.
Efraimidis and Spirakis note that \PWithout can be solved in parallel with span $\Oh{\log n}$ and
work $\Oh{n\log n}$~\cite{efraimidis1999par}.
They also note that solving the selection problem suffices if the output need not be sorted.
The optimal dynamic data structure for \POne~\cite{vitter2013theory}
admits a parallel bulk update in the (somewhat
exotic) combining-CRCW-PRAM model, where concurrent
write operations are \emph{combined} in a Fetch\&Add step
\cite{gottlieb1983ultracomp}.  This step adds all concurrently written values to
the existing contents of the memory location and also records all prefix sums
that would occur if this process was to be executed sequentially.  Even so, this
does not help with problem \PWithout since batch sizes are one.

\subsection{Reservoir Sampling}
Uniform sampling from data streams has been studied since at least the early
1960s \cite{fan62res} and several asymptotically optimal algorithms are known
\cite{Vit85,li1994reservoir}.  Their key insight is that it is possible to
compute the distance between two samples in constant
time~\cite{devroye1986,li1994reservoir}.  More recently, sampling from the union
of multiple data streams has also received some attention
\cite{chung2016simple,cormode2012continuous,woodruff2011optimal,cormode2010optimal}.
However, in addition to assuming synchronous operation of the nodes and the
network, the \emph{distributed streaming} (or \emph{continuous distributed
  monitoring}) \emph{model} used therein relies on a centralized coordinator
node which stores the entire sample and is the exclusive communication partner
of all other nodes, severely limiting scalability in practice.  An algorithm in
a shared-memory mini-batch model was presented recently~\cite{TT19parstream}.

For \emph{weighted} items, Chao \cite{chao1982reservoir} presents an elegant
algorithm for an alternative definition of weighted sampling (see footnote 1
on page \pageref{fn:probdef}).  The exponential clocks method (see
\cref{prel:exp}) allows for a simple selection of the smallest values
\cite{arratia2002expclock,efraimidis2006reservoir,efraimidis2015stream}.  A
reduction to sampling with replacement eliminates the effects of numerical
inaccuracies of floating-point representations~\cite{braverman2015float}.  To
our knowledge, only Ref.~\cite{JSTW19wres} considers weighted distributed reservoir
sampling.  It uses the distributed streaming model, resulting in challenges
orthogonal to those we face here. %

\section{Alias Table Construction (Problem \POne)}\label{s:alias}

Before discussing parallel alias table construction algorithms
we introduce a simpler and more space efficient sequential
algorithm that is a better basis for parallelization.

\subsection{Improved Sequential Alias Table Construction}\label{seq}
\enlargethispage*{2\baselineskip}

Previous methods need auxiliary arrays/queues of size $\Th{n}$ in order to decide in which order the buckets are filled. \citet{vose1991alias} mentions that this can be avoided but does not give a concrete algorithm.  \citet{GLLD2012gabornoise} provide an implementation of such an algorithm, but do not describe or analyze it.
We now independently describe a method with this property.

The idea is that
two indices $i$ and $j$ sweep the input array with respect to light and heavy items, respectively.
The loop invariant is that the weight of items corresponding to light (heavy) items preceding $i$ ($j$) has already been distributed over some buckets and that their corresponding buckets have
already been constructed.
Variable $w$ stores the weight of the part of item $j$ that has not yet been assigned to buckets.
Each iteration of the main loop
advances one of the indices and initializes one bucket.
When the residual weight $w$ exceeds $W/n$, item $j$ is used to fill bucket $i$,
the residual weight $w$ is reduced by $W/n-w_i$.
Otherwise, the remaining weight of heavy item $j$ fits into bucket $j$ and
the remaining capacity of bucket $j$ is filled from the next heavy item.
Algorithm~\ref{alg:simpleSweep} gives pseudocode that emphasizes the
high degree of symmetry between these two cases.
\Cref{fig:example} gives an example, where \cref{fig:splitexample} shows a
snapshot of the state of the algorithm after processing 7 of 13 items, and
\cref{fig:aliasexample} shows the resulting alias table.
\clearpage

\begin{algorithm2e}[p]
  \caption{A sweeping algorithm for building alias tables.}\label{alg:simpleSweep}
  \KwIn{$\seq{w_1,\ldots,w_n} \in \real^n$ the weights of the $n$ input items\\ \quad assume sentinel items $w_{n+1}=\infty$ and $w_{n+2}=0$ to avoid some special case treatments}
  \KwOut{$b$, an alias table consisting of $n$ pairs $(w,a)$ of (partial) weight
    $w$ and alias $a$}
  \Function{\FuncSty{sweepingAliasTable}$(\seq{w_1,\ldots,w_n})$}{
    $W \Def \sum_{i=1}^{n}{w_i}$ \Comment{total weight}
    $i \Def \min\setGilt{k>0}{w_k \leq W/n}$ \Comment{first light item}
    $j \Def \min\setGilt{k>0}{w_k > W/n}$ \Comment{first heavy item}
    $w \Def w_j$ \Comment{current heavy item}
    \lIf(\Commenti{All weights are equal}){$j=n+1$}{$\forall k = 1..n: b[k].p = w_k; b[k].a=k;$}
    \While{$j \leq n$}{
      \If(\Commenti{Pack a light bucket.}){$w>W/n$}{

        $b[i].w\Def w_i$ \Comment{Item $i$ completely fits here.}
        $b[i].a\Def j$ \Comment{Item $j$ fills the remainder of bucket $i$.}
        $w \Def (w + w_i) - W/n$ \Comment{Update residual weight of item $j$.}
        $i\Def \min\setGilt{k>i}{w_k \leq W/n}$\Comment{next light item}
      }\Else(\Commenti{Pack a heavy bucket.}){
        $b[j].w\Def w$ \Comment{Now item $j$ completely fits here.}
        $j'\Def \min\setGilt{k>j}{w_k > W/n}$ \Comment{next heavy item}
        $b[j].a\Def j'$;\quad $j \Def j'$\Comment{Proceed with item $j'$}
        $w \Def (w + w_{j'}) - W/n$ \Comment{Compute residual weight avoiding cancellation issues}
      }
    }
  }
\end{algorithm2e}
\begin{figure}[p]
  \begin{subfigure}[t]{\linewidth}\centering
    \includegraphics{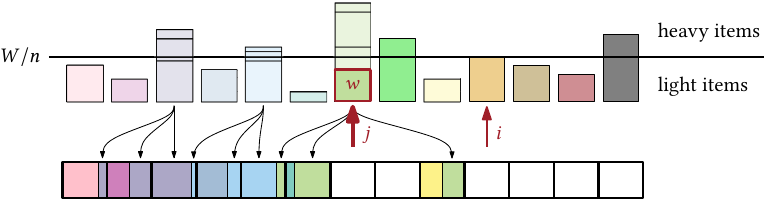}
    \vspace*{0.75em}
    \caption[State after packing half of the buckets]{State of the algorithm
      after packing 7 of 13 buckets.} \label{fig:splitexample}
  \end{subfigure}
  \begin{subfigure}[b]{\linewidth}
    \vspace*{.75em}\centering
    \includegraphics{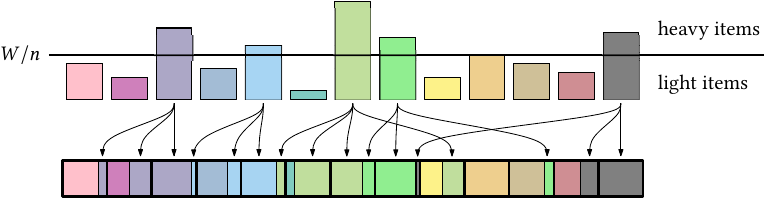}
    \vspace*{0.75em}
    \caption[The finished alias table]{The finished alias table.}
    \label{fig:aliasexample}
  \end{subfigure}
  \caption[Example of alias table construction]{Example of alias table
    construction.  Items' weights are given by their boxes' heights, colors are
    for illustration purposes only.  Each bucket contains a piece of its item,
    aligned to the left in the bucket, and optionally an alias, aligned right
    and marked with an arrow from its item for illustration.}
  \label{fig:example}
\end{figure}
\FloatBarrier

\subsection{Shared-Memory Parallel Alias Tables}\label{se}

The basic idea behind our \emph{splitting based algorithm} is to parallelize the
sequential algorithm of the previous subsection by determining its state at
carefully chosen points during its execution.  In the example of
\cref{fig:example}, splitting the computation roughly in half means
reconstructing the state of \cref{fig:splitexample} without actually executing
the sequential algorithm.  This state-reconstruction can be done in parallel for
$p-1$ positions in logarithmic time, splitting the work into $p$ packets of
equal size, one for each PE\@.  The splitting positions are chosen
in such a way that the work performed by the sequential algorithm
between two neighboring positions is a $1/p$ share of the total work.
This work represents a subproblem that is assigned to one PE\@.
A state is defined by the indices of the next heavy and light item
and by how much of the current heavy item remains to be assigned.
The resulting subproblems consist of consecutive subsets of light and heavy
items, where the heavy items may overlap with preceding or subsequent
subproblems at the boundaries.  A subproblem may receive a piece of a heavy item
from a preceding subproblem, and/or cede a piece of the last heavy item to one
or more subsequent subproblems.  The items of these
subproblems can be allocated precisely within their respective buckets.  We can
thus handle the subproblems completely independently in parallel, requiring no
further synchronization or communication.  The splitting of a heavy item over
two or more subproblems does not complicate this beyond the sequential
algorithm: pieces of it are assigned as aliases of light items in all but the
last subproblem where it occurs.  Thus, the resulting data structure is still an
alias table.

\begin{algorithm2e}[p!]
  \caption{Pseudocode for parallel splitting based alias table construction (PSA).}\label{alg:psa}
  \KwIn{$\seq{w_1,\ldots,w_n} \in \real^n$ the weights of the $n$ input items}
  \KwOut{$b$, an alias table consisting of $n$ pairs $(w,a)$ of (partial) weight
    $w$ and alias $a$}
  \Function{\FuncSty{psaAliasTable}$(\seq{w_1,\ldots,w_n})$}{
    $W \Def \sum_{i=1}^{n}{w_i}$ \Comment{total weight}
    $h \Def \seqGilt{i \in 1..n}{w_i>W/n}$ \Comment{parallel traversal finds heavy items...}
    $\ell \Def \seqGilt{i \in 1..n}{w_i \leq W/n}$ \Comment{... and light items}
    \ParFor{$k \Def 1$ \To $p-1$}{
      $(i_k, j_k, \text{spill}_k) \Def \FuncSty{split}(\ceil{nk/p})$ \Comment{split into $p$ pieces}
    }
    $(i_0,j_0,\text{spill}_0) := (0, 0, 0) \in \nat\times\nat\times\real$ \Comment{cover boundaries}
    $(i_p,j_p) := (n,n) \in \nat \times \nat$\;
    \lParFor{$k\Def 1$ \To $p$}{
      \FuncSty{pack}$(i_{k-1}+1,i_k,\quad j_{k-1}+1, j_k,\quad \text{spill}_{k-1})$
    }
  }

  \BlankLine
  \BlankLine
  \KwIn{$n' \in \nat$, the number of buckets to be filled by the left subproblem}
  \KwOut{$i,j\in\nat$, the last light and heavy items to be filled in the left
    subproblem, $s\in \real$ the partial weight of heavy item $j+1$ \emph{not}
    to be used by the left subproblem}
  \Function{\FuncSty{split}$(n')$}{
    $a \Def 1$; \quad
    $b \Def \min(n', \card{h})$ \Comment{$a..b$ is the search range for $j$}
    \Loop(\Commenti{binary search}){
      $j \Def \floor{(a+b)/2}$ \Comment{bisect search range}
      $i \Def n'-j$ \Comment{Establish the invariant $i+j=n'$}
      $\sigma \Def \sum_{x\leq i}w_{\ell[x]}+\sum_{x\leq j}w_{h[x]}$ \Comment{work to the left; use precomputed prefix sums}
      \If{$\sigma\leq n'W/n$ \And $\sigma+w_{h[j+1]}> n'W/n$}{
        \Return $(i,j,w_{h[j+1]}+\sigma-n'W/n)$
      }
      \lIf*{$\sigma\leq n'W/n$}{
        $a\Is j+1$
      }\lElse{
        $b\Is j-1$ \Commenti{narrow search range}
      }
    }
  }

  \BlankLine
  \BlankLine
  \KwIn{$\imin..\imax$ the range of light items to assign,
    $\jmin..\jmax$ the range of heavy items to assign,
    $s$ the maximum amount of heavy item $\jmin-1$ to use}
  \Procedure{\FuncSty{pack}$(\imin, \imax, \jmin, \jmax, s)$}{
    $i \Def \imin$; \quad  %
    $j \Def \jmin-1$; \quad
    $w \Def s$  \Comment{$\ell[i]$ and $h[j]$ are the current light/heavy items,}
    \If(\Commenti{$w$ is the part of the current heavy item still to be assigned}){$s=0$}{
      $j\Increment$;\quad $w\Def w_{h[j]}$;
    }
    \Loop{
      \If(\Commenti{pack a heavy bucket}){$w \leq W/n$}{
        \lIf{$j>\jmax$}{\Return}
        $b[h[j]].w\Is w$\;
        $b[h[j]].a\Is h[j+1]$\;
        $w \Is (w + w_{h[j+1]}) - W/n$; \quad $j\Increment$\;
      }\Else(\Commenti{pack a light bucket}){
        \lIf{$i>\imax$}{\Return}
        $b[\ell[i]].w\Is w_{\ell[i]}$\;
        $b[\ell[i]].a\Is h[j]$\;
        $w \Is (w + w_{\ell[i]}) - W/n$; \quad $i\Increment$
      }
    }
  }
\end{algorithm2e}

\begin{SCfigure}[][b]
  \includegraphics{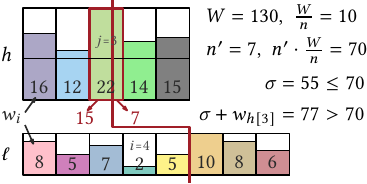}
  \caption{Parallel alias table construction: splitting $n=13$ items into two
    parts of size $n'=7$ and $n-n'=6$.  Continuation of the example of
    \cref{fig:example}.}\label{fig:split}
\end{SCfigure}

We first explain how to split $n$ items into two subsets of size
$n'$ and $n-n'$.  Similar to Vose's algorithm,
we first compute arrays $\ell$ and
$h$ containing the indices of the light and heavy items, respectively.  We
then determine indices $i$ and $j$ such that $i+j=n'$ and
\[ \sigma\Is\sum_{x\leq i}w_{\ell[x]}+\sum_{x\leq j}w_{h[x]}\leq n'W/n
  \quad\text{and}\quad \sigma+w_{h[j+1]}> n'W/n\punkt \] %
These values can be
determined by binary search over the value of $j$.  By precomputing prefix sums
of the weights of the items in $\ell$ and $h$, each iteration of the
binary search takes constant time.
The resulting subproblem then consists of the light items $L\Is\set{\ell[1],\ldots,\ell[i]}$,
the heavy items $H\Is\set{h[1],\ldots,h[j]}$ and a piece of size $n'W/n-\sigma$ of
item $h[j+1]$.   See \cref{fig:split} for a continuation of the example of
\cref{fig:example}.

To split the input into $p$ independent subproblems of near-equal size, we
perform the above two-way-split for the values $n'_k=\ceil{nk/p}$
for $k\in 1..p-1$.  PE $k$ is then responsible for filling a set of
buckets corresponding to sets of light and heavy items, each
represented by a range of indices into $\ell$ and $h$. A piece of a
further heavy item may be used to make the calculation work out. Note
that a subproblem might have an empty set of light or heavy items
and that a single heavy item may be assigned partially to multiple
subproblems, but only the last PE using it
will fill its bucket.

Algorithm~\ref{alg:psa} gives detailed pseudocode.
It uses function \emph{split} to compute $p-1$ different splits in parallel.
The result triple $(i,j,s)$ of \emph{split} specifies that buckets $\ell[1]\cdots\ell[i]$ and
$h[1]\cdots h[j]$ shall be filled using the left subproblem and a part of weight $s$ from item $h[j+1]$
is \emph{not used} on the left side, \ie spilled over to the right side.

This splitting information is then used to make $p$ parallel calls to procedure \emph{pack} -- giving each PE the task to fill at most $\ceil{n/p}$ buckets. \emph{Pack} has input parameters specifying ranges of heavy and light items it should use. The parameter \emph{spill} determines how much weight of item $h[\jmin-1]$ can be used for that.
\emph{Pack} works similar to the sweeping algorithm from Algorithm~\ref{alg:simpleSweep}.
If the residual weight of item $h[\jmin-1]$ drops below $W/n$, this item is also packed in this call. The body of the main loop is dominated by one if-then-else case distinction. When the residual weight of the current heavy item falls below $W/n$, its bucket is filled using the next heavy item.
Otherwise, some of its weight is used to fill the current light item.

\begin{theorem}\label{thm:paralias}
  We can construct alias tables with
  work $\Th{n}$ and span
  $\Thsmall{\log n}$ on CREW PRAMs.
\end{theorem}
\begin{proof}
  The algorithm requires linear work and logarithmic span for
  identifying light and heavy items and for computing prefix sums
  \cite{Ble89} over them. Splitting works in logarithmic time.
  Then each PE needs time $\Oh{n/p}$ to fill the buckets
  assigned to its subproblem.
\end{proof}

\subsubsection{Performance Optimization in Practice}\label{se:opt}

Observe that we can improve the practical performance of this algorithm by
greedily assigning ranges of items to the PEs before we begin the splitting,
with each PE stopping when it runs out of either light or heavy items in its
allotment, and applying PSA  only to those items that could
not be assigned.  By directly assigning some items, the number of items for
which the auxiliary data structures of Algorithm~\ref{alg:psa} have to be
computed is reduced.  However, the number of items that can be handled this way
is dependent on the input distribution and may be arbitrarily small.
Nonetheless, this
optimization can yield significant improvements for typical inputs (see
experiments in \cref{exp:cons}, where this method is labeled \psag).

More concretely, in a first step, PE $i$ handles items $n(i-1)/p+1$ to $ni/p$.
In this range, the PEs run the
sequential assignment algorithms on their subset of the input, stopping once
the next light or heavy item would be beyond their part of the input (\ie the
item's index exceeding $ni/p$).  Algorithm~\ref{alg:psa} is then only executed on
these remaining items.  To do so, each PE records the index of the first item
that could not be assigned greedily and the number of items in its
allotment that were not assigned.  The PEs then compute a prefix sum over the
latter value.  This gives them the global rank of their unfinished items, which
is required in the index calculations (arrays $h$ and $\ell$).
Algorithm~\ref{alg:psa} then proceeds as before.
The only necessary change is that
finished items are skipped when computing the arrays $h$ and $\ell$.

\subsubsection{Using Block-Wise Prefix Sums}\label{se:block}

We can further optimize memory usage by employing \emph{block-wise} prefix sums
in our parallel alias table construction algorithm.  The idea is for each block
to represent logarithmically many -- \ie $\log(n)$ -- items for prefix sum
computation, resulting in a prefix sum with $n/\log(n)$ entries.  Each entry
additionally stores the number of light (heavy) items thus far, \ie the rank of
the last item in the block.  This allows us to replace both the $\ell$ and $h$
arrays and the full prefix sums over the weight of the light and heavy items.

Consider first the aforementioned prefix sums, which are used to split the input
into subproblems.  We can now use this coarse prefix sum to narrow the search range to at most two
blocks, one of which contains the heavy item that is the true splitter, and at most two
candidate blocks for the matching light item.  To determine the true splitters,
we now need to consider only the items in the candidate blocks.
Extracting these blocks -- \ie refining the weight prefix sums to the item level -- takes
logarithmic time as we only have to compute a prefix sum of the weights of the
items within them, and the search can then proceed on these items using the
regular splitting function of \cref{alg:psa}.  Observe that this does not
increase the asymptotic work or span of the algorithm.  The prefix sums are used
only in the binary search, the modified version of which takes time
$\Oh{\log(n/\log n)+\log n} = \Oh{\log n}$, which is unchanged from before.

Observe that this change also allows us to get rid of the $\ell$ and $h$ arrays
in \cref{alg:psa}, and we only need to store how many heavy items occur up to
each block.  In \emph{pack}, after an initial random access into $h$ -- which we replace with a
binary search on the aforementioned prefix sum -- all subsequent accesses are a
single linear scan and can thus be replaced with a scan over the input weights
as in \cref{alg:simpleSweep}.
Pseudocode and additional analyses are given in the dissertation of the first
author \cite[Chapter~3.3.2.b)]{lorenzdiss}.

These changes should yield a speed increase overall because alias table
construction is not limited by computation but by the throughput of the memory
(see \cref{exp}).  Using this optimization, we reduce memory use for the
auxiliary arrays from $2n$ to $3n/\log n$ words (we only need to store one of
the count prefix sums, the other is trivially derivable because the sum of both
is always the total number of items thus far, which is known).

\subsection{Compressed Data Structures for \POne}\label{ss:succinct}

\citet{BLK13} give a construction similar to alias tables that allows expected
query time $\Oh{r}$ using $2n/r+o(n)$ bits of additional space.  We describe the
variant for $r=1$ in some more detail (a generalization for $r>1$ is straight-forward).  We assign $\ceil{nw_i/W}$ buckets to
each item, \ie at most $2n$ in total.  Item $i$ is assigned to buckets
$k=\sum_{j<i}\ceil{nw_j/W}$ to
$\sum_{j\leq i}{\ceil{nw_j/W}}-1=k+\ceil{nw_i/W}-1$, inclusively.  A query
samples a bucket $j$ uniformly at random.  Suppose bucket $j$ is assigned to
item $i$.  If $j\in k..k+\ceil{nw_i/W}-2$, \ie $j$ is not the last bucket item
$i$ is assigned to, item $i$ is returned.  If $j=k+\ceil{nw_i/W}-1$,
then item $i$ is returned with probability $nw_i/W-\floor{nw_i/W}$.  Otherwise,
bucket $j$ is rejected and the query starts over.  Since the overall success
probability is $\geq 1/2$, the expected query time is constant.

The central observation for compression is that it suffices to store
one bit for each bucket that indicates whether a new item starts at
bucket $b[i]$. When querying bucket $j$, the item stored in it can be
determined by counting the number of 1-bits up to position $j$. This
\emph{rank}-operation can be supported in constant time using an
additional data structure with $o(n)$ bits.  Further reduction in
space is possible by representing $r$ items together as one entry in
$b$.

Both constructing the bit vector and
constructing the rank data structure is easy to parallelize using
prefix sums (for adding scaled weights and counting bits, respectively)
and embarrassingly parallel computations.  \citet{Shun2017improved}
even gives a bit parallel algorithm needing only $\Oh{n/\log n}$ work
for computing the rank data structure. We get the following result:

\begin{theorem}
  Bringmann and Larsen's $n/r+o(n)$ bit data structure can be built using
  $\Oh{n}$ work and $\Oh{\log n}$ span
  allowing queries in expected time $\Oh{r}$.
\end{theorem}

\subsection{Distributed Alias Table Construction}\label{dist}

The shared-memory parallel algorithm described in \cref{se} can also be adapted
to a distributed-memory machine (\eg using PRAM emulation
\cite{ranade91}).  However, this requires information about many items
to be communicated.  Hence, more communication-efficient algorithms are
important for large $n$.  To remedy this problem, we will now view sampling as a
2-level process implementing the owner-computes approach underlying many
distributed algorithms.

Let $E_i$ denote the set of items allocated to PE $i$.  For each PE $i$, we
create a \emph{meta-item} of weight $W_i\Def\sum_{j\in E_i}w_j$, and construct a
\emph{top-level table} for sampling from the meta-items.  Sampling now amounts
to sampling a meta-item and then delegating the task to sample an actual item
from $E_i$ to PE~$i$.  Alias tables for the $E_i$ can be built independently on
each PE. %
We present two approaches to constructing the top-level table:

\begin{theorem}\label{thm:dist_simple}
  Assuming that $\Oh{n/p}$ items are allocated to each PE, we can sample a
  single item in time $\Oh{\Tstart}$ after preprocessing a 2-level alias table,
  which can be done in time $\Oh{n/p}$ plus the following communication overhead:
  \begin{thmenum}
  \item \label{dist:replicated} $\Oh{\Tword p + \Tstart \log p}$ with replicated preprocessing, or
  \item \label{dist:twoalias} $\Oh{\Tstart \log p}$ with a distributed top-level
    data structure based on \emph{2-alias tables} described below.%
    \footnote{In the conference version we describe how to build alias
      tables in time $\Oh{\Tstart\log^2 p}$ using PRAM emulation.}
  \end{thmenum}
\end{theorem}

For \cref{dist:replicated}, we can perform an all-gather operation on the
meta-items and compute an alias table for the meta-items on each PE in a
replicated way.

For \cref{dist:twoalias}, we employ the data redistribution algorithm
of Ref.~\cite[Section~9]{topkTR} to
redistribute the heavy meta-items such that each PE has total weight
$W/p$ from heavy or light meta-items.  This algorithm combines
parallel algorithms for merging and prefix sums and runs in time
$\Oh{\Tstart\log p}$. It has the property that each PE receives at
most two pieces of heavy meta-items.  Thus, after redistribution, each
PE is possibly responsible for one light meta-item and pieces of two
heavy meta-items with weights summing up to $W/p$. We call this a
\emph{2-alias table}. This data structure can be sampled in an analogous fashion to alias-tables -- first sample a bucket (a PE) uniformly at random; there sample one of three possible meta-items.
See Ref.~\cite[Chapter 3.3.3]{lorenzdiss} for a more detailed explanation.

\begin{proof}[Proof (\cref{thm:dist_simple})]
  Constructing the local alias tables takes time
  $\Oh{\max_i\card{E_i}}\subseteq\Oh{n/p}$.

  \cref{dist:replicated} requires all-gathering $p$ items, which can be
  implemented in time $\Oh{\Tword p + \Tstart \log p}$
  \cite[Chapter~13.5.1]{MehSanPar}, resulting in the claimed running time bound.

  For \cref{dist:twoalias}, we need to compute $W/p$ using a
  collective reduction operation. The redistribution algorithm
  \cite{topkTR} then computes a prefix sum and performs a parallel
  merging operation using Batcher's algorithm \cite{Batcher1968sortnet}.
  All this is possible in time $\Oh{\Tstart\log p}$.

  Sampling from the resulting two-level alias table needs an additional
  indirection.  First, a meta-bucket $j$ is drawn uniformly at random.
  Then a request is sent to
  PE $j$ which identifies the subset $E_i$ from which the item should be
  selected and delegates the task of sampling from $E_i$ to PE $i$.\footnote{If
    we ensure that meta-items have similar size (see \cref{ss:redistribute}),
    then we can arrange the meta-items in such a way that $i=j$ most of the
    time.} Computing the meta-bucket from the top-level table of case
  \emph{(ii)} is similar to sampling from a regular alias table, except that the
  buckets are now partitioned into three pieces for items $i$, $A_{i,1}$, and
  $A_{i,2}$, and for $x \Def W/p \cdot \rand{}$ we need to determine whether
  \begin{enumerate*}[label=(\alph*)]
  \item $x \leq W_i$, in which case item $i$ is returned,
  \item $W_i < x \leq W_i+A_{i,1}$, returning the first alias $A_{i,1}$, or
  \item $x > W_i+A_{i,1}$, returning $A_{i,2}$.
  \end{enumerate*}
\end{proof}

\subsubsection{Redistributing Items}\label{ss:redistribute}

As discussed so far, \emph{constructing} distributed-memory 2-level alias tables is communication efficient. However, when large items are predominantly allocated on few PEs, \emph{sampling} many items can lead to an overload on PEs with large $W_i$. We can remedy this problem by moving large items to different PEs or even by splitting them between multiple PEs.
This implies a trade-off between redistribution cost (part of preprocessing) and load balance during sampling.

We now look at the case where an adversary can choose an arbitrarily skewed distribution of item sizes but where the items are randomly allocated to PEs
(or that we actively randomize the allocation implying $\Th{n/p}$ additional communication volume).

\begin{theorem}\label{thm:redistribute}
  If items are randomly distributed over the PEs initially,
  it suffices to redistribute $\Oh{\log p}$ items from each PE
  such that afterwards each PE has total weight $\Oh{W/p}$ in expectation
  and $\Oh{n/p+\log p}$ (pieces of) items.
  This redistribution takes expected time $\Oh{\Tstart\log p}$.
\end{theorem}
\begin{proof}
    Let us distinguish between \emph{super-heavy} items whose weight
    exceeds $cW/(p\log p)$ for an appropriate constant $c$ and the
    remaining \emph{ordinary} items.  The expected maximum weight
    allocated to a PE based on ordinary items is $\Oh{W/p}$
    \cite{San96ahab}.  Thus, we need only redistribute super-heavy
    items.  Yet, by definition of super-heavy items, at most
    $p\log(p)/c$ of them can exist in the entire input.  By standard
    balls-into-bins arguments \cite[Theorem 1, Case 2]{raab1998balls},
    only $\Oh{\log p}$ super-heavy items are initially allocated to
    any PE with high probability.

    As in \cref{thm:dist_simple}~(ii),
    we employ the redistribution algorithm from Ref.~\cite[Section~9]{topkTR}:
    PEs distribute their super-heavy items. They can receive pieces of super-heavy items
    of total weight $W/p$ minus the weight of their local ordinary items.
    The difference now is that a PE may send or receive $\Oh{\log p}$ super-heavy items.
    However, this does not change the asymptotic complexity of $\Oh{\Tstart\log p}$.
\end{proof}

\section{Output Sensitive Sampling With Replacement (Problem \PWith)}\label{kwith}
The algorithm of \cref{se} easily generalizes to sampling $k$
items with replacement by simply executing $k$ queries.  Since the
precomputed data structures are static, these queries can be run in
parallel.  We obtain optimal span $\Oh{1}$ and work $\Oh{k}$.

\begin{corollary}\label{cor:kWR}
  After a suitable alias table data structure has been computed, we can sample
  $k$ items with replacement with work $\Oh{k}$ and span $\Oh{1}$.
\end{corollary}

Yet if the weights are skewed this may not be optimal
since large items will be sampled multiple times. Here, we
describe an \emph{output sensitive} algorithm that outputs only
different items in the sample together with how often they were
sampled, \ie a set $S$ of pairs $(i,k_i)$ indicating that item~$i$
was sampled $k_i$ times. The work will be proportional to the output
size $s$ up to a small additive term.
For highly skewed distributions, even $k\gg n$ may make sense.
Note that outputting multiplicities may be important for appropriately
processing the samples.  For example, let $X$ denote a random variable where
item $i$ is sampled with probability $w_i/W$ and suppose we want a truthful
estimator for the expectation of $f(X)$ for some function~$f$ that is assumed to
correlate with the weights.  Then the Hansen-Hurwitz estimator,
$\smash{\hat f(X)} = W/k \cdot \sum_{(i,k_i)\in S}k_if(i)/w_i$, provides the
desired result \cite{hansen1943estimator}.

We will combine and adapt three previously used techniques for related problems:
the bucket tables from \cref{prel:alias}, the
divide-and-conquer technique from \cref{prel:daq} \cite{SandersLHSD18},
and the Poisson sampling algorithm of \citet{bringmann2017subset}.
We approximately sort the items into $u=\ceil{\log U}=\ceil{\log(\wmax/\wmin)}$
groups of items whose weights differ by at most a factor of two using integer
sorting.  To do this, weight $w_i$ is mapped to group $\floor{\log(w_i/\wmin)}$.

To help determine the number of samples to be drawn from each group,
we build a complete binary tree with one \emph{nonempty} group at each leaf.  Interior
nodes store the total weight of items in groups assigned to their
subtrees.  This \emph{divide-and-conquer tree (DC-tree)}
allows us to generalize the divide-and-conquer
technique described in \cref{prel:daq} to weighted items.  Suppose we
want to sample $k$ items from a subtree rooted at an interior node
whose left subtree has total weight $L$ and whose right subtree has
total weight $R$. Then the number of items $k'$ to be sampled from the
left has a binomial distribution ($k$ trials with success probability
$L/(L+R)$). We can generate $k'$ accordingly and then recursively
sample $k'$ items from the left subtree and $k-k'$ items from the
right subtree. A recursive algorithm can thus split the number of
items to be sampled at the root into numbers of items sampled at each
group. When a subtree receives no samples, the recursion can be
stopped. Since the distribution of weights to groups can be highly
skewed, this stopping rule will be important in the analysis.

For each group $G$, we integrate bucket tables and DC-tree as follows.
Let $n_G$ be the number of items in the group, and let $[a, 2a)$ be
the range of weights mapped to this group.  For the bucket table we
can use a very simple variant that stores these $n_G$ items with
weights from the interval $[a,2a)$ in $n_G$ buckets of capacity
  $2a$. Sampling one item then uses a rejection method that repeats
  the sampling attempt when the random variate leads to an empty part
  of a bucket.%
  \footnote{If desired, we can also avoid rejection sampling by
    mapping the items into up to $2n_G$ buckets of size
    $a$ without gaps. This way there are at most two items in each
    bucket. Note that this is still different from alias tables because we need to map ranges
    of consecutive items to ranges of buckets. This is not possible for alias tables.}

\begin{figure}[t]
  \centering
  \includegraphics{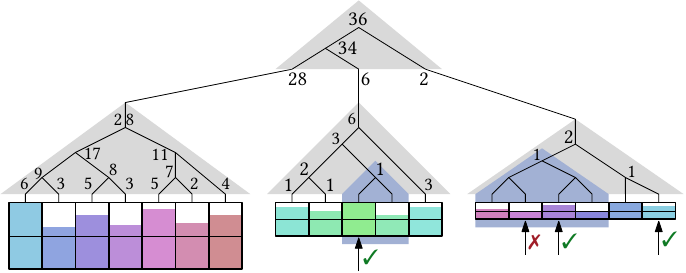}
  \caption{Output-sensitive sampling: assignment of multiplicities with
    $k=36$.}\label{fig:os}
\end{figure}

We also build a DC-tree for each group.  A simple
linear mapping of items into the bucket table allows us to
associate a range of relevant buckets $b_T$ with each subtree $T$ of the
group DC-tree.

For sampling $m$ items from a group $G$,
we use the DC-tree to decide how many samples each subtree has to contribute.
When this number decreases to $1$ for a subtree $T$, we sample this item directly and in constant expected time from the buckets in the range $b_T$.
\footnote{For practical implementations, it may make sense to stop descending
  the tree when $m$ drops below some constant.  We can then join duplicates
  using a small hash table that fits in cache.  This optimization is described
  in more detail on page \pageref{exp:opt}.}

Figure~\ref{fig:os} gives an example.
We obtain the following
complexities:
\begin{theorem}\label{thm:swr}
  Preprocessing for problem \PWith can be done in the time and span
  needed for integer sorting $n$ elements with
  $u=\ceil{\log U}=\ceil{\log(\wmax/\wmin)}$ different keys%
  \footnote{\cref{ss:models}
     discusses the cost of this operation on different models of computation.}
  plus linear work and
  logarithmic span (on a CREW PRAM) for building the actual data
  structure.  Using this data structure, sampling a multiset $S$ with
  $k$ items and $s$ different items can be done with span
  $\Oh{\log n}$ and expected work $\Oh{s+\log n}$ on a CREW PRAM.
\end{theorem}
\begin{proof}
  Besides sorting the items into groups, we have to build binary trees
  of total size $\Oh{n}$.  This can be done with logarithmic span and
  linear work in a way analogous to a reduction operation \cite[Section~13.2]{MehSanPar}.
  The main difference is that the computed intermediate sums of item weights
  are stored in an implicit representation of the tree \cite[Section~6.1]{MehSanPar}.
  The bucket-tables which
  have total size $n$ can be constructed as in \cref{se}.

  The span of a query is essentially the depth of the trees,
  $\log u +\log n\leq 2\log n$.
  Bounding the work for a query is more complicated since, in the
  worst case, the algorithm can traverse paths of logarithmic length
  in the DC-trees for just a constant number of samples taken at its
  leaves. However, this is unlikely to happen, and we show that in
  expectation the overhead is a constant factor plus an additive
  logarithmic term. We are thus allowed to charge a constant amount of
  work to each different item in the output and can afford a leftover
  logarithmic term.

  We first consider the top-most DC-tree
  $T$ that divides samples between groups.  For the analysis
  we distinguish three ranges of groups.
  To the left, there are \emph{heavy} groups whose items are sampled with probability at least $1/2$.
  To the right, there are \emph{light} groups whose items are sampled with probability at most $1/n$.
  We call the groups between heavy and light groups \emph{medium groups}.

  Consider the path $P$ in $T$ leading to the leftmost medium group.  Subtrees
  branching from $P$ to the left are complete subtrees that lead to heavy groups
  only. Since all leaves represent non-empty groups, we can charge the cost for
  traversing the left trees to the items that are actually sampled with probability at least $1/2$.

  Since item weight bounds decrease geometrically from group to group, there can be at most
  $\Oh{\log n}$ medium groups.
  Because the medium groups are consecutive in $T$, they fit into subtrees
  of logarithmic total size.  Therefore, traversing the paths to them
  causes $\Oh{\log n}$ work in total.

  Each item from the light groups yield a sample with probability at most $1/n$.
  There are at most $u\leq n$ such
  groups.  Fix an arbitrary one of them, $G$, and let $n_G$ be the number of
  items in this group.  Then the expected cost of traversing the path to $G$ is
  the path's length ($\leq \log u \leq \log n$) times the likelihood that the
  path is taken ($\leq n_G/n$).  Summing this over all light groups yields total
  expected work $\log n \sum_{G} n_G/n \leq \log n$.

  Finally, Lemma~\ref{lem:uniform} shows that the work for traversing
  DC-trees within a group is linear in the output size from each
  group. Summing this over all groups yields the desired bound.
\end{proof}

\begin{lemma}\label{lem:uniform}
  Consider a DC-tree plus bucket array for sampling with replacement
  of $k$ out of $n$ items where weights are in
  the range $[a,2a)$. Then the expected work for sampling is
    $\Oh{s}$ where $s$ is the number of different sampled items.
\end{lemma}
\begin{proof}
  If $k\geq n$, $\Om{n}$ items are sampled in expectation at a total cost
  of $\Oh{n}$. So assume $k<n$ from now on.  The
  first $\log k+\Oh{1}$  levels of $T$ may be traversed
  completely, contributing a total cost of $\Oh{k}$.
  For the lower
  levels, we count the number $Y$ of visited nodes from which at least
  2 items are sampled.  This is proportional to the total number of
  visited nodes since nodes from which only one item is sampled
  contribute only constant expected cost (for directly sampling from
  the array) and since there are at most $2Y$ such nodes.

  Let $X$ denote the number of items sampled at a
  node at level $\ell$ of tree $T$.  An interior node at level $\ell$
  represents $2^{L-\ell}$ leaves with total weight $W_{\ell}\leq 2a2^{L-\ell}$
  where $L=\ceil{\log n}$. $X$ has a binomial distribution with $k$ trials and
  success probability
  $$
    \rho=\frac{W_{\ell}}{W}\leq \frac{2a2^{L-\ell}}{a2^{L-1}}=4\cdot 2^{-\ell}\punkt
  $$
  Hence, the probability of sampling at least two items at a node at level
  $\ell$ can be estimated as
  \[
    \prob{X\geq 2}=1-\prob{X=0}-\prob{X=1}=1-(1-\rho)^k-k\rho(1-\rho)^{k-1}\approx (k\rho)^2/2
  \]
  where the ``$\approx$'' holds for $k\rho\ll 1$ and
  was obtained by series development in the variable $k\rho$.

  The expected cost at level $\ell>\log k+\Oh{1}$ is thus estimated as the
  product of the above probability and the number of nodes at this level, \ie
  \[
    2^\ell\prob{X\geq 2}\approx 2^\ell(k\rho)^2/2\leq 2^{\ell}(k\cdot 4\cdot2^{-\ell})^2/2=8k^22^{-\ell}\punkt
  \]
  At level $\ell=\ceil{\log k}+3+i$ we thus get expected cost
  $\leq 8k^2 2^{-\ceil{\log k}}\cdot 2^{-3}\cdot 2^{-i} \leq k \cdot 2^{-i}$.
  Summing this over $i$ to cover all lower levels of the tree yields expected
  total cost $Y\in\Oh{k}$.
\end{proof}

\subsection{Distributed Case}\label{kwith:dist}
The batched character of sampling with replacement makes this
setting attractive for a distributed implementation using the
owner-computes approach.  Each PE builds the data structure described
above for its local items.  Furthermore, we build a top-level DC-tree
that distributes the samples between the PEs, \ie with one leaf for
each PE.  We will see below that this can be done using a bottom-up
reduction over the total item weights on each PE, \ie no PRAM
emulation or replication is needed.  Each PE only needs to store the
partial sums appearing on the path in the reduction tree leading to
its leaf. Sampling itself can then proceed without communication --
each PE simply descends its path in the top-level DC-tree analogous to the
uniform case \cite[Section~3.2]{SandersLHSD18}.  Afterwards, each PE knows how many
samples to draw from its local items. Note that we assume $k$ to be
known on all PEs and that communication for computing results from the
sample is not considered here.

\begin{theorem}\label{thm:pwithD}
  Sampling $k$ out of $n$ items with replacement (problem \PWith) can
  be done in a communication-free way with processing overhead
  $\Oh{\log p}$ in addition to the time needed for taking the local
  sample.  Building and distributing the DC-tree for distributing the
  samples is possible in time $\Oh{\Tstart\log p}$.
\end{theorem}
\begin{proof}
  It remains to explain how the reduction can be done in such a way
  that it can be used as a DC-tree during a query and such that each
  PE knows the path in the reduction tree leading to its leaf.  First
  assume that $p=2^d$ is a power of two. Then we can use the well-known
  hypercube algorithm for all-reduce (\eg \cite{KumEtAl94}).
  In iteration $i\in 1..d$ of this algorithm, a PE knows the sum for
  its local $i-1$ dimensional subcube and receives the sum for the
  neighboring subcube along dimension $i$ to compute the sum for its
  local $i$ dimensional subcube. For building the DC-tree, each PE
  simply records all these values.
  For general (\ie non-power-of-two) values of $p$, we first build the DC tree for
  $d=\floor{\log p}$.  Then, each PE $i$ with $i<2^d$ and $j\Def i+2^d<p$ receives the
  aggregate local item weight from PE $j$ and then sends its path to PE $j$.
\end{proof}

Observe that the way \cref{thm:pwithD} is stated, this sample distribution
algorithm can also be applied to the naive local sampling method of repeatedly
sampling one item from an alias table.

Similar to \cref{dist}, it depends on the assignment of the
items to the PEs whether this approach is load balanced for the local
computations. Before, we needed a balanced distribution of both the number
of items and the items' weights. Now the situation is better
because items may be sampled multiple times but
require work only once. On the other hand, we do not want to split
heavy items between multiple PEs since this would increase the amount
of work needed to process the sample. It would also undermine the idea
of communication-free sampling if we had to collect samples of the
same item assigned to different PEs.  Below, we once more analyze the
situation for items with arbitrary weight that are allocated to the
PEs randomly.

\begin{theorem}\label{thm:pwithDR}
  Consider an arbitrary set of item sizes and let $u=\log(\max_iw_i/\min_iw_i)$.
  If items are randomly
  assigned to the PEs initially, then preprocessing takes expected
  time $\Oh{\isort{u}{n/p}{1}+\Tstart\log p}$ where
  $\isort{u}{x}{1}$ denotes the time for sequential integer
  sorting of $x$ elements using keys from the range $0..u$.%
  \footnote{Note that this will be linear in all practically relevant situations.}
  Using this
  data structure, sampling a multiset $S$ with $k$ items and $s$
  different items can be done in expected time $\Oh{s/p+\log p}$.
\end{theorem}
\begin{proof}
  For preprocessing, standard balls-into-bins arguments tell us that a random
  assignment of $n\geq p\log p$ items to $p$ PEs results in no more than
  $n/p + \Thsmall{\sqrt{n/p \cdot \log p}} = \Th{n/p}$ items at any PE with high
  probability (see \eg Ref.~\cite{raab1998balls}), and $\Oh{\log p}$ items if
  $n<p\log p$.  Thus, we have at most $\Oh{n/p + \log p}$ items at any PE with
  high probability.

  Since sorting is now a local operation, we only need an efficient sequential
  algorithm for approximately sorting integers. The term $\Tstart\log p$ is for
  the global DC-tree as in Theorem~\ref{thm:pwithD}.

  A sampling operation will sample $s$ items. Since their allocation is
  independent of the choice of the sampled items, we can apply the same
  balls-into-bins bounds as above to conclude that no more than
  $\Oh{s/p+\log p}$ sample items are allocated to any PE with high probability.
\end{proof}

\section[Sampling Without Replacement (Problem \PWithout)]{Sampling Without Replacement (Problem \PWithout)}\label{kno}

A common approach to sampling without replacement is to sample \emph{with}
replacement and reject items that were already sampled.  Using a hash table to
quickly identify duplicate samples works well for small sample sizes in the
uniform (\ie unweighted) sampling setting \cite{SandersLHSD18}.  It would,
however, not work well for weighted inputs with skewed distribution, where few
items would dominate the sample with replacement, leading to an extremely high
number of rejected items.  But presume that we knew an
$\ell > k$ so that a sample of size $\ell$ \emph{with replacement} contains at
least $k$ and no more than $\Oh{k}$ unique items with sufficiently high
probability.  Then we could use the output-sensitive algorithm for \PWith of
\cref{kwith} to obtain a sample $S$ of size $\geq k$
and subsequently draw an appropriate subsample of size $k$ from $S$ using
an algorithm that has work linear in $\card{S}$.

A crucial step is to find $\ell$.  This depends heavily on the
distribution of the input: if all items have similar weight, $\ell$ will be
little larger than $k$, but if the weights are heavily skewed, it could have to
be much larger.  Our main observation is that by having distributed the
items into groups of similar weight (see \cref{kwith}), we already have enough
information to compute a good value for $\ell$ in logarithmic time.  The basis
of this estimation is to assume that sufficiently heavy items are sampled once
and lighter items are sampled with probability proportional to their weight. We
precompute the data needed for the estimation for each group and then, at query
time, perform a binary search over the groups.  Suppose the currently considered
group stores items with weights in the range $[a, 2a)$. Then we try the value
$\ell= \ceil{W/a}$. We (over-)estimate the resulting number of unique samples as
  $\card{\setGilt{i}{w_i\geq a}}+\ell\cdot\sum_{i:w_i<a}w_i/W$.

This is the precomputed \emph{number} of items with weight $\geq a$ (those in heavier
groups, including the current group) plus $\ell$ times the relative
\emph{weight} of the items with weight $<a$.  Below we exploit that aiming for an
overestimate of $2k$ different samples will actually yield at least $k$ with high probability. First, we state the above estimation as a lemma -- which we shall
prove later -- providing both an over- and an underestimate.

\begin{lemma}%
  \label{lem:nr:exp}
  Let $X$ be the random variable describing the number of unique items in a
  weighted sample with replacement of size $\ell$.  Then,
  \[\left(1-\frac1e\right)t_{\ell} \leq \expect{X} \leq t_{\ell}
    \kern.5em\text{where}\kern.5em
  t_{\ell} \Def \ell \cdot \sum_{i:\, \frac{w_i}{W} < \frac 1\ell} {\frac{w_i}{W}}
  + \card{\set{i \mathrel{\Big|} \frac{w_i}{W} \geq \frac1\ell}}\punkt \]
\end{lemma}
Before we dive into the technical parts, we now describe the algorithm in more
detail.

\paragraph*{Preprocessing}
To support fast queries, we need to do some additional precomputation during
construction.  First, compute $W$ and apply the preprocessing of
\cref{kwith}. Then, compute each group's relative total weight,
$g_i \Def \sum_{j \in G_i}{w_j/W}$, where $G_i$ denotes the group with index
$i$, and their exclusive prefix sum, $h_i \Def \sum_{j < i}g_i$.  Further, let
$c_i \Def \sum_{j < i}{\card{G_j}}$ be the number of items in all groups before
group $i$.

\paragraph*{Finding a Sample $S$ with $\card{S}\geq k$}
At query time, we can find a suitable value of $\ell$ using binary search over
the non-empty groups.  These are exactly the leaves of the top-level DC-tree
constructed by the preprocessing for \PWith.  Let $i$ be the index of the group
currently under consideration and $[a_i, 2a_i)$ be the interval of probabilities
assigned to the group.  Lemma~\ref{lem:nr:exp} gives upper and lower bounds for
the expected number of unique items in a sample of size $\ell$, which are only a
constant factor apart.  With the data structures computed during construction,
we can evaluate $t_\ell$ for $\ell \Def \ceil{W/a_i}$ as
$t_{\ell} \Is \ell \cdot h_i + (n - c_i)$.  Lemma~\ref{lem:nr:group} states that
we can find a group $i$ whose minimum weight $a_i$ gives us a value of $\ell$
such that sampling $\ell$ items with replacement yields $2k$ unique items
in expectation, but at the same time not too many more (\ie $\Oh{k}$).  This
ensures that the output contains at least $k$ unique items with sufficient
probability.  However, this group may be empty and thus not considered in the
binary search.  Let $G_i$ and $G_j$ be the non-empty groups with items of
next-smaller and next-larger weights, respectively.  Then we are free to choose
$\ell$ in the range $2a_i..a_j$.  We do this by solving the inequality
\( \smash{\expect{X} \geq (1-1/e)t_\ell \stackrel{!}{=} 2k} \) of
\cref{lem:nr:exp} for $\ell$.  As there are no items with weights in the range
considered, the sum of weights and the set cardinality in the equation of the lemma
remain constant, and we can solve for $\ell$ in constant time to find the most
suitable value of~$t_\ell$, obtaining
\( \ell = h_i^{-1} \left( 2k\frac{e}{e-1} - (n - c_i) \right) \).
We then draw a sample $S$ of size $\ell$ \emph{with} replacement.
If $S$ contains fewer than $k$ different items, we discard it and
repeat the procedure until the number of different sampled items
is $\geq k$.

\paragraph*{Subsampling}
A difficulty with using the output sensitive algorithm for sampling with replacement is that we obtain no information about the order in which items are sampled while the definition of sampling without replacement uses this order to decide which items are included in the sample. However, we can exploit that every order is equally likely in the following sense: Consider $S$ as the sequence of $\ell$ sampled items in an arbitrary order. Any permutation of $S$ is equally likely to
be the result of an equivalent sequential sampling with replacement with the same overall result.
The probability of any sampled item $e$ to be the first in a randomly permuted  sequence is proportional to its multiplicity in $S$. After removing $e$, the same is true for the remaining items, and so on. Thus, we can obtain a subsample with the right probabilities using \PWithout on $S$ with the multiplicities in $S$ as the new weights. Since this new problem has size $\Oh{k}$ in expectation,
we can use the exponential clocks method from \cref{prel:exp}:
We generate appropriate random variates for each of $\card{S}$ items and perform  a parallel selection for the items that receive rank $\leq k$.
We get:
\begin{theorem}\label{thm:without}
  Preprocessing for problem \PWithout is possible with the same work and time
  bounds as for problem \PWith (Theorem~\ref{thm:swr}).  Sampling $k$ items
  without replacement is then possible with expected span $\Oh{\log n}$ and expected work
  $\Oh{k}$ on a CREW PRAM.
\end{theorem}
\begin{proof}%
  The computation of $W$ as well as the $g_i$, $h_i$, and $c_i$ is possible with
  linear work and logarithmic span.  Thus, the running time of the construction
  phase is dominated by the preprocessing for \PWith of Theorem~\ref{thm:swr}.
  The binary search for a query takes time $\Oh{\log n}$, as each step of the
  search requires constant time using the values computed during the
  construction phase.

  For the remaining cost of a query, note that $\card{S}$ can be
  written as the sum of independent indicator random variables $X_i$
  where $X_i=1$ if item $i$ is sampled. Thus, using Chernoff bounds,
  we can conclude that $\card{S}$ is sharply concentrated around its
  expectation.%
  \footnote{More concretely, one useful form of Chernoff bounds states
    that the probability of $\card{S}$ to deviate by more than a factor
    of~$c$ from its expectation $s=\expect{|S|}$ is at most
    $e^{-c^2s/3}$. Hence, for $s=\Om{\log n}$, the probability of
    deviations by more than a constant factor is polynomially small with
    a power of the polynomial we are free to choose.}  We thus calculate
  the complexities based on expected values since large deviations are
  unlikely. In particular, the case that too few samples are drawn,
  which requires another iteration of the query, is sufficiently rare
  to ignore its effect on overall complexity.

  Downsampling the
  resulting sample needs $\Oh{k}$ work and $\Oh{\log k}$ span for a selection of
  $k$ out of $\Th{k}$ items \cite[Section 4]{HubSan16topk}.  In summary, we
  obtain the claimed bound.
\end{proof}

\begin{figure}
  \centering
  \includegraphics{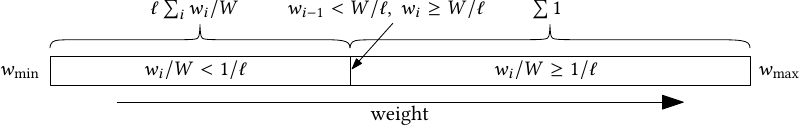}
  \caption{How the expected number of samples $\expect{X}$ is influenced by the
    choice of $\ell$ depends on the distribution of the inputs (\cref{lem:nr:exp}).  Items with
    small probability of being sampled contribute linearly in their weight,
    whereas large items are expected to be sampled.}\label{fig:lsel}
\end{figure}

\begin{proof}[Proof (\cref{lem:nr:exp})]
  We can bound the number of unique items in a sample with replacement of size
  $\ell$ by considering the probability of any item to \emph{not} be sampled.
  Item $i$ is \emph{not} part of the sample with probability $(1-w_i/W)^\ell$,
  and we obtain $\expect{X} = \sum_{i=1}^{n}{1-(1-w_i/W)^\ell}$ where $X$ is the
  random variable describing the number of unique items in the sample.

  We split the formula for the expectation of $X$ into two parts: the items for
  which $w_i/W<1/\ell$, and those with $w_i/W \geq 1/\ell$.  For the first
  group, by Bernoulli's inequality, we obtain
  $1-(1-w_i/W)^\ell \leq \ell w_i/W$.  We find the following lower bound for the
  inclusion probability of items in the first group:
  \begin{align*}
    1-(1-w_i/W)^\ell & =1- (1-b_i/n)^{c\cdot n}            &  & \text{where } b_i \Def nw_i/W \text{ and } c \Def \ell/n \\
                     & \geq 1-e^{-c b_i} = 1-e^{-\ell w_i/W} &  & \text{by } e^x \geq (1+x/n)^n                            \\
                     & \geq \left(1-1/e\right)\ell w_i/W   &  & \text{because } e^{-x}\leq 1 - (1-1/e)x \text { for } x \in (0,1)
  \end{align*}
  and thus the first term of $t$ is explained.  Consider now the items with
  $w_i/W \geq 1/\ell$, \ie the items which yield at least one sample in
  expectation.  Clearly, $1-(1-w_i/W)^\ell \leq 1$.  In the other direction,
  because $w_i/W \geq 1/\ell$, we obtain
  $1-(1-w_i/W)^\ell \geq 1-(1-1/\ell)^\ell \geq 1-1/e$ by the inequality
  $(1+x/n)^n \leq e^x$ for $\ell>1$.  Otherwise, if $\ell\leq 1$, either no item
  with $w_i/W\geq 1/l=1$ exists, or we are in the trivial case with $n=1$ and
  $w_1=W$.  Summing the results for the first and second group, we obtain the
  claimed inequalities.
\end{proof}

An example of this is illustrated in \cref{fig:lsel}.
By applying the above estimation to the groups used by the algorithm of
\cref{kwith}, we can quickly obtain an estimate for the output size that
is at most a factor of two worse.

\begin{lemma}\label{lem:nr:group}
  Applying the estimation of Lemma~\ref{lem:nr:exp} to groups of items of
  similar weight as in \cref{kwith} loosens the bound on $\expect{X}$ by
  at most a factor of two.
\end{lemma}
\begin{proof}
  Item $i$ is in group $\floor{\log (w_i/\wmin)}$.  Label the groups $G_1..G_u$
  where $u \Def \ceil{\log U} = \ceil{\log (\wmax/\wmin)}$ and let the value
  range of group $i$ be $[a_i, 2a_i)$.  The difference to Lemma~\ref{lem:nr:exp}
  now is that we can only choose $\ell$ as $\ceil{W/a_i}$ for some group $i$.
  This necessitates moving entire groups between the left and right parts of the
  estimation (with left and right as illustrated in \cref{fig:lsel}).  Pick a
  group $i$, \ie $\ell \Def \ceil{W/a_i}$, and consider the effects of choosing
  group $i+1$ instead, \ie
  $\ell' \Def \ceil{W/a_{i+1}} = \ceil{W/(2a_i)} = \ell/2$.  This results in
  moving group $i$ from the right part of the estimation to the left part.
  The contribution of any group $j>i$ is the same for $t_\ell$ and
  $t_{\ell'}$ because their items are expected to be sampled in both instances.
  Furthermore, the contribution of groups $1..i-1$ is halved, as $\ell'=\ell/2$.
  It remains to bound the contribution of group $i$, which is moved from the
  right part of the estimation to the left part: it contributes $\card{G_i}$ to
  $t_\ell$, and $\xi \Def \ell' \sum_{j \in G_i}{w_j/W}$ to~$t_{\ell'}$.
  However, as all items in $G_i$ are in the weight range $[a_i, 2a_i)$ and
  $\ell' = \ceil{W/a_{i+1}} = \ceil{W/(2a_i)}$, we have
  $\xi \geq \ell' a_i \card{G_i}/W \geq a_i/a_{i+1} \cdot \card{G_i} =
  \card{G_i}/2$, and thus the expected number of samples contributed by group
  $i$ at most halves, too.  Thus, $t_{\ell} \leq 2 t_{\ell'}$, and for any
  $\ell$ of Lemma~\ref{lem:nr:exp} we can find an $\hat \ell \geq \ell$ that
  yields no fewer, but also at most twice as many unique items in expectation
  while only relying on information about the groups, not the $w_i$.
\end{proof}

\subsection{Distributed Sampling without Replacement}\label{ss:nr_dist}

We again use the owner-computes approach and adapt the distributed
data structure for problem \PWith from
\cref{kwith:dist}. However, to find the right estimate for the
number $\ell$ of samples with replacement, we need to perform a global
estimation taking all items into account. This can be achieved by
finding the global sum of all the local prefix sums in each step of
the binary search. This increases the latency of the algorithm to
$\Oh{\Tstart\log(p)\log(n)}$. Also, in each iteration we need a nested
binary search to find corresponding buckets in the local
bucket arrays.

If the global number of groups $u$ is not too large, we can
consider a different trade-off between preprocessing cost and query
cost. We precompute replicated arrays of global group sizes and weights. This
allows finding the right value of $\ell$ without communication at query time.
Local sorting is then performed using
bucket sort with $u$ groups in time $\Oh{n/p +u}$.
The global group sizes can then be computed as a sum all-reduction on the
local arrays of group sizes in time $\Oh{\Tword u +\Tstart\log p}$.

During a query, after sampling with replacement, we need a global parallel
selection algorithm to reduce the sample size to $k$. This can be done in
expected time $\Oh{\Tword k/p+\Tstart\log kp}$ using the unsorted selection
algorithm from Ref.~\cite[Section~4]{HubSan16topk}.  The term $\Tword k/p$ stems from
the need to redistribute samples taken within the selection algorithm. For
randomly distributed data this is not necessary and the term becomes a local
work term $k/p$.  We get the following result:

\begin{theorem}\label{thm:pwithoutDR}
  Consider an arbitrary set of item sizes and let $u=\log(\max_iw_i/\min_iw_i)$.
  If items are randomly assigned to the PEs initially, then preprocessing takes
  expected time $\Oh{\isort{u}{n/p}{1}+\Tstart\log p}$ where $\isort{u}{n/p}{1}$
  denotes the time for sequential integer sorting of $n/p$ elements using keys
  from the range $0..u$.  Using this data structure, sampling $k$ items without
  replacement can be done in expected time $\Oh{k/p+\Tstart\log(n)\log(p)}$.

  A variant that uses a replicated array of global group sizes works with
  $\Oh{n/p + \Tword u+\Tstart\log p}$ preprocessing time and query time
  $\Oh{k/p + \log u + \Tstart\log kp}$.
\end{theorem}

A proof is provided in Ref.~\cite[Chapter 3.5.1]{lorenzdiss}.

\section{Permutation (Problem \PPerm)}\label{perm}

As already explained in \cref{prel:exp}, weighted permutation
can be reduced to sorting random variates of the form $-\ln(R)/w_i$,
where $R$ is a uniform random variate.  This is helpful because a lot is known
about parallel sorting.  The downside is that sorting may need
superlinear work in the worst case.  However, since we are sorting
\emph{random} numbers, we may still get linear expected work.
This is well-known when sorting uniform random variates; \eg
Ref.~\cite[Theorem~5.9]{MehSanPar}.  The idea is to do a linear-time mapping of the
random variates to a small number of buckets such that the occupancy of a
bucket has a binomial distribution with constant
expectation.  Then the buckets can be sorted using a comparison-based
algorithm without getting more than linear work in total.
Here we explain how to achieve the same
for the highly skewed distribution needed for \PPerm.

We use the monotonous transformation of the above geometric deviate based key mapping function to
$f(R,w_i)\Is n\ln(-\ln(R)n\wmax/w_i)$,
where $\wmax=\max_jw_j$.
Except for an expected constant number of items, the random number $R$ will be in the range
$[1/n,1-1/n]$. This corresponds to a range of
\[\Big[n\ln-\frac{\ln(1-1/n)n\wmax}{\wmax}  ,n\ln-\frac{\ln(1/n)n\wmax}{\wmin}\Big]\approx
\left[0, n\ln(nU\ln n)\right]
\]
for $f$ where
$\wmin=\min_jw_j$ and $U=\wmax/\wmin$.
Values outside this range will be mapped to keys $-1$ and  $\ceil{n\ln(nU\ln n)}$. The remaining items will be mapped to the integer key $\floor{f(R,w_i)}$.
We then perform integer sorting on the truncated keys and apply comparison-based sorting to the resulting buckets of items with the same integer keys.

\begin{theorem}\label{thm:perm}
  Problem \PPerm can be solved in the work and span needed for integer sorting of $n$ items
  with keys from the range $0..\Oh{n\ln(nU)}$.
\end{theorem}
\begin{proof}
  The main proof obligation is to show that buckets have constant
  expected occupancy regardless of the weight distribution. First note
  that $f$ can be written as $n(\ln(-\ln R)+\ln(n\wmax/w_i))$. The
  factor $n$ scales the bucket size.  The two terms multiplied
  with $n$ separate the influence of the uniform variate $R$ and of the
  weights.  This means that the weight term just shifts the
  distribution produced by the term $\ln(-\ln(R))$, \ie the overall
  distribution is a mix of shifted versions of the same distribution
  -- one for each weight value. This means that the maximum bucket occupancy
  is maximized if all weights are the same. Let us ignore the scaling
  factor $n$ and the shift for now and concentrate on the mapping
  $g(R)=\ln(-\ln R)$. The value of $R$ needed to produce a value $x$
  is $e^{-e^x}$.  The value of $R$ needed to produce a value
  $x+\delta$ is $e^{-e^{x+\delta}}$. Hence, the probability to produce
  a value in $[x,x+\delta]$ is $e^{-e^x}-e^{-e^{x+\delta}}$. In other
  words, the probability density is maximized at $x$ when the
  derivative of $e^{-e^x}$ is minimized. Using calculus, it can be
  shown that this is the case at $x=0$.  Hence, the probability that
  an item is mapped to a bucket is bounded by the width of the bucket
  times $e^{-e^0}=1/e$. Taking the scaling factor $n$ into account
  now, we see that the bucket width is $1/n$, yielding a probability
  of $\leq 1/(ne)$ for any item to be mapped to any particular bucket.
  Hence, the occupancy of any bucket is bounded by a binomial
  distribution with success probability $\leq 1/(ne)$ and $n$ trials,
  \ie with an expected occupancy of $1/e$.
  The proof of Ref. \cite{MehSanPar} for the uniform case transfers for
  the cost of sorting the buckets with overall work $\Oh{n}$.  The
  span for that part is $\Oh{\log n}$ with high probability. This can
  easily be shown using Chernoff bounds.
\end{proof}

Depending on $U$ and the machine model, this result can
yield very efficient algorithms.  If span $n^{\epsilon}$ for
some constant $\epsilon$ is acceptable, we get linear work if $\log U$
is polynomial in $n$ using radix sort.  If $U$ itself is polynomial in $n$ we get
logarithmic span on a CRCW PRAM \cite{raja1989sort}.

\section{Poisson Sampling (Problem \PSubset)}\label{subset}

Poisson sampling, also known as subset sampling, is a generalization of Bernoulli sampling to the weighted case.
The unweighted case can be solved
in  linear expected time with regard to the output size by computing the geometrically distributed
distances between items in the sample \cite{AhrDie85}.  The na\"ive algorithm for the weighted problem, which consists of throwing a biased coin
for each item, requires $\Oh{n}$ time.  \citet{bringmann2017subset} show that
this is optimal if only a single subset is desired, and present a sequential
algorithm that is also optimal for multiple queries.

The difference between \PSubset on the one hand and \POne/\PWith on the other hand is that
we do not have a fixed sample size but rather treat the item weights as
independent inclusion probabilities in the sample (this requires $w_i\leq 1$ for
all $i$). Hence, different algorithms are required.
Observe that the expected sample size is $W \leq n$.  Then our goal is to devise
a parallel preprocessing algorithm with work $\Oh{n}$ which subsequently permits
sampling with work $\Oh{1+W}$.

We now parallelize the approach of Ref.~\cite{bringmann2017subset}.
Similar to our algorithm for sampling with replacement, this algorithm
is based on
grouping items into sets with similar weight. In each group, one can
use ordinary Bernoulli sampling in connection with rejection
sampling. Load balanced division between PEs can be done with a prefix
sum calculation over the weights in each group.

\begin{theorem}\label{thm:subset}
  Preprocessing for problem \PSubset can be done with work $\Oh{n}$
  and span $\Oh{\log n}$.  A query can then be implemented with expected
  $\Oh{W+1}$ \mbox{work and $\Oh{\log n}$ span}.
\end{theorem}
\begin{proof}
  Approximately sort the items into $L+1$ buckets with $L \Def \ceil{\log n}$,
  where bucket $i$ is
  $B_i \Def \setsmall{j \mid 2^{-i} \geq w_j \geq 2^{-(i+1)}}$ for $i \in 0..L-1$ and
  $B_L \Def \setsmall{j \mid 2^{-L} \geq w_j}$ contains all sufficiently improbable
  items.  This can be done with linear work and logarithmic span on a CREW
  PRAM \cite[Lemma~3.1]{raja1989sort}.  Let $\sigma$ denote a permutation of
  the items implied by this re-ordering.

  Next, precompute an assignment of consecutive ranges of permuted items to PEs.
  Bucket boundaries can be ignored since there are no dependencies between items.
  Thus, we compute a prefix sum
  $\hat{w}_i \Def \sum_{j\leq i}{w_{\sigma(j)}}$ over the inclusion
  probabilities in their new order.  PE~$i$ then handles the items whose
  $\hat w_j$ fall into the range $[(i-1)W/p,\ i\cdot W/p)$.  These can be found
  with linear work and constant span by checking for every neighboring pair
  of items whether a boundary falls between them, and if so, which PEs'.

  Sampling then proceeds on all buckets in parallel using the
  assignment calculated during preprocessing.  In bucket $i$, all
  items have weight at most
  $\overline{w}_i\Def \max_{j \in B_i}{w_j}\leq 2^{-i}$,
  and, with the exception of bucket $L$, at least
  $\overline{w}_i/2$.  PEs generate geometrically distributed skip
  values $v \Def \floor{\ln(\rand)/\ln(1-\overline{w}_i)}$
  and then consider the $j \Def v / \overline{w}_i$-th item.  The algorithm then
  uses rejection to output the item with probability
  $w_j / \overline{w}_i$.  This process is repeated
  until a PE exceeds its allotted item range.
  In all buckets $i<L$, the acceptance probability is at least
  $1/2$ for every item, leading to an efficient algorithm in expectation.
  In bucket $L$, the smaller acceptance probability is
  not a problem, as with high probability only a constant number of items is ever considered
  to begin with.  This is because the probability of being sampled is at most
  $1/n$ for items in bucket $L$. In total, this requires
  work $\Oh{1+W}$ in expectation.

  Although all PEs have to do about the same expected amount of work,
  there are some random fluctuations between the actual amount of work
  depending on the actual exponential variates that are computed.
  Using Chernoff bounds once more, these can be bounded to
  $\Oh{W/p+\log n}$ with high probability, leading to a (conservative)
  $\log n$ term for the span.
  Each PE returns
  an array containing its part of the sample.
  An additional
  prefix sum computation can be used to rearrange the output into one contiguous array.
  This requires work $\Oh{n}$ and span $\Oh{\log n}$.
\end{proof}

\subsection{Distributed Poisson Sampling}\label{subset:dist}

For the distributed setting, observe that problem \PSubset can be
solved entirely independently over disjoint subsets of the input if we
are not interested in load balancing. No communication is needed.
We thus can directly adapt the result of Ref.~\cite{bringmann2017subset}.
Note that for a PE with $n_i$ items, it suffices to sort them into $\log n_i$ categories now which is possible in linear time.
A query on a PE with total weight $W_i$ will take expected time $\Oh{1+W_i}$. As expected parallel query time we get $\Oh{\max_iW_i+\log n}$
using an argument analogous to the proof of \cref{thm:subset}.

Once more, we analyze load
balancing for the case that the items are distributed randomly.

\begin{theorem}\label{thm:subsetD}
  When items are distributed randomly,
  Poisson sampling (problem \PSubset) can
  be done in a communication-free way with expected preprocessing overhead
  $\Oh{n/p+\log p}$ and expected sampling time $\Oh{W/p+\log p}$.
\end{theorem}
\begin{proof}
  The bound for the preprocessing follows by applying balls-into-bins bounds to
  the distribution of the number of items as in the proof of \cref{thm:pwithDR}.
  The query bound follows similarly by exploiting that the expected
  maximum load is maximized when all the weight is concentrated in $W$ items of
  weight 1 \cite{San96ahab}.
\end{proof}

\section{Sampling with a Reservoir}\label{res}

We adapt the streaming algorithm of \citet{efraimidis2006reservoir} to a
distributed mini-batch streaming model (also known as \emph{discretized
  streams}), where PEs process variable-size mini-batches of items one at a
time.  After processing a mini-batch, the task is to update the sample to be a
uniform (or weighted) random sample without replacement of size $\min(k, n')$ of
all $n'$ items seen so far, up to and including the items of the current batch.
The PEs' memory is too small to store previous mini-batches, only the current
mini-batch is available in memory.  This is a generalization of the traditional
data stream model and widely used in practice, \eg in Apache Spark Streaming
\cite{sparkstream2013}, where it is called \emph{discretized streams}.  The
basic idea is to keep the reservoir, \ie the data structure used to maintain the
sample, in a distributed priority queue \cite{HubSan16topk}.

\subsection{Mini-Batch Model}
Our model is a batched view on streaming algorithms.  Items arrive as a series
of \emph{mini-batches} (or \emph{batches} for short) on small time
intervals.  For example, each mini-batch
could be defined as the set of all items that arrived within a certain time
window (\eg in Apache Spark Streaming~\cite{sparkstream2013}).  Because memory
is limited, only the current mini-batch is available in memory at each point in
time.  This is a generalization of other models of streaming algorithms: on a
sequential underlying machine with batch size 1, we obtain the sequential
streaming model (see, \eg \cite{TT19parstream}).  In a distributed model with
$p$ sites (nodes) which exchange fixed-size messages with a coordinator, batch
size~1 yields the distributed streaming model, also known as the
\emph{continuous distributed monitoring} model~\cite{cormode2013monitoring}.  In
this paper, we use the distributed message-passing model described in
\cref{ss:models}.

Unless explicitly specified, we make no assumptions about the distribution
of mini-batch sizes or items across PEs or over time.
In the algorithm analysis, we denote by $b$ the maximum number of items in
the current batch at any PE, and by $B$ the sum of all PE's current batch sizes.
Thus, an algorithm for the mini-batch model can correctly handle arbitrarily
imbalanced inputs; however, load balance may
suffer if the number of items per PE differs widely.

\subsection{Weighted Reservoir Sampling}\label{res:alg}

The basic idea of our algorithm is to combine the exponential clocks method,
associating with each item an exponentially distributed key (\cf
\cref{prel:exp}), with a communication-efficient bulk priority queue to maintain
the set of the $k$ items with the smallest keys, \ie the sample.  Each PE is
solely responsible for the items that were seen in its input, and no PE gets a
special role (such as a coordinator node).  During each batch, a PE inserts into
its local reservoir all items whose key is smaller than the largest key of any
item in the sample (the global \emph{threshold}).  When a batch finishes, the
PEs perform a distributed selection for the $k$-th smallest key, which becomes
the new global threshold, and discard all items with larger keys.  The remaining
items form the new sample.  During a mini-batch, the threshold remains
unchanged.

First, we show how to adapt sampling with skip values (exponential jumps)
\cite[Section 4]{efraimidis2006reservoir} to the exponential variates described
in \cref{prel:exp}.  This allows for faster and numerically more efficient
generation in practice. The difference
is a simple $x \mapsto -\ln(x)$
mapping.  Because of the sign inversion, the reservoir $R$ now contains the
items with the \emph{smallest} associated keys.  Let $v_i$ denote the key of
item $i$, that is, the exponentially distributed variate associated with it, and
define $T \Def \max_{i \in R}v_i$ as the threshold value, \ie the largest key of
any item in the reservoir.  Then the amount of weight to be skipped is
exponentially distributed with rate $T$, which can be computed as
$X \Def -\ln(\rand)/T$, where $\rand$ is uniformly random in $(0,1]$.  The key
associated with the newly sampled item $j$ is then
$v_j \Def -\ln(\rand[e^{-Tw_j},1])/w_j$, with $\rand[a,b]\Def a+\rand(b-a)$.
The range of this variate has been chosen so that $v_j$ is less than $T$,
as it has already been determined that item $j$ must be part of the
reservoir.  In the sequential setting, item $j$ then replaces
the item with the largest key in the reservoir, and the threshold $T$ is updated
to the now-largest key.

We now flesh out the remaining parts of our algorithm.  Firstly, the
reservoir is maintained in a distributed fashion.  Each PE's
\emph{local reservoir} is a B+ tree that is augmented to also support
the \emph{split}, \emph{rank}, and \emph{select} operations in
logarithmic time (see, \eg \cite[Sections 7.3.2 and
7.5.2]{MehSanPar}).
Operation \emph{split$(T,x)$} outputs two trees
$T_{\leq}$ and $T_{>}$ such that $T_{\leq}$ contains
$\setGilt{e\in T}{e\leq x}$ and $T_{>}$ contains $T\setminus T_{\leq}$,
function \emph{rank}$(x)$ computes $\card{\setGilt{e\in T}{e\leq x}}$, and
function \emph{select}$(k)$ finds the element with rank $k$ in $T$.

The \emph{split} operation is used to quickly discard the
items that are no longer part of the sample at the end of a batch, while
\emph{rank} and \emph{select} are required for the threshold selection process.

Secondly, the algorithm \emph{globally} maintains the threshold $T$, which is
the same at all PEs and does not change during a mini-batch.  The PEs process
their items using the skip distance method described above, inserting the new
candidate sample items into their \emph{local} reservoirs.

Lastly, once all items of the mini-batch are processed, the PEs jointly select
the globally $k$-th smallest key (see \cref{prel:sel}) in the union of all local
reservoirs.  This key becomes the insertion threshold for the next batch.  Each
PE then discards all items with larger keys using a \emph{split} operation on its local reservoir.  The
remaining items in the union of all local reservoirs are a weighted sample
without replacement of size $k$ of all items seen so far.  Algorithm \ref{alg:w}
gives pseudocode. %

\begin{algorithm2e}[bt]
  \caption{Pseudocode for weighted reservoir sampling.}
  \label{alg:w}
  \KwIn{$A$, the local portion of the mini-batch consisting of items with weight
    $w\in\rplus$ and index $i\in\nat$; $T$, the previous mini-batch's threshold
    (initially $-\infty$); $R$, the local reservoir (initially empty), a B+ tree
    mapping keys from $\real$ to item IDs}
  \KwOut{$T'$, the updated threshold; $R'$, the updated local reservoir}
  \Function{\FuncSty{processBatch}$(A, T, R)$}{
    \If(\Commenti{fewer than $k$ items seen globally before this batch}){$T < 0$}{
      \ForEach{$(w,i) \in A$}{
        $R.\FuncSty{insert}(-\ln(\rand)/w, i)$ \Comment{exponentially distributed keys}
      }
    }\Else{
      $j \Def 0$ \Comment{1-based index of next item, initially invalid}
      \While{$j\leq\card{A}$}{
        $X \Def -\ln(\rand)/T$ \Comment{weight to be skipped}
        \While(\Commenti{skip $X$ amount of weight in total}){$X>0$}{
          $j \Def j+1$\;
          \lIf{$j > \card{A}$}{\Break from both loops}\label{c:break}
          $X \Def X-A[j].w$\;
        }
        $x \Def \exp(-T \cdot A[j].w)$\;
        $v \Def -\ln(\rand[x,1])/A[j].w$ \Comment{new key}
        $R.\FuncSty{insert}(v, A[j].i)$\;
      }
    }

    $(T',i) \Def \FuncSty{select}(R, k) \in \real \times \nat$
    \Comment{select $k$ globally smallest and new threshold}

    $(R', \_) \Def R.\FuncSty{splitAt}(i)$ \Comment{discard local items with larger keys}
    \Return $(T', R')$ \Comment{return new threshold and reservoir}
  }
\end{algorithm2e}

\begin{theorem}\label{thm:res}
  For weighted reservoir sampling with sample size $k$, processing a mini-batch
  of up to $b$ items per PE is possible in time $\Oh{b + (b^*+1)\log(b^*+k) + \Tsel}$,
  where $b^*\leq b$ is the maximum number of items from the mini-batch below the
  insertion threshold on any PE, and $\Tsel$ is the time for selection from
  sorted sequences of size at most $b^*+k$ per PE (see \cref{prel:sel}).
\end{theorem}
\begin{proof}
  By definition of $b^*$, the local insertions require time
  $\Oh{b^*\log(b^*+k)}$ in total because each local reservoir has size at most
  $k$ at the start of the batch.  Since we have to process each item's weight
  even when using skip distances, $\Oh{b}$ time is required to identify the
  items to be inserted into the reservoir.  The selection operation takes time
  $\Tsel$ which varies depending on the specifics of the input.  The number of
  candidate items per PE for the selection is clearly bounded by the local
  reservoir size of at most $k+b^*$.  The split operation to discard the items
  with keys exceeding the new threshold takes time logarithmic therein, \ie
  $\Oh{\log(k+b^*)}$.

  Now consider the implications of splitting the single stream of the sequential
  case into multiple independently-handled streams without carrying over any
  remaining skip weight $X$ at the end of the stream (line \ref{c:break}
  in \cref{alg:w}).  This maintains correctness because the process is designed
  so that each unit of weight has the same probability of spawning a sample,
  regardless of when the procedure was started.  Thus, partitioning the
  stream does not affect correctness.

  Keeping the threshold fixed for the duration of each mini-batch is correct because the set of all items above \emph{any} threshold always
  forms a weighted random sample (whose size is not known a priori).  Here, the
  threshold is chosen so that the sample size is guaranteed to be at least $k$,
  and the selection process determines a new threshold to restore a fixed sample
  size of $k$ items.
\end{proof}

\begin{theorem}\label{thm:load}
  If all items' weights are independently drawn from the same continuous
  distribution and all batches have the same number of items on every PE, then
  our algorithm inserts no more than
  $\Ohsmash{\frac kp \log\frac{n}{k} + \log p}$ items into \emph{any} local
  reservoir in expectation.
\end{theorem}

We first show a lemma on the number of insertions into
\emph{each} PE's local reservoir.

\begin{lemma}\label{lem:load}
  If the item weights are %
  independently drawn from a common
  continuous distribution %
  and all batches have the same number of items on every PE,
  then our algorithm inserts
  $\Ohsmall{\frac kp(1+\log\frac{n}{k})}$ items into each local reservoir in
  expectation.
\end{lemma}
\begin{proof}
  Efraimidis and Spirakis~\cite[Proposition 7]{efraimidis2006reservoir} show that if
  the $w_i$ are independent random variates from a common continuous
  distribution, their sequential reservoir sampling algorithm inserts
  $\Oh{k \log(n/k)}$ items into the reservoir in expectation.  We adapt this to
  mini-batches of $b$ items per PE. Let $X_i$ denote the number of insertions on
  a PE for batch~$i$.  We obtain a binomially distributed random variable with
  expectation
  \[
    \expect{X_i} = \sum_{j=1}^{b}{\prob{\text{item $j$
          is inserted}}} = b\cdot\frac{k}{n_\mathrm{pre}} \leq b\cdot\frac{k}{ipb}=\frac{k}{ip},
  \]
  where $n_\mathrm{pre}$ is the number of items seen globally before the batch
  began.  For the initial $i_0=\frac{k}{bp}$ iterations, this probability
  exceeds one, and we account for this with $b$ insertions per PE, \ie
  $b\cdot\frac{k}{bp}=k/p$ overall.  For mini-batches $i_0 \leq i < \frac{n}{pb}$ we
  obtain
  \begin{align*}
    \expect{\sum X_i} & \leq \! \sum_{\frac{k}{bp}\leq i\leq \frac{n}{bp}}{\frac{k}{ip}}
                        = \frac{k}{p} \! \sum_{\frac{k}{bp}\leq i\leq \frac{n}{bp}}{\frac1i}
                        = \frac{k}{p}\left(H_{\frac{n}{bp}}-H_{\frac{k}{bp}}\right)\\
                      & \leq \frac{k}{p}\left(1+\ln\frac{n}{bp}-\ln\frac{k}{bp}\right)
                        = \frac kp\left(1+\ln\frac{n}{k}\right),
  \end{align*}
  where $H_n$ is the $n$-th harmonic number.
\end{proof}

We can now use this lemma to prove the theorem.

\begin{proof}[Proof (\cref{thm:load})]
To obtain the maximum load over all PEs, we apply a normal approximation to the
bound on the $X_i$ from the proof of Lemma~\ref{lem:load}, obtaining
$Y_i \sim \mathcal{N}\left(\frac{k}{ip},
  \frac{k}{ip}\left(1-\frac{k}{ipb}\right)\right)$.  Summing these over the
mini-batches as above, we again obtain a normal distribution whose mean and
variance are the sum of its summands' means and variances.  We then apply a
bound on the maximum of i.i.d normal random variables obtained using the
Cramér-Chernoff method \cite[Chapter 2.5]{BLM13concineq},
$\expect{\max_{j=1..p}Z_j} \leq \mu + \sigma\sqrt{2\ln p}$ for
$Z_j \sim \mathcal{N}(\mu, \sigma^2)$.  Using the mean of
the $Y_i$ as an upper bound to their variance, we obtain
$\mu + \sqrt{2\mu\ln p}$ as an upper bound to the maximum per-PE load for
$\mu = \frac kp (1+\ln\frac nk)$.  Thus, the expected \emph{bottleneck} number
of insertions into any local reservoir is
$\Ohsmall{\frac kp \log \frac nk + \log p}$.
\end{proof}

\section{Experiments}\label{exp}
We now report experiments on alias tables (problem~\POne, \cref{s:alias}) and
the closely related problem of sampling with replacement
(problem~\PWith, Section~\ref{kwith}) as well as reservoir sampling (problem~\PRes, Section~\ref{res}).
We concentrate on these problems because our algorithms for problems~\PWithout, \PPerm, and \PSubset
are basically reductions to sorting and selection problems that
have been studied elsewhere (\PWithout also needs \PWith, which
we do study here).  The algorithm for problem \PPerm is quite simple and not
that different from the unweighted case studied in \citet{SandersLHSD18}.
We are also not aware of competing parallel approaches that we could compare to.

Our experiments are split into two parts.  First, we consider problems \POne and
\PWith in a shared-memory parallel setting.  Thereafter, we report distributed
experiments on \PRes.

\paragraph*{Experimental Platform}\label{exp:hw}%
We use machines with Intel and AMD processors in our experiments.  The Intel
machine has four Xeon Gold 6138 CPUs ($4 \times 20$ cores, 160 hyper-threads, of
which we use up to 158 to minimize the influence of system tasks on our
measurements) and 768\,GiB of DDR4-2666 main memory.  The AMD machine is a
single-socket system with a 32-core AMD EPYC 7551P CPU (64 hyper-threads, of
which we use up to 62) and 256\,GiB of DDR4-2666 RAM.  While single-socket, this
machine also has non-uniform memory access (NUMA) characteristics, as the CPU
consists of four dies (NUMA nodes) internally, with part of the memory attached to each die.
Both machines run Ubuntu 18.04.  The reservoir sampling experiments were
conducted on up to 256 nodes of ForHLR~II, an HPC cluster with two 10-core Intel
Xeon E5-2660~v3 CPUs per node.  Each core is one PE (20 PEs per node) and
communication uses MPI (OpenMPI 4.0).  All implementations are in C\texttt{++}
and compiled with GNU \texttt{g++} 9.2.0 (flags \texttt{-O3 -flto
  -march=native}).
Our measurements do not include time spent on memory allocation and mapping
where sizes are known in advance.

\paragraph*{Implementation Details}\label{exp:impl}%
We implemented alias table construction using our parallel splitting algorithm
(\psa, Algorithm~\ref{alg:psa}, as well as \psag, the optimization described in
\cref{se:opt}).  For \PWith, we implemented the distributed output-sensitive
sampling algorithm (\osens, \cref{kwith:dist}) in shared memory using a
sequential implementation of the algorithm of \cref{thm:swr} as base method.
This approach should help on
NUMA machines since it has better memory access locality.
\osnd is a variant of \osens described below.
We also
implemented the two-level alias table construction algorithm in shared memory (\twolvl, \cref{dist}, \cref{dist:replicated})
and sequential alias table using both Vose's method (Algorithm~\ref{alg:vose})
and our sweeping construction (Algorithm~\ref{alg:simpleSweep}).  The \twolvl
algorithm can use either of these as its base case.  When using Vose's method,
we refer to it as \twoclassic, and \twosweep when using our sweeping algorithm.
Our parallel
implementations distribute large arrays over the available NUMA nodes
and pin threads to NUMA nodes to maintain data locality.

For reservoir sampling, we implemented our algorithm of \cref{res} using the
approximate multisequence selection algorithm of \cite[Section
IV-C]{HubSan16topk}, using one (labeled \emph{``ours''}) or eight pivots
(labeled \emph{``ours-8''}) and exact bounds
($\kmin=\kmax=k$).\footnote{In this configuration, its asymptotical running time
  matches the non-approximate algorithm \cite[Section IV-B]{HubSan16topk}, but
  its pivot selection speeds up convergence in practice.} We compare it to a
centralized algorithm which uses the same thresholding procedure but gathers all
candidate items at a designated root PE, where it uses sequential selection to
maintain the sample, subsequently labeled \emph{gather}.
The local reservoirs are implemented as B+ trees, based on an implementation of
Bingmann~\cite{tlx} augmented with join/split operations
by Akhremtsev~\cite{akhremtsev2016tree} and rank/select support (see, \eg Ref.~\cite{MehSanPar}).

All pseudorandom numbers are generated with 64-bit Mersenne Twisters~\cite{MatNis98}, using the Intel
Math Kernel Library \cite{intel-mkl} on the Intel machine and the cluster, and
\texttt{dSFMT}\footnote{\url{http://www.math.sci.hiroshima-u.ac.jp/~m-mat/MT/SFMT/}, version 2.2.3}
on the AMD machine.

All of our implementations are publicly available as free software.  The code
and scripts for the first part of our experiments, problems \POne and \PWith,
can be found at \url{https://github.com/lorenzhs/wrs}, while the artifacts for
problem \PRes are available at \url{https://github.com/lorenzhs/reservoir}.

\paragraph*{Optimizations}\label{exp:opt}
For the output sensitive algorithm \osens, we use
an additional optimization that aborts the tree descent and uses the base case
bucket table when fewer than 128 samples are to be drawn
from at least half as many items.  The resulting items are then deduplicated
using a hash table to ensure that each item occurs only once in the output.
A variant without this deduplication is called \osnd and
can be useful for applications that allow the total multiplicity
of sampled items to be split over several
occurrences.

All alias table queries are implemented in a branchless manner.  To do so, we
store the bucket index and its alias in a temporary array $A$ of size two, and
use indexing with the conditional to return the correct item.  Using the
notation of \cref{prel:alias}, let $A[0]=r$ and $A[1]=a_r$.  The query then
returns $A$ indexed with whether the coin came up heads (0) or tails (1), which
uses the result of a comparison as an index into $A$ instead of a conditional
branch.  The result is an improvement of 20\,\% to 25\,\% in sequential query
times (measured for $n=10^8$ items).
In a parallel setting, this optimization yields a smaller improvement since memory access costs dominate.

\paragraph*{Input Data}
Unless specified otherwise, we use uniformly random weights from the interval
$(0,1]$ as inputs.  In our query experiments, we additionally use random
permutations of the weights $\set{1^{-s},2^{-s},\ldots,n^{-s}}$ for a parameter
$s$ to obtain more skewed ``power-law'' or Zipf-like distributions.  For
reservoir sampling, the weights were scaled to the interval $(0,100]$.

\subsection{Sequential Performance}\label{exp:seq}
Surprisingly, many common existing implementations of alias tables (\eg
\texttt{gsl\-\_ran\-\_discrete\-\_preproc} in the GNU Scientific Library (GSL)
\cite{GSL} or \texttt{sample} in the R project for statistical computing
\cite{Rman}) use a struct-of-arrays memory layout for the alias table data
structure.  By using an array of structs instead, we can
improve memory locality, incurring at most one instead of up to two cache misses per
query.  Combined with branchless choice inside the buckets and a faster
random number generator, our query is four times faster than that of
GSL version 2.5 (measured for $n=10^8$).  At the same time, alias table
construction using our implementation of Vose's method is 30\,\% faster than
GSL on the Intel machine and 44\,\% faster on the AMD machine.
Other popular statistics packages, such as NumPy (version 1.5.1, function
\texttt{np.choice}) or Octave (Statistics package version 1.4.0, function
\texttt{randsample}) employ algorithms with superlinear query time for \PWith.  We
therefore use our own implementation of Vose's algorithm as the baseline in our
evaluation.

Among our sequential implementations, construction with Vose's method is
20--25\,\% faster than our sweeping algorithm,\footnote{In the conference version
  of this paper, we reported a somewhat smaller gap.  We have since further
  optimized our implementation of Vose's method.  Specifically, the optimization
  to make queries branchless was previously implemented by storing both the
  index and the alias in the table, instead of using a temporary array of size
  two at query time.} with the gap being slightly
smaller on the AMD machine than on the Intel machine.
Including the time for
memory allocations narrows the gap by around two percentage points.  Thus, when
used purely sequentially, Vose's method remains the fastest, and we use it as
the sequential baseline for our speedup experiments in \cref{exp:cons}.

\def\basesi{1.14}
\def\basesibig{11.4}

\def\basesa{1.35}
\def\basesabig{13.5}

\def\basewi{114}

\def\basewa{136}

\subsection{Construction}\label{exp:cons}
\begin{figure}[p!]
\pgfplotsset{
  log basis x=2,
  log ticks with fixed point,
  legend pos=north west,
  xlabel={Number of threads $p$},
  width=7.4cm,
  height=6.0cm,
  cycle list name=cons,
  xmin=0.75,
}
\begin{subfigure}[t]{0.5\linewidth}
\begin{tikzpicture}
  \begin{semilogxaxis}[
    ylabel={Speedup},
    xmax=200,
    ymin=-3,
    ymax=45,
    xtick={1,2,4,8,16,32,64,128},
  ]
  \addplot coordinates { (1,0.967593) (2,1.91066) (4,3.73202) (8,7.30088) (12,10.5075) (16,13.4637) (24,17.2682) (32,19.9636) (40,20.2748) (52,24.0261) (64,26.7196) (80,26.5024) (104,28.5079) (128,29.982) (158,29.6385) };
  \addlegendentry{2lvl-classic};
  \addplot coordinates { (1,0.854823) (2,1.69423) (4,3.3423) (8,6.44313) (12,9.19046) (16,11.866) (24,16.7803) (32,20.9427) (40,22.6828) (52,23.1487) (64,27.8943) (80,29.7942) (104,33.8571) (128,38.9897) (158,39.0022) };
  \addlegendentry{2lvl-sweep};
  \addplot coordinates { (1,0.611324) (2,1.21174) (4,2.41496) (8,4.53411) (12,6.52944) (16,8.10064) (24,9.62773) (32,12.135) (40,13.8288) (52,18.3177) (64,21.3196) (80,22.0416) (104,26.1153) (128,28.808) (158,30.7472) };
  \addlegendentry{OS};
  \addplot coordinates { (1,0.560116) (2,1.05151) (4,1.96958) (8,3.5978) (12,3.82734) (16,4.95826) (24,6.62288) (32,7.87904) (40,8.80176) (52,9.68551) (64,10.0313) (80,10.0276) (104,10.6187) (128,10.6093) (158,9.46903) };
  \addlegendentry{PSA};
  \addplot coordinates { (1,0.864254) (2,1.68432) (4,3.27629) (8,5.92731) (12,8.26888) (16,10.1716) (24,10.6774) (32,13.607) (40,16.4039) (52,19.3844) (64,20.6767) (80,21.217) (104,23.5785) (128,24.292) (158,21.5195) };
  \addlegendentry{PSA+};

  \addplot+[mark=none, dashed, gray] coordinates {(80, -2) (80, 48)} node[right,pos=0.35] {HT};
  \addplot+[mark=none, dashed, black] coordinates {(0.65, 1) (200, 1)};

  \end{semilogxaxis}
\end{tikzpicture}
\caption{Strong scaling, Intel machine, $n=10^8$}\label{fig:strongintel}
\end{subfigure}
\hfill
\begin{subfigure}[t]{0.49\linewidth}
\begin{tikzpicture}
  \begin{semilogxaxis}[
    ylabel={},
    xmax=80,
    ymin=-2,
    ymax=32,
  ]
  \addplot coordinates { (1,1.02275) (2,2.03518) (4,4.01756) (6,5.88048) (8,7.74419) (12,9.92175) (16,10.9092) (20,11.838) (24,12.4425) (32,12.7044) (40,12.9839) (48,12.9648) (56,12.8986) (62,12.443) };
  \addlegendentry{2lvl-classic};
  \addplot coordinates { (1,0.752354) (2,1.50034) (4,2.97878) (6,4.4485) (8,5.867) (12,8.16534) (16,9.83203) (20,12.0977) (24,14.225) (32,18.1081) (40,20.6239) (48,23.7739) (56,26.3009) (62,27.2999) };
  \addlegendentry{2lvl-sweep};
  \addplot coordinates { (1,0.428983) (2,0.851082) (4,1.69902) (6,2.50934) (8,3.3237) (12,4.60698) (16,5.45725) (20,6.55722) (24,7.76595) (32,9.68526) (40,11.4063) (48,13.313) (56,14.3171) (62,15.0454) };
  \addlegendentry{OS};
  \addplot coordinates { (1,0.573779) (2,1.11683) (4,2.158) (6,2.94017) (8,3.86639) (12,4.74076) (16,5.46258) (20,6.00431) (24,6.30449) (32,6.24832) (40,6.18412) (48,6.31306) (56,6.30669) (62,6.07429) };
  \addlegendentry{PSA};
  \addplot coordinates { (1,0.726279) (2,1.4387) (4,2.83187) (6,2.26748) (8,5.52646) (12,7.44814) (16,9.3351) (20,11.4735) (24,13.5184) (32,17.1125) (40,19.6183) (48,22.3338) (56,24.5529) (62,24.6579) };
  \addlegendentry{PSA+};

  \addplot+[mark=none, dashed, gray] coordinates {(32, -2) (32, 39)} node[right,pos=0.44] {HT};
  \addplot+[mark=none, dashed, black] coordinates {(0.65, 1) (200, 1)};

  \end{semilogxaxis}
\end{tikzpicture}
\caption{Strong scaling, AMD machine, $n=10^8$}\label{fig:strongamd}
\end{subfigure}
\begin{subfigure}[t]{0.5\linewidth}
\vspace*{.88em}
\begin{tikzpicture}
  \begin{semilogxaxis}[
    ylabel={Speedup},
    xmax=200,
    ymin=-3,
    ymax=68,
    xtick={1,2,4,8,16,32,64,128},
  ]
  \addplot coordinates { (1,0.964468) (2,1.90834) (4,3.71209) (8,7.45221) (12,10.9244) (16,14.354) (24,18.8523) (32,22.6018) (40,25.1244) (52,27.9475) (64,29.4641) (80,29.9441) (104,30.1647) (128,31.1517) (158,30.747) };
  \addlegendentry{2lvl-classic};
  \addplot coordinates { (1,0.849894) (2,1.69799) (4,3.35658) (8,6.59044) (12,9.53164) (16,12.6116) (24,18.651) (32,24.1856) (40,28.9902) (52,33.5194) (64,38.8726) (80,43.2497) (104,47.0522) (128,54.0403) (158,60.8395) };
  \addlegendentry{2lvl-sweep};
  \addplot coordinates { (1,0.630765) (2,1.19323) (4,2.34903) (8,4.52272) (12,6.36461) (16,8.41488) (24,11.375) (32,14.1896) (40,15.7661) (52,17.6713) (64,19.3676) (80,19.6415) (104,20.1477) (128,20.7773) (158,20.0358) };
  \addlegendentry{OS};
  \addplot coordinates { (1,0.563327) (2,1.08051) (4,2.10175) (8,4.11682) (12,5.78637) (16,7.15121) (24,8.92957) (32,9.84743) (40,10.3996) (52,11.0196) (64,11.3398) (80,11.0754) (104,11.2782) (128,11.4974) (158,11.2901) };
  \addlegendentry{PSA};
  \addplot coordinates { (1,0.863596) (2,1.71431) (4,3.38762) (8,6.56797) (12,9.4928) (16,12.4954) (24,18.1501) (32,22.9259) (40,26.7022) (52,31.3237) (64,35.4575) (80,40.384) (104,45.0234) (128,50.2189) (158,54.9576) };
  \addlegendentry{PSA+};

  \addplot+[mark=none, dashed, gray] coordinates {(80, -3) (80, 70)} node[right,pos=0.54] {HT};
  \addplot+[mark=none, dashed, black] coordinates {(0.65, 1) (200, 1)};

  \end{semilogxaxis}
\end{tikzpicture}
\caption{Strong scaling, Intel machine, $n=10^9$}\label{fig:strongintelbig}
\end{subfigure}
\hfill
\begin{subfigure}[t]{0.49\linewidth}
\vspace*{.88em}
\begin{tikzpicture}
  \begin{semilogxaxis}[
    ylabel={},
    xmax=80,
    ymin=-2,
    ymax=35,
  ]
  \addplot coordinates { (1,1.02275) (2,2.03569) (4,4.08043) (6,5.99746) (8,7.95027) (12,10.1052) (16,11.1898) (20,12.0637) (24,12.6542) (32,12.527) (40,12.6205) (48,12.5907) (56,12.5549) (62,12.2963) };
  \addlegendentry{2lvl-classic};
  \addplot coordinates { (1,0.753549) (2,1.50328) (4,3.00876) (6,4.50497) (8,5.99781) (12,8.37671) (16,10.1688) (20,12.6174) (24,14.7623) (32,19.0669) (40,21.7058) (48,24.8704) (56,27.9819) (62,30.1386) };
  \addlegendentry{2lvl-sweep};
  \addplot coordinates { (1,0.428967) (2,0.84993) (4,1.69029) (6,2.43594) (8,3.25618) (12,4.39033) (16,5.27052) (20,6.22822) (24,7.08247) (32,8.69736) (40,9.83678) (48,10.5088) (56,10.9071) (62,10.975) };
  \addlegendentry{OS};
  \addplot coordinates { (1,0.577486) (2,1.1256) (4,2.22315) (6,3.03881) (8,3.93001) (12,4.89007) (16,5.51773) (20,6.07705) (24,6.30626) (32,6.20515) (40,6.12148) (48,6.29216) (56,6.30733) (62,6.13343) };
  \addlegendentry{PSA};
  \addplot coordinates { (1,0.729523) (2,1.45141) (4,2.89811) (6,2.3015) (8,5.78155) (12,8.05634) (16,9.81215) (20,12.1573) (24,14.3601) (32,18.4569) (40,21.4157) (48,24.6242) (56,27.6642) (62,28.9896) };
  \addlegendentry{PSA+};

  \addplot+[mark=none, dashed, gray] coordinates {(32, -2) (32, 42)} node[right,pos=0.4] {HT};
  \addplot+[mark=none, dashed, black] coordinates {(0.65, 1) (200, 1)};

  \end{semilogxaxis}
\end{tikzpicture}
\caption{Strong scaling, AMD machine, $n=10^9$}\label{fig:strongamdbig}
\end{subfigure}
\begin{subfigure}[b]{0.5\linewidth}
\vspace*{1em}
\begin{tikzpicture}
  \begin{semilogxaxis}[
    ylabel={Speedup},
    xmax=200,
    ymin=-3,
    ymax=68,
    xtick={1,2,4,8,16,32,64,128},
  ]
  \addplot coordinates { (1,0.96812) (2,1.84495) (4,3.6197) (8,7.16722) (12,10.6047) (16,13.8074) (24,18.3848) (32,22.0493) (40,24.7311) (52,27.5219) (64,29.2382) (80,29.9219) (104,30.0655) (128,31.0105) (158,30.6532) };
  \addlegendentry{2lvl-classic};
  \addplot coordinates { (1,0.851147) (2,1.63554) (4,3.25975) (8,6.35659) (12,9.25763) (16,12.1909) (24,18.0218) (32,23.2521) (40,28.0325) (52,33.0548) (64,38.2093) (80,42.2069) (104,46.8709) (128,53.7674) (158,60.8503) };
  \addlegendentry{2lvl-sweep};
  \addplot coordinates { (1,0.665333) (2,1.27612) (4,2.38325) (8,4.48003) (12,6.31625) (16,8.01749) (24,10.8078) (32,12.6502) (40,14.6798) (52,16.5532) (64,18.158) (80,19.1467) (104,20.1287) (128,20.9632) (158,21.4641) };
  \addlegendentry{OS};
  \addplot coordinates { (1,0.557225) (2,0.943105) (4,1.80666) (8,3.48984) (12,4.8667) (16,5.68857) (24,7.84522) (32,8.16234) (40,9.2092) (52,10.3204) (64,10.9284) (80,10.9336) (104,11.2361) (128,11.4811) (158,11.3232) };
  \addlegendentry{PSA};
  \addplot coordinates { (1,0.856819) (2,1.57088) (4,3.09261) (8,5.70201) (12,8.40062) (16,11.11) (24,15.8189) (32,19.1909) (40,23.5755) (52,27.4609) (64,31.4079) (80,38.9079) (104,45.1427) (128,51.8973) (158,56.8945) };
  \addlegendentry{PSA+};

  \addplot+[mark=none, dashed, gray] coordinates {(80, -3) (80, 75)} node[right,pos=0.5] {HT};
  \addplot+[mark=none, dashed, black] coordinates {(0.65, 1) (200, 1)};

  \end{semilogxaxis}
\end{tikzpicture}
\caption{Weak scaling, Intel machine, $n/p=10^7$}\label{fig:weakintel}
\end{subfigure}
\begin{subfigure}[b]{0.49\linewidth}
\vspace*{1em}
\begin{tikzpicture}
  \begin{semilogxaxis}[
    ylabel={},
    xmax=80,
    ymin=-2,
    ymax=35,
  ]
  \addplot coordinates { (1,1.02232) (2,2.01626) (4,3.97618) (6,5.84538) (8,7.70886) (12,9.99821) (16,11.1145) (20,12.1014) (24,12.8851) (32,12.7992) (40,13.0186) (48,12.9498) (56,12.9385) (62,12.5134) };
  \addlegendentry{2lvl-classic};
  \addplot coordinates { (1,0.754703) (2,1.50895) (4,2.96037) (6,4.42313) (8,5.87111) (12,8.1835) (16,10.0173) (20,12.4144) (24,14.7035) (32,19.1997) (40,21.8702) (48,24.9744) (56,28.0628) (62,30.1169) };
  \addlegendentry{2lvl-sweep};
  \addplot coordinates { (1,0.443326) (2,0.871623) (4,1.72103) (6,2.51472) (8,3.32081) (12,4.55991) (16,5.3756) (20,6.40245) (24,7.35861) (32,8.91064) (40,10.2612) (48,11.0648) (56,11.5304) (62,11.356) };
  \addlegendentry{OS};
  \addplot coordinates { (1,0.583875) (2,1.10788) (4,2.09721) (6,2.86325) (8,3.86183) (12,4.77007) (16,5.54058) (20,6.08879) (24,6.3732) (32,6.34424) (40,6.26037) (48,6.43734) (56,6.46017) (62,6.2464) };
  \addlegendentry{PSA};
  \addplot coordinates { (1,0.728502) (2,1.42246) (4,2.77022) (6,2.2425) (8,5.44639) (12,7.57959) (16,9.58029) (20,11.9511) (24,14.223) (32,18.2634) (40,21.4455) (48,24.6369) (56,27.5991) (62,28.6683) };
  \addlegendentry{PSA+};

  \addplot+[mark=none, dashed, gray] coordinates {(32, -2) (32, 44)} node[right,pos=0.4] {HT};
  \addplot+[mark=none, dashed, black] coordinates {(0.65, 1) (200, 1)};

  \end{semilogxaxis}
\end{tikzpicture}
\caption{Weak scaling, AMD machine, $n/p=10^7$}\label{fig:weakamd}
\end{subfigure}
\caption{Strong (top and middle) and weak (bottom) scaling evaluation of
  parallel alias table construction techniques. Strong scaling with input sizes
  $n=10^8$ (top) and $n=10^9$ (middle), weak scaling with $n/p=10^7$.  Speedups
  are measured relative to our optimized implementation of Vose's method
  (\cref{alg:vose}, \cref{prel:alias}).}
\label{fig:cons}
\end{figure}
Speedups compared to an optimized sequential implementation of Vose's alias
table construction algorithm are shown in \cref{fig:cons} (strong scaling with
fixed $n=10^8$ and $n=10^9$ as well as weak scaling with $n/p=10^7$ uniform random
variates).\footnote{As \cref{fig:cons} shows only speedup values, we
  here provide the running times of the sequential baseline with regard to
  which they are given.  These are, from (a) to (f), $\basesi\,s$, $\basesa\,s$,
  $\basesibig\,s$, $\basesabig\,s$, $\basewi\,ms$, and $\basewa\,ms$.}  %
Observe that for $n=10^8$, the per-thread input size is quite small
for the high thread counts achievable on the Intel machine, resulting in
somewhat noisy measurements (\cref{fig:strongintel}).  Note that \osnd is not
shown, as it only differs from \osens in the query phase, not during
construction.  Speedups do
not increase further once the machine's memory subsystem is saturated, limiting
the speedup that can be achieved with techniques that require multiple passes
over the data (\psa, \twoclassic).  Unlike \psa (and what \psag attempts to mitigate), \twolvl can be constructed
almost independently by the PEs and requires fewer passes over the data.
Sequentially, Vose's method is faster than our sweeping algorithm (see \cref{exp:seq}).
However, when used as base case of \twolvl, our algorithm scales much better to high thread counts
because it reduces the memory traffic and since hyper-threading (HT) helps to
hide the overhead of branch mispredictions.  For large per-thread inputs,
\twosweep achieves more than twice the speedup of \twoclassic.
On random inputs such as the ones used here, \psag can greedily process almost
the entire input, and thus achieves excellent speedups.
Preprocessing for \osens introduces some overhead but requires fewer passes over
memory than \psa and achieves approximately twice the speedups as a result.
Hyper-threading yields larger improvements for \osens construction on the AMD
machine than on the Intel machine.  We can also see that by processing each
group of each thread independently, \osens makes good use of the cache:
the large non-inclusive L3 caches of
the Intel machine gives it a boost for $n=10^8$ (\cref{fig:strongintel}).  Once
groups no longer fit into cache ($n=10^9$, \cref{fig:strongintelbig}), speedups
are somewhat lower.

In the weak scaling experiments (\cref{fig:weakintel,fig:weakamd}), we again see
clearly how \twoclassic and \psa are limited by memory access costs.  Using more
than two threads per available memory channel ($4\times6$ for the Intel machine,
8 for the AMD machine) yields nearly no additional benefit for these algorithms.
Meanwhile, \twosweep, \psag, and -- to a lesser extent -- \osens are not so much limited
by the available memory throughput as by the latency of memory accesses.  As a
result, they continue to scale well, even for the highest thread counts.

\def\basesu{215}
\def\basesp{220}
\def\basewu{21.3}
\def\basewp{21.4}

\begin{figure}[bt]
\pgfplotsset{
    log basis x=2,
    log ticks with fixed point,
    xlabel={Number of threads $p$},
    legend pos=north west,
    width=7.3cm, %
    height=6.2cm,
    xmin=0.75,
    xmax=200,
    xtick={1,2,4,8,16,32,64,128},
    cycle list name=query,
}
\begin{subfigure}[t]{0.5\linewidth}
\begin{tikzpicture}
  \begin{semilogxaxis}[
    ylabel={Speedup},
    ymin=-3,
    ymax=70,
  ]
  \addplot coordinates { (1,0.574025) (2,0.659624) (4,1.20684) (8,2.3535) (12,3.43796) (16,4.51309) (24,6.20992) (32,8.24555) (40,9.4572) (52,12.4872) (64,12.6376) (80,12.7824) (104,15.7699) (128,16.1047) (158,17.4098) };
  \addlegendentry{2lvl};
  \addplot coordinates { (1,0.414297) (2,0.821794) (4,1.5769) (8,2.98549) (12,4.26704) (16,5.42471) (24,6.57864) (32,8.45394) (40,10.3237) (52,13.0051) (64,14.8832) (80,15.3108) (104,17.3271) (128,19.4065) (158,18.4991) };
  \addlegendentry{OS};
  \addplot coordinates { (1,0.896728) (2,1.70751) (4,3.27685) (8,5.93828) (12,7.52067) (16,9.8265) (24,14.257) (32,18.7505) (40,22.2958) (52,27.8558) (64,30.7663) (80,29.2538) (104,33.1219) (128,36.2906) (158,35.6721) };
  \addlegendentry{OS-ND};
  \addplot coordinates { (1,1) (2,1.04577) (4,1.85955) (8,3.53556) (12,5.59581) (16,6.93401) (24,10.0623) (32,13.4704) (40,15.613) (52,18.7073) (64,20.1325) (80,19.4331) (104,24.7631) (128,25.837) (158,26.3432) };
  \addlegendentry{PSA};

  \addplot+[mark=none, dashed, gray] coordinates {(80, -3) (80, 70)} node[right,pos=0.7] {HT};
  \addplot+[mark=none, dashed, black] coordinates {(0.65, 1) (200, 1)};

  \end{semilogxaxis}
\end{tikzpicture}
\caption{Strong scaling, uniform input}\label{fig:qstronguni}
\end{subfigure}
\hfill
\begin{subfigure}[t]{0.477\linewidth}
\begin{tikzpicture}
  \begin{semilogxaxis}[
    ymin=-3,
    ymax=70,
  ]
  \addplot coordinates { (1,0.59113) (2,0.687315) (4,1.25008) (8,2.4168) (12,3.58068) (16,4.60086) (24,5.98706) (32,7.77616) (40,10.3635) (52,12.4172) (64,12.923) (80,13.1215) (104,16.3219) (128,16.6687) (158,17.8424) };
  \addlegendentry{2lvl};
  \addplot coordinates { (1,1.03944) (2,1.98146) (4,3.77392) (8,6.88611) (12,7.76013) (16,10.0467) (24,14.5243) (32,18.7521) (40,22.4201) (52,28.2256) (64,31.425) (80,34.1296) (104,38.2318) (128,40.953) (158,37.7571) };
  \addlegendentry{OS};
  \addplot coordinates { (1,2.09579) (2,3.90633) (4,7.31186) (8,10.9186) (12,16.0714) (16,20.8282) (24,30.5221) (32,36.8469) (40,44.6574) (52,52.9183) (64,58.7102) (80,60.2441) (104,63.1921) (128,63.2386) (158,54.8526) };
  \addlegendentry{OS-ND};
  \addplot coordinates { (1,1) (2,1.06704) (4,1.96044) (8,3.7198) (12,5.62601) (16,7.23745) (24,9.415) (32,12.6777) (40,17.0029) (52,17.1853) (64,20.4939) (80,19.8475) (104,25.0806) (128,26.1013) (158,26.7465) };
  \addlegendentry{PSA};

  \addplot+[mark=none, dashed, gray] coordinates {(80, -3) (80, 70)} node[right,pos=0.7] {HT};
  \addplot+[mark=none, dashed, black] coordinates {(0.65, 1) (200, 1)};

  \end{semilogxaxis}
\end{tikzpicture}
\caption{Strong scaling, power-law input with $s=1$}\label{fig:qstrongpow}
\end{subfigure}
\begin{subfigure}[b]{0.51\linewidth}
\vspace*{1em}
\begin{tikzpicture}
  \begin{semilogxaxis}[
    ylabel={Speedup},
    ymin=-10,
    ymax=160,
  ]
  \addplot coordinates { (1,0.588202) (2,0.664234) (4,1.18522) (8,2.36584) (12,3.48515) (16,4.25889) (24,6.60972) (32,8.59354) (40,10.5053) (52,13.1021) (64,15.9758) (80,19.1936) (104,20.8843) (128,21.3637) (158,22.18) };
  \addlegendentry{2lvl};
  \addplot coordinates { (1,0.392864) (2,0.742057) (4,1.47122) (8,2.89869) (12,4.28908) (16,5.8615) (24,8.6325) (32,12.0246) (40,14.6986) (52,16.2515) (64,20.1255) (80,23.3174) (104,30.1898) (128,37.373) (158,42.955) };
  \addlegendentry{OS};
  \addplot coordinates { (1,0.857293) (2,1.50505) (4,2.87508) (8,5.61739) (12,8.19541) (16,10.7929) (24,15.6013) (32,21.0845) (40,26.4775) (52,34.2178) (64,40.5061) (80,41.6654) (104,52.5623) (128,65.0187) (158,75.4106) };
  \addlegendentry{OS-ND};
  \addplot coordinates { (1,1.00019) (2,1.01615) (4,1.86538) (8,3.72618) (12,5.46965) (16,6.54452) (24,10.4681) (32,13.6915) (40,16.7828) (52,21.0511) (64,25.4605) (80,30.1018) (104,30.3229) (128,33.2073) (158,33.6486) };
  \addlegendentry{PSA};

  \addplot+[mark=none, dashed, gray] coordinates {(80, -10) (80, 160)} node[right,pos=0.6] {HT};
  \addplot+[mark=none, dashed, black] coordinates {(0.65, 1) (200, 1)};

  \end{semilogxaxis}
\end{tikzpicture}
\caption{Weak scaling, uniform input}\label{fig:qweakuni}
\end{subfigure}
\hfill
\begin{subfigure}[b]{0.48\linewidth}
\vspace*{1em}
\begin{tikzpicture}
  \begin{semilogxaxis}[
    ymin=-10,
    ymax=160,
  ]
  \addplot coordinates { (1,0.593817) (2,0.664233) (4,1.19911) (8,2.3792) (12,3.48504) (16,4.33038) (24,6.64933) (32,8.67056) (40,10.5877) (52,13.305) (64,16.1109) (80,19.4607) (104,21.1326) (128,21.9774) (158,22.5736) };
  \addlegendentry{2lvl};
  \addplot coordinates { (1,0.785537) (2,1.42431) (4,2.98996) (8,6.13629) (12,8.2356) (16,10.645) (24,17.1941) (32,22.5018) (40,29.0008) (52,39.8227) (64,46.7217) (80,51.5003) (104,65.1848) (128,75.8323) (158,86.8742) };
  \addlegendentry{OS};
  \addplot coordinates { (1,1.65133) (2,2.39996) (4,4.92692) (8,10.2483) (12,16.2678) (16,21.6673) (24,34.1653) (32,42.8887) (40,52.1572) (52,76.6922) (64,84.632) (80,91.7196) (104,113.587) (128,130.985) (158,143.444) };
  \addlegendentry{OS-ND};
  \addplot coordinates { (1,1.0002) (2,1.06126) (4,1.89671) (8,3.72404) (12,5.50111) (16,6.7803) (24,10.5768) (32,13.8033) (40,16.6222) (52,21.1746) (64,25.6563) (80,30.2082) (104,31.9895) (128,30.9837) (158,33.8866) };
  \addlegendentry{PSA};

  \addplot+[mark=none, dashed, gray] coordinates {(80, -10) (80, 160)} node[right,pos=0.6] {HT};
  \addplot+[mark=none, dashed, black] coordinates {(0.65, 1) (200, 1)};

  \end{semilogxaxis}
\end{tikzpicture}
\caption{Weak scaling, power-law input with $s=1$}\label{fig:qweakpow}
\end{subfigure}
\caption{Query strong and weak scaling for $n=10^9$ items. Sample size
  for strong scaling $k = 10^7$, per-thread sample size for weak scaling
  $k/p = 10^6$.  Speedups relative to seq. alias tables. Intel machine.}
\label{fig:queryscaling}
\end{figure}

\subsection{Queries}

We performed strong and weak scaling experiments for queries
(\cref{fig:queryscaling}) as well as throughput measurements for different
sample sizes (\cref{fig:query}).  Observe that the
different configurations of \twolvl and \psa/\psag, respectively, differ only in how
they approach construction, and have identical query behavior.  Thus, we do not
consider them separately in the query evaluation.

\paragraph*{Scaling} First, consider the scaling experiments of
\cref{fig:queryscaling}.\footnote{The sequential processing times with
  regard to which the speedups are reported are, from (a) to (d),
  $\basesu\,ms$, $\basesp\,ms$, $\basewu\,ms$, and $\basewp\,ms$.}
These experiments were conducted on the Intel machine,
as its highly non-uniform memory access characteristics highlight the
differences between the algorithms.  All speedups are given relative to sampling
sequentially from an alias table.  The strong scaling experiments
(\cref{fig:qstronguni,fig:qstrongpow}) deliberately use a small sample size to
show scaling to low per-thread sample counts ($\approx 64\,000$ samples per thread for 158
threads).  We can see that all algorithms have good scaling behavior.
Hyper-threading (marked ``HT'' in the plots) yields additional speedups, as it
helps to hide memory latency.  This already shows that the bottleneck is
random access latency to memory.  Sampling from \psa and \twolvl is done
completely independently by all threads, with no interaction apart from reading
from the same shared data structures.  Because the Intel machine
has four NUMA nodes, most queries have to access another NUMA node's
memory.  This limits the speedups achievable using alias tables (\psa, \twolvl).

In contrast, for \osens and \osnd, threads access only local memory after a
shared top-level sample assignment stage.  This
benefits scaling, especially on NUMA machines and machines with local
last-level caches, such as newer AMD CPUs.  As a result, \osnd
achieves the best speedups, despite this benchmark producing very few samples
with multiplicity greater than one (\cref{fig:qstronguni}, sample size is 1\,\%
of input size).  On the other hand, deduplication in the base case of \osens has
considerable overhead, making it roughly 25\,\% slower than sampling from an
alias table for such inputs, even sequentially (deduplication is still faster than fully
descending the tree, \ie omitting the optimization described at the start of this section).

The weak scaling experiments of \cref{fig:qweakuni,fig:qweakpow} show even
better speedups because many more samples are drawn here than in our
strong scaling experiment, reducing overheads.  Sampling from a
classical alias table (\psa) achieves a speedup of over 30 here, again limited by
memory accesses rather than computation.  Meanwhile, the
output-sensitive methods (\osens, \osnd) reap the benefits of accessing only
local memory.

Comparing the results for uniformly random weights
(\cref{fig:qstronguni,fig:qweakuni}) with those for power-law inputs
(\cref{fig:qstrongpow,fig:qweakpow}) clearly shows that the query performance of
alias tables (\psa, \twolvl) is independent of the input's weight distribution.
Meanwhile, for \osens and its variant \osnd, uniformly random weights are a
difficult input because they result in few items with multiplicity greater than
one (see the throughput measurements below for further details).  Thus, they
achieve significantly higher speedups for the power-law input.
\enlargethispage*{\baselineskip}
\clearpage

\begin{figure}[p]
\begin{subfigure}[t]{\linewidth}
\begin{tikzpicture}
  \pgfplotsset{group/.cd,
    group size=2 by 2,
    group name=plots,
    horizontal sep=1.87mm,
    vertical sep=1.87mm,
  }
  \begin{groupplot}[
    group style={y descriptions at=edge left},
    groupplot ylabel={Throughput ($\cdot10^9$ samples/s)},
  ]
  \pgfplotsset{
    xmode=log,
    log basis x=10,
    width=7cm,
    height=5cm,
    ymin=0,
    ymax=6.9,
    xmin=5e5,
    xmax=2e9,
    axis y line* = left,
    every axis plot post/.append style={error bars/.cd, y dir=both, y explicit},
    cycle list name=query,
  }
  \nextgroupplot[xticklabels={}]
  \addplot coordinates { (1000000,0.427671) +- (0,0.130853) (3162000,0.673509) +- (0,0.140627) (10000000,0.775444) +- (0,0.143511) (31623000,0.838966) +- (0,0.0797068) (1e+08,0.982941) +- (0,0.0523717) (3.16228e+08,1.01379) +- (0,0.0431118) (1e+09,1.03695) +- (0,0.0434435) };
  \addplot coordinates { (1000000,0.391466) +- (0,0.028266) (3162000,0.682413) +- (0,0.0656421) (10000000,0.851426) +- (0,0.101981) (31623000,1.1076) +- (0,0.0872345) (1e+08,1.74307) +- (0,0.106815) (3.16228e+08,2.41725) +- (0,0.115504) (1e+09,2.81983) +- (0,0.0342995) };
  \addplot coordinates { (1000000,0.52401) +- (0,0.0396964) (3162000,1.01142) +- (0,0.0688715) (10000000,1.75719) +- (0,0.16612) (31623000,1.9272) +- (0,0.138862) (1e+08,3.03053) +- (0,0.224603) (3.16228e+08,5.10147) +- (0,0.240184) (1e+09,6.51046) +- (0,0.204151) };
  \addplot coordinates { (1000000,0.501504) +- (0,0.0377331) (3162000,0.920984) +- (0,0.0601088) (10000000,1.23997) +- (0,0.13902) (31623000,1.33534) +- (0,0.0370468) (1e+08,1.49327) +- (0,0.0793024) (3.16228e+08,1.61975) +- (0,0.00946006) (1e+09,1.66195) +- (0,0.0100674) };

  \draw [mark=none,black] (1.5e6,3.5) node [right] {uniform input};

  \nextgroupplot[xticklabels={}]
  \addplot coordinates { (1000000,0.428666) +- (0,0.13313) (3162000,0.679719) +- (0,0.150812) (10000000,0.777603) +- (0,0.134748) (31623000,0.851488) +- (0,0.0795599) (1e+08,0.98511) +- (0,0.0476552) (3.16228e+08,1.022) +- (0,0.0413446) (1e+09,1.04208) +- (0,0.044175) };
  \addplot coordinates { (1000000,0.354849) +- (0,0.0305754) (3162000,0.610731) +- (0,0.0619728) (10000000,0.779624) +- (0,0.0969506) (31623000,1.00648) +- (0,0.072157) (1e+08,1.50318) +- (0,0.0959265) (3.16228e+08,2.20565) +- (0,0.133053) (1e+09,2.49636) +- (0,0.0456432) };
  \addplot coordinates { (1000000,0.475555) +- (0,0.0418205) (3162000,0.893369) +- (0,0.066395) (10000000,1.4253) +- (0,0.138693) (31623000,1.61076) +- (0,0.11448) (1e+08,2.51845) +- (0,0.181255) (3.16228e+08,4.00248) +- (0,0.178024) (1e+09,5.05132) +- (0,0.157299) };
  \addplot coordinates { (1000000,0.498393) +- (0,0.0367145) (3162000,0.925568) +- (0,0.0581609) (10000000,1.23057) +- (0,0.132309) (31623000,1.26578) +- (0,0.0895261) (1e+08,1.54919) +- (0,0.0124504) (3.16228e+08,1.61851) +- (0,0.0103433) (1e+09,1.65914) +- (0,0.0105511) };

  \draw [mark=none,black] (1.5e6,3.5) node [right] {power-law, $s=0.5$};

  \nextgroupplot[xlabel={Sample size}, ymode=log, ymin=0.17, ymax=4000, ytick={0.1,1,10,100,1000.01,10000},
      yticklabel=\pgfmathparse{exp(\tick)}\pgfmathprintnumber{\pgfmathresult}]
  \addplot coordinates { (1000000,0.434839) +- (0,0.137951) (3162000,0.674514) +- (0,0.143905) (10000000,0.789251) +- (0,0.130606) (31623000,0.851061) +- (0,0.0750289) (1e+08,0.999106) +- (0,0.0471698) (3.16228e+08,1.03978) +- (0,0.0433076) (1e+09,1.06553) +- (0,0.0397741) };
  \addplot coordinates { (1000000,0.447708) +- (0,0.0350118) (3162000,0.962515) +- (0,0.0801285) (10000000,1.69991) +- (0,0.177716) (31623000,2.15256) +- (0,0.339624) (1e+08,3.52927) +- (0,0.325459) (3.16228e+08,5.02418) +- (0,0.367296) (1e+09,7.82547) +- (0,0.405069) };
  \addplot coordinates { (1000000,0.475116) +- (0,0.0440152) (3162000,1.32222) +- (0,0.100582) (10000000,2.50792) +- (0,0.198164) (31623000,4.03865) +- (0,0.441528) (1e+08,5.81464) +- (0,0.592397) (3.16228e+08,8.52165) +- (0,0.530546) (1e+09,13.1128) +- (0,0.879432) };
  \addplot coordinates { (1000000,0.502719) +- (0,0.0380592) (3162000,0.926442) +- (0,0.0558101) (10000000,1.20332) +- (0,0.129736) (31623000,1.33327) +- (0,0.0412196) (1e+08,1.48589) +- (0,0.0888135) (3.16228e+08,1.61493) +- (0,0.0138065) (1e+09,1.65318) +- (0,0.0147159) };

  \draw [mark=none,black] (1.5e6,60) node [right] {power-law, $s=1$};

  \nextgroupplot[xlabel={Sample size}, ymode=log, ymin=0.17, ymax=4000, ytick={0.1,1,10,100,1000.01,10000}]
  \addplot coordinates { (1000000,0.450077) +- (0,0.0909795) (3162000,0.727412) +- (0,0.119777) (10000000,0.894186) +- (0,0.131727) (31623000,1.03908) +- (0,0.0561869) (1e+08,1.25466) +- (0,0.0514635) (3.16228e+08,1.3033) +- (0,0.0555404) (1e+09,1.32059) +- (0,0.064909) };
  \addplot coordinates { (1000000,0.548082) +- (0,0.0517963) (3162000,1.80932) +- (0,0.193829) (10000000,5.80291) +- (0,0.700492) (31623000,20.6354) +- (0,2.7179) (1e+08,72.0388) +- (0,6.32938) (3.16228e+08,241.233) +- (0,19.9653) (1e+09,686.169) +- (0,80.1739) };
  \addplot coordinates { (1000000,0.54381) +- (0,0.0547766) (3162000,1.79148) +- (0,0.187737) (10000000,5.92393) +- (0,0.780376) (31623000,20.7184) +- (0,2.75777) (1e+08,73.2781) +- (0,6.37585) (3.16228e+08,238.185) +- (0,16.4514) (1e+09,705.866) +- (0,73.0254) };
  \addplot coordinates { (1000000,0.502273) +- (0,0.037936) (3162000,0.928747) +- (0,0.0611639) (10000000,1.19775) +- (0,0.133272) (31623000,1.33884) +- (0,0.0307567) (1e+08,1.54999) +- (0,0.0114923) (3.16228e+08,1.61573) +- (0,0.0132463) (1e+09,1.65584) +- (0,0.0139405) };

  \draw [mark=none,black] (1.5e6,200) node [right] {power-law, $s=2$};
  \end{groupplot}

  \makeatletter
  \pgfplotsset{
    every groupplot y label/.style={
      rotate=90,
      at={($({\pgfplots@group@name\space c2r1.north}-|{\pgfplots@group@name\space c2r1.outer
east})!0.5!({\pgfplots@group@name\space c2r2.south}-|{\pgfplots@group@name\space c2r2.outer east})$)},
      anchor=center
    }
  }
  \makeatother
  \begin{groupplot}[
    group style={y descriptions at=edge right},
    groupplot ylabel={Output size / sample size}
  ]
  \pgfplotsset{
    xmode=log,
    log basis x=10,
    width=7cm,
    height=5cm,
    xtick=\empty,
    axis y line* = right,
    ymin=0.42,
    ymax=1.1,
    xmin=5e5,
    xmax=2e9,
    every axis plot post/.append style={dashed, mark options={solid,fill=white}},
    cycle list shift=1,
    cycle list name=query,
  }

  \nextgroupplot
  \addplot coordinates { (1000000,0.998678) (3162000,0.998147) (10000000,0.993558) (31623000,0.97935) (1e+08,0.936484) (3.16228e+08,0.818714) (1e+09,0.56759) };
  \addplot coordinates { (1000000,1.0) (3162000,1.0) (10000000,1.0) (31623000,1.0) (1e+08,1.0) (3.16228e+08,1.0) (1e+09,0.995616) };

  \nextgroupplot
  \addplot coordinates { (1000000,0.997814) (3162000,0.994479) (10000000,0.984199) (31623000,0.960056) (1e+08,0.901427) (3.16228e+08,0.777158) (1e+09,0.556695) };
  \addplot coordinates { (1000000,0.999796) (3162000,0.999477) (10000000,0.998712) (31623000,0.995023) (1e+08,0.983696) (3.16228e+08,0.950984) (1e+09,0.829843) };

  \nextgroupplot[ymode=log,ymin=2e-5,ymax=1,ytick={1e-5,1e-4,1e-3,1e-2,1e-1,1}]
  \addplot coordinates { (1000000,0.489139) (3162000,0.433757) (10000000,0.379911) (31623000,0.325611) (1e+08,0.271679) (3.16228e+08,0.217902) (1e+09,0.164526) };
  \addplot coordinates { (1000000,0.535932) (3162000,0.478146) (10000000,0.420591) (31623000,0.367875) (1e+08,0.314032) (3.16228e+08,0.257316) (1e+09,0.205605) };

  \nextgroupplot[ymode=log,ymin=2e-5,ymax=1,ytick={1e-5,1e-4,1e-3,1e-2,1e-1,1}]
  \addplot coordinates { (1000000,0.001252) (3162000,0.000740354) (10000000,0.0004626) (31623000,0.000262309) (1e+08,0.00013523) (3.16228e+08,7.58978e-05) (1e+09,4.4454e-05) };
  \addplot coordinates { (1000000,0.001504) (3162000,0.000927894) (10000000,0.0005833) (31623000,0.000336559) (1e+08,0.00018332) (3.16228e+08,0.000103742) (1e+09,5.9134e-05) };
  \end{groupplot}
  \begin{customlegend}[
    legend entries={2lvl,OS,OS-ND,PSA,\kern0.6em Right axis:,OS,OS-ND},
    legend style={at={{($(5.25,3.8)$)}}, anchor=center, column sep=2pt},
    legend columns=-1]

    \addlegendimage{mdred,every mark/.append style={fill=mlred},mark=*}%
    \addlegendimage{mdbrown,every mark/.append style={fill=mlbrown},mark=square*}%
    \addlegendimage{mdgreen,every mark/.append style={fill=mlgreen},mark=square*}%
    \addlegendimage{mdgray,every mark/.append style={fill=mlgray},mark=diamond*}%
    \addlegendimage{empty legend}
    \addlegendimage{mdbrown,every mark/.append style={fill=mlbrown},mark=square*,dashed, mark options={solid,fill=white}}%
    \addlegendimage{mdgreen,every mark/.append style={fill=mlgreen},mark=square*,dashed, mark options={solid,fill=white}}%
  \end{customlegend}
\end{tikzpicture}
\vspace*{-2.5mm}
\caption{Intel machine (158 threads).  Note the logarithmic $y$-axes for the
  bottom plots.}\label{fig:queryintel}
\end{subfigure}
\begin{subfigure}[b]{\linewidth}
\vspace*{1.2em}
\begin{tikzpicture}
  \pgfplotsset{group/.cd,
    group size=2 by 2,
    group name=plots,
    horizontal sep=1.87mm,
    vertical sep=1.87mm,
  }
  \begin{groupplot}[
    group style={y descriptions at=edge left},
    groupplot ylabel={Throughput ($\cdot10^9$ samples/s)},
  ]
  \pgfplotsset{
    xmode=log,
    log basis x=10,
    width=7cm,
    height=5cm,
    ymin=0,
    ymax=3.2,
    xmin=5e5,
    xmax=2e9,
    axis y line* = left,
    every axis plot post/.append style={error bars/.cd, y dir=both, y explicit},
    cycle list name=query,
  }
  \nextgroupplot[xticklabels={}]
  \addplot coordinates { (1000000,0.309685) +- (0,0.0521314) (3162000,0.37516) +- (0,0.0439745) (10000000,0.441681) +- (0,0.0154393) (31623000,0.516193) +- (0,0.0127639) (1e+08,0.516928) +- (0,0.00553201) (3.16228e+08,0.526499) +- (0,0.00576385) (1e+09,0.527964) +- (0,0.00478506) };
  \addplot coordinates { (1000000,0.233785) +- (0,0.0233743) (3162000,0.279331) +- (0,0.0232871) (10000000,0.353686) +- (0,0.0209631) (31623000,0.453702) +- (0,0.0216903) (1e+08,0.671774) +- (0,0.0111872) (3.16228e+08,1.01188) +- (0,0.0145573) (1e+09,1.21889) +- (0,0.0111839) };
  \addplot coordinates { (1000000,0.502167) +- (0,0.0471137) (3162000,0.643048) +- (0,0.0540772) (10000000,0.804597) +- (0,0.0603065) (31623000,1.00494) +- (0,0.0455916) (1e+08,1.35689) +- (0,0.0412836) (3.16228e+08,2.2072) +- (0,0.0957625) (1e+09,3.01955) +- (0,0.110211) };
  \addplot coordinates { (1000000,0.486943) +- (0,0.0496204) (3162000,0.626703) +- (0,0.0580914) (10000000,0.772377) +- (0,0.0445931) (31623000,0.853465) +- (0,0.0240337) (1e+08,0.849346) +- (0,0.0338731) (3.16228e+08,0.830187) +- (0,0.0193995) (1e+09,0.885354) +- (0,0.0150213) };

  \draw [mark=none,black] (1.5e6,1.5) node [right] {uniform input};

  \nextgroupplot[xticklabels={}]
  \addplot coordinates { (1000000,0.310224) +- (0,0.0556156) (3162000,0.374989) +- (0,0.0426513) (10000000,0.441995) +- (0,0.0155328) (31623000,0.517664) +- (0,0.0114562) (1e+08,0.515569) +- (0,0.00515549) (3.16228e+08,0.524921) +- (0,0.00240787) (1e+09,0.52751) +- (0,0.00355301) };
  \addplot coordinates { (1000000,0.211987) +- (0,0.021098) (3162000,0.262353) +- (0,0.0231793) (10000000,0.336058) +- (0,0.01762) (31623000,0.428128) +- (0,0.0164121) (1e+08,0.590123) +- (0,0.0115777) (3.16228e+08,0.925584) +- (0,0.00907342) (1e+09,1.09472) +- (0,0.00913401) };
  \addplot coordinates { (1000000,0.42072) +- (0,0.0388921) (3162000,0.545482) +- (0,0.0431704) (10000000,0.686775) +- (0,0.0473711) (31623000,0.888459) +- (0,0.0395262) (1e+08,1.16667) +- (0,0.0374562) (3.16228e+08,1.8118) +- (0,0.0604816) (1e+09,2.34011) +- (0,0.0542642) };
  \addplot coordinates { (1000000,0.492542) +- (0,0.0558899) (3162000,0.618165) +- (0,0.0530459) (10000000,0.769218) +- (0,0.0253639) (31623000,0.8499) +- (0,0.0360486) (1e+08,0.887584) +- (0,0.0225015) (3.16228e+08,0.917413) +- (0,0.0103045) (1e+09,0.861073) +- (0,0.0270686) };
  \draw [mark=none,black] (1e6,1.5) node [right] {power-law, $s=0.5$};

  \nextgroupplot[xlabel={Sample size}, ymode=log, ymin=0.17, ymax=4000, ytick={0.1,1,10,100,1000.01,10000}, %
      yticklabel=\pgfmathparse{exp(\tick)}\pgfmathprintnumber{\pgfmathresult}]
  \addplot coordinates { (1000000,0.309801) +- (0,0.0503313) (3162000,0.375425) +- (0,0.0421208) (10000000,0.442627) +- (0,0.0130986) (31623000,0.516475) +- (0,0.0114754) (1e+08,0.518139) +- (0,0.00464148) (3.16228e+08,0.527358) +- (0,0.00545214) (1e+09,0.529075) +- (0,0.00340926) };
  \addplot coordinates { (1000000,0.399702) +- (0,0.0339018) (3162000,0.591154) +- (0,0.057992) (10000000,0.7993) +- (0,0.0714094) (31623000,1.18456) +- (0,0.0618433) (1e+08,1.68324) +- (0,0.0580736) (3.16228e+08,2.28639) +- (0,0.0491066) (1e+09,3.38422) +- (0,0.0473888) };
  \addplot coordinates { (1000000,0.745431) +- (0,0.0815698) (3162000,1.15553) +- (0,0.118461) (10000000,1.55008) +- (0,0.14252) (31623000,2.11657) +- (0,0.177938) (1e+08,3.31972) +- (0,0.204977) (3.16228e+08,4.63146) +- (0,0.205566) (1e+09,6.66061) +- (0,0.24816) };
  \addplot coordinates { (1000000,0.49471) +- (0,0.0549625) (3162000,0.641267) +- (0,0.0645736) (10000000,0.768068) +- (0,0.0373383) (31623000,0.848571) +- (0,0.0435088) (1e+08,0.894103) +- (0,0.0198393) (3.16228e+08,0.850665) +- (0,0.0430075) (1e+09,0.906728) +- (0,0.0325908) };

  \draw [mark=none,black] (1.5e6,60) node [right] {power-law, $s=1$};

  \nextgroupplot[xlabel={Sample size}, ymode=log, ymin=0.17, ymax=4000, ytick={0.1,1,10,100,1000,10000}]
  \addplot coordinates { (1000000,0.292388) +- (0,0.0379803) (3162000,0.316837) +- (0,0.0312583) (10000000,0.307829) +- (0,0.0400729) (31623000,0.314356) +- (0,0.0106555) (1e+08,0.310192) +- (0,0.0109349) (3.16228e+08,0.312704) +- (0,0.0084313) (1e+09,0.313795) +- (0,0.00828756) };
  \addplot coordinates { (1000000,2.45149) +- (0,0.213333) (3162000,7.77983) +- (0,0.71097) (10000000,25.3183) +- (0,2.31232) (31623000,81.0596) +- (0,7.561) (1e+08,244.876) +- (0,21.5888) (3.16228e+08,672.933) +- (0,90.6535) (1e+09,1668.34) +- (0,293.467) };
  \addplot coordinates { (1000000,2.50165) +- (0,0.207952) (3162000,7.91883) +- (0,0.675405) (10000000,24.6551) +- (0,2.44719) (31623000,82.2777) +- (0,7.59333) (1e+08,277.065) +- (0,29.7864) (3.16228e+08,863.529) +- (0,82.6562) (1e+09,2460.23) +- (0,366.487) };
  \addplot coordinates { (1000000,0.48499) +- (0,0.0469573) (3162000,0.638542) +- (0,0.0632877) (10000000,0.76853) +- (0,0.0309809) (31623000,0.839367) +- (0,0.0539111) (1e+08,0.892648) +- (0,0.0397796) (3.16228e+08,0.901283) +- (0,0.0128959) (1e+09,0.851657) +- (0,0.0992777) };

  \draw [mark=none,black] (1e6,200) node [right] {power-law, $s=2$};
  \end{groupplot}

  \makeatletter
  \pgfplotsset{
    every groupplot y label/.style={
      rotate=90,
      at={($({\pgfplots@group@name\space c2r1.north}-|{\pgfplots@group@name\space c2r1.outer
east})!0.5!({\pgfplots@group@name\space c2r2.south}-|{\pgfplots@group@name\space c2r2.outer east})$)},
      anchor=center
    }
  }
  \makeatother
  \begin{groupplot}[
    group style={y descriptions at=edge right},
    groupplot ylabel={Output size / sample size}
  ]
  \pgfplotsset{
    xmode=log,
    log basis x=10,
    width=7cm,
    height=5cm,
    xtick=\empty,
    axis y line* = right,
    ymin=0.42,
    ymax=1.1,
    xmin=5e5,
    xmax=2e9,
    every axis plot post/.append style={dashed, mark options={solid,fill=white}},
    cycle list shift=1,
    cycle list name=query,
  }

  \nextgroupplot
  \addplot coordinates { (1000000,0.999231) (3162000,0.997945) (10000000,0.993278) (31623000,0.979326) (1e+08,0.936384) (3.16228e+08,0.818828) (1e+09,0.567658) };
  \addplot coordinates { (1000000,1.0) (3162000,1.0) (10000000,1.0) (31623000,1.0) (1e+08,1.0) (3.16228e+08,1.0) (1e+09,0.994201) };

  \nextgroupplot
  \addplot coordinates { (1000000,0.997561) (3162000,0.994316) (10000000,0.984592) (31623000,0.95999) (1e+08,0.901475) (3.16228e+08,0.777127) (1e+09,0.556818) };
  \addplot coordinates { (1000000,0.999779) (3162000,0.999536) (10000000,0.998685) (31623000,0.994773) (1e+08,0.983584) (3.16228e+08,0.95102) (1e+09,0.829871) };

  \nextgroupplot[ymode=log,ymin=2e-5,ymax=1,ytick={1e-5,1e-4,1e-3,1e-2,1e-1,1}]
  \addplot coordinates { (1000000,0.487117) (3162000,0.433014) (10000000,0.379651) (31623000,0.325766) (1e+08,0.271763) (3.16228e+08,0.217821) (1e+09,0.164543) };
  \addplot coordinates { (1000000,0.53185) (3162000,0.476974) (10000000,0.420139) (31623000,0.367765) (1e+08,0.314042) (3.16228e+08,0.257305) (1e+09,0.205627) };

  \nextgroupplot[ymode=log,ymin=2e-5,ymax=1,ytick={1e-5,1e-4,1e-3,1e-2,1e-1,1}]
  \addplot coordinates { (1000000,0.001379) (3162000,0.000802024) (10000000,0.0004438) (31623000,0.000248838) (1e+08,0.00013922) (3.16228e+08,7.71405e-05) (1e+09,4.346e-05) };
  \addplot coordinates { (1000000,0.001782) (3162000,0.00100285) (10000000,0.0005809) (31623000,0.00033093) (1e+08,0.00018935) (3.16228e+08,0.000104915) (1e+09,5.823e-05) };
  \end{groupplot}
\end{tikzpicture}
\vspace*{-2.5mm}
\caption{AMD machine (62 threads).  Note the logarithmic $y$-axes for the bottom
  plots.}\label{fig:queryamd}
\end{subfigure}
\caption{Query throughput of the different methods for $n=10^9$, using all
  available cores.  Top left: uniform inputs, top right: power-law with $s=0.5$,
  bottom left: power-law $s=1$, bottom right: power-law $s=2$.  Dashed lines on
  the right $y$-axis belong to the same-colored solid lines on the left $y$-axis
  and show fraction of output size over sample size for output-sensitive
  algorithms.}
\label{fig:query}
\end{figure}

\begin{figure}[p]
\begin{subfigure}[t]{\linewidth}
\begin{tikzpicture}
  \pgfplotsset{group/.cd,
    group size=2 by 2,
    group name=plots,
    horizontal sep=1.87mm,
    vertical sep=1.87mm,
  }
  \begin{groupplot}[
    group style={y descriptions at=edge left},
    groupplot ylabel={Time per sample ($ns$)}
  ]
  \pgfplotsset{
    xmode=log,
    log basis x=10,
    legend columns=4,
    legend style={anchor=south,at={(axis description cs:0,1)},yshift=1mm},
    width=7cm, %
    height=5cm,
    ymin=0,
    ymax=3.1,
    xmin=5e5,
    xmax=2e9,
    axis y line* = left,
    every axis plot post/.append style={error bars/.cd, y dir=both, y explicit},
    cycle list name=query,
  }
  \nextgroupplot[xticklabels={}]
  \addplot coordinates { (1000000,2.3398) (3162000,1.4879) (10000000,1.29821) (31623000,1.21724) (1e+08,1.08634) (3.16228e+08,1.20488) (1e+09,1.69889) };
  \addplot coordinates { (1000000,2.5562) (3162000,1.46848) (10000000,1.18235) (31623000,0.922013) (1e+08,0.612602) (3.16228e+08,0.505326) (1e+09,0.624741) };
  \addplot coordinates { (1000000,1.90963) (3162000,0.990792) (10000000,0.572895) (31623000,0.529898) (1e+08,0.35235) (3.16228e+08,0.239441) (1e+09,0.27059) };
  \addplot coordinates { (1000000,1.99533) (3162000,1.08809) (10000000,0.811864) (31623000,0.764762) (1e+08,0.715082) (3.16228e+08,0.754126) (1e+09,1.06) };

  \draw [mark=none,black] (3e6,1.8) node [right] {uniform input};

  \nextgroupplot[xticklabels={}]

  \addplot coordinates { (1000000,2.33775) (3162000,1.47979) (10000000,1.3063) (31623000,1.22349) (1e+08,1.12591) (3.16228e+08,1.25908) (1e+09,1.72352) };
  \addplot coordinates { (1000000,2.82405) (3162000,1.64695) (10000000,1.30292) (31623000,1.03508) (1e+08,0.737863) (3.16228e+08,0.583403) (1e+09,0.719463) };
  \addplot coordinates { (1000000,2.10725) (3162000,1.1259) (10000000,0.71268) (31623000,0.646767) (1e+08,0.440406) (3.16228e+08,0.321496) (1e+09,0.355559) };
  \addplot coordinates { (1000000,2.01069) (3162000,1.08673) (10000000,0.825462) (31623000,0.823042) (1e+08,0.715948) (3.16228e+08,0.795043) (1e+09,1.08252) };

  \draw [mark=none,black] (3e6,1.8) node [right,text width=2cm, align=right] {power-law, $s=0.5$};

  \nextgroupplot[xlabel={Sample size},ymode=log,ymin=0.2,ymax=4e4,ytick={0.1,1,10,100,1e3,1e4,1e5}]
  \addplot coordinates { (1000000,4.7142) (3162000,3.41799) (10000000,3.33713) (31623000,3.60832) (1e+08,3.68437) (3.16228e+08,4.41475) (1e+09,5.7039) };
  \addplot coordinates { (1000000,4.5787) (3162000,2.39527) (10000000,1.54939) (31623000,1.42662) (1e+08,1.04301) (3.16228e+08,0.913659) (1e+09,0.776656) };
  \addplot coordinates { (1000000,4.31457) (3162000,1.74364) (10000000,1.05021) (31623000,0.760378) (1e+08,0.633071) (3.16228e+08,0.538674) (1e+09,0.463494) };
  \addplot coordinates { (1000000,4.07766) (3162000,2.48853) (10000000,2.1888) (31623000,2.30329) (1e+08,2.47735) (3.16228e+08,2.84246) (1e+09,3.67637) };

  \draw [mark=none,black] (3e6,100) node [right] {power-law, $s=1$};

  \nextgroupplot[xlabel={Sample size},ymode=log,ymin=0.2,ymax=4e4,ytick={0.1,1,10,100,1e3,1e4,1e5}]

  \addplot coordinates { (1000000,1609.89) (3162000,1772.24) (10000000,2554.95) (31623000,3900.83) (1e+08,5750.39) (3.16228e+08,9852.86) (1e+09,17236.6) };
  \addplot coordinates { (1000000,1322.01) (3162000,712.504) (10000000,393.7) (31623000,196.423) (1e+08,100.151) (3.16228e+08,53.2317) (1e+09,33.1733) };
  \addplot coordinates { (1000000,1332.4) (3162000,719.6) (10000000,385.657) (31623000,195.636) (1e+08,98.4575) (3.16228e+08,53.9128) (1e+09,32.2476) };
  \addplot coordinates { (1000000,1442.59) (3162000,1388.05) (10000000,1907.41) (31623000,3027.44) (1e+08,4654.72) (3.16228e+08,7947.62) (1e+09,13746.8) };

  \draw [mark=none,black] (2e7,800) node [right] {power-law, $s=2$};
  \end{groupplot}

  \makeatletter
  \pgfplotsset{
    every groupplot y label/.style={
      rotate=90,
      at={($({\pgfplots@group@name\space c2r1.north}-|{\pgfplots@group@name\space c2r1.outer
east})!0.5!({\pgfplots@group@name\space c2r2.south}-|{\pgfplots@group@name\space c2r2.outer east})$)},
      anchor=center %
    }
  }
  \makeatother
  \begin{groupplot}[
    group style={y descriptions at=edge right},
    groupplot ylabel={Output size / sample size}
  ]
  \pgfplotsset{
    xmode=log,
    log basis x=10,
    width=7cm, %
    height=5cm,
    xtick=\empty,
    axis y line* = right,
    ymin=0.42,
    ymax=1.1,
    xmin=5e5,
    xmax=2e9,
    every axis plot post/.append style={dashed, mark options={solid,fill=white}},
    cycle list shift=1,
    cycle list name=query,
  }

  \nextgroupplot
  \addplot coordinates { (1000000,0.998678) (3162000,0.998147) (10000000,0.993558) (31623000,0.97935) (1e+08,0.936484) (3.16228e+08,0.818714) (1e+09,0.56759) };
  \addplot coordinates { (1000000,1.0) (3162000,1.0) (10000000,1.0) (31623000,1.0) (1e+08,1.0) (3.16228e+08,1.0) (1e+09,0.995616) };

  \nextgroupplot
  \addplot coordinates { (1000000,0.997814) (3162000,0.994479) (10000000,0.984199) (31623000,0.960056) (1e+08,0.901427) (3.16228e+08,0.777158) (1e+09,0.556695) };
  \addplot coordinates { (1000000,0.999796) (3162000,0.999477) (10000000,0.998712) (31623000,0.995023) (1e+08,0.983696) (3.16228e+08,0.950984) (1e+09,0.829843) };

  \nextgroupplot[ymode=log,ymin=1.5e-5,ymax=1,ytick={1e-5,1e-4,1e-3,1e-2,1e-1,1}]
  \addplot coordinates { (1000000,0.489139) (3162000,0.433757) (10000000,0.379911) (31623000,0.325611) (1e+08,0.271679) (3.16228e+08,0.217902) (1e+09,0.164526) };
  \addplot coordinates { (1000000,0.535932) (3162000,0.478146) (10000000,0.420591) (31623000,0.367875) (1e+08,0.314032) (3.16228e+08,0.257316) (1e+09,0.205605) };

  \nextgroupplot[ymode=log,ymin=1.5e-5,ymax=1,ytick={1e-5,1e-4,1e-3,1e-2,1e-1,1}]
  \addplot coordinates { (1000000,0.001252) (3162000,0.000740354) (10000000,0.0004626) (31623000,0.000262309) (1e+08,0.00013523) (3.16228e+08,7.58978e-05) (1e+09,4.4454e-05) };
  \addplot coordinates { (1000000,0.001504) (3162000,0.000927894) (10000000,0.0005833) (31623000,0.000336559) (1e+08,0.00018332) (3.16228e+08,0.000103742) (1e+09,5.9134e-05) };
  \end{groupplot}
  \begin{customlegend}[
    legend entries={2lvl,OS,OS-ND,PSA,\kern0.6em Right axis:,OS,OS-ND},
    legend style={at={{($(5.25,3.8)$)}}, anchor=center, column sep=2pt},
    legend columns=-1]

    \addlegendimage{mdred,every mark/.append style={fill=mlred},mark=*}%
    \addlegendimage{mdbrown,every mark/.append style={fill=mlbrown},mark=square*}%
    \addlegendimage{mdgreen,every mark/.append style={fill=mlgreen},mark=square*}%
    \addlegendimage{mdgray,every mark/.append style={fill=mlgray},mark=diamond*}%
    \addlegendimage{empty legend}
    \addlegendimage{mdbrown,every mark/.append style={fill=mlbrown},mark=square*,dashed, mark options={solid,fill=white}}%
    \addlegendimage{mdgreen,every mark/.append style={fill=mlgreen},mark=square*,dashed, mark options={solid,fill=white}}%
  \end{customlegend}
\end{tikzpicture}
\vspace*{-2.5mm}
\caption{Intel machine (158 threads). Note the logarithmic $y$-axes for the bottom plots.}
\label{fig:tpiintel}
\end{subfigure}
\begin{subfigure}[b]{\linewidth}
\vspace*{1.2em}
\begin{tikzpicture}
  \pgfplotsset{group/.cd,
    group size=2 by 2,
    group name=plots,
    horizontal sep=1.87mm,
    vertical sep=1.87mm,
  }
  \begin{groupplot}[
    group style={y descriptions at=edge left},
    groupplot ylabel={Time per sample ($ns$)}
  ]
  \pgfplotsset{
    xmode=log,
    log basis x=10,
    width=7cm, %
    height=5cm,
    ymin=0,
    ymax=5.3,
    xmin=5e5,
    xmax=2e9,
    axis y line* = left,
    every axis plot post/.append style={error bars/.cd, y dir=both, y explicit},
    cycle list name=query,
  }
  \nextgroupplot[xticklabels={}]
  \addplot coordinates { (1000000,3.23125) (3162000,2.67115) (10000000,2.27921) (31623000,1.97834) (1e+08,2.06562) (3.16228e+08,2.3199) (1e+09,3.33658) };
  \addplot coordinates { (1000000,4.2803) (3162000,3.58754) (10000000,2.84626) (31623000,2.25082) (1e+08,1.58949) (3.16228e+08,1.20708) (1e+09,1.44524) };
  \addplot coordinates { (1000000,1.9927) (3162000,1.55838) (10000000,1.25116) (31623000,1.01619) (1e+08,0.786929) (3.16228e+08,0.553381) (1e+09,0.583396) };
  \addplot coordinates { (1000000,2.055) (3162000,1.59902) (10000000,1.30336) (31623000,1.19654) (1e+08,1.25717) (3.16228e+08,1.47126) (1e+09,1.98971) };

  \draw [mark=none,black] (5e6,3.4) node [right] {uniform input};

  \nextgroupplot[xticklabels={}]
  \addplot coordinates { (1000000,3.2303) (3162000,2.68232) (10000000,2.29818) (31623000,2.01247) (1e+08,2.15118) (3.16228e+08,2.45127) (1e+09,3.40457) };
  \addplot coordinates { (1000000,4.72725) (3162000,3.83393) (10000000,3.02264) (31623000,2.43335) (1e+08,1.87941) (3.16228e+08,1.39017) (1e+09,1.64055) };
  \addplot coordinates { (1000000,2.3819) (3162000,1.84395) (10000000,1.47906) (31623000,1.17257) (1e+08,0.950641) (3.16228e+08,0.71019) (1e+09,0.767461) };
  \addplot coordinates { (1000000,2.03458) (3162000,1.62714) (10000000,1.32054) (31623000,1.22577) (1e+08,1.24955) (3.16228e+08,1.40256) (1e+09,2.0857) };

  \draw [mark=none,black] (3e6,3.36) node [right,text width=2cm,align=right] {power-law, $s\kern-.4em =\kern-.4em 0.5$};

  \nextgroupplot[xlabel={Sample size},ymode=log,ymin=0.2,ymax=2e5,ytick={0.1,1,10,100,1e3,1e4,1e5}]
  \addplot coordinates { (1000000,6.61888) (3162000,6.14188) (10000000,5.9509) (31623000,5.94603) (1e+08,7.1045) (3.16228e+08,8.70421) (1e+09,11.4874) };
  \addplot coordinates { (1000000,5.13015) (3162000,3.90054) (10000000,3.29542) (31623000,2.59251) (1e+08,2.18692) (3.16228e+08,2.00763) (1e+09,1.79589) };
  \addplot coordinates { (1000000,2.7508) (3162000,1.99546) (10000000,1.69928) (31623000,1.45092) (1e+08,1.10886) (3.16228e+08,0.991098) (1e+09,0.91248) };
  \addplot coordinates { (1000000,4.14492) (3162000,3.59572) (10000000,3.42942) (31623000,3.619) (1e+08,4.1171) (3.16228e+08,5.39605) (1e+09,6.70286) };

  \draw [mark=none,black] (5e6,600) node [right] {power-law, $s=1$};

  \nextgroupplot[xlabel={Sample size},ymode=log,ymin=0.2,ymax=2e5,ytick={0.1,1,10,100,1e3,1e4,1e5}]
  \addplot coordinates { (1000000,2474.39) (3162000,4058.69) (10000000,7428) (31623000,12940.9) (1e+08,23364.3) (3.16228e+08,41189.5) (1e+09,73102) };
  \addplot coordinates { (1000000,295.119) (3162000,165.291) (10000000,90.3125) (31623000,50.1858) (1e+08,29.5963) (3.16228e+08,19.1403) (1e+09,13.7496) };
  \addplot coordinates { (1000000,289.203) (3162000,162.39) (10000000,92.7415) (31623000,49.4428) (1e+08,26.1578) (3.16228e+08,14.9157) (1e+09,9.32394) };
  \addplot coordinates { (1000000,1491.75) (3162000,2013.87) (10000000,2975.23) (31623000,4846.56) (1e+08,8119) (3.16228e+08,14290.9) (1e+09,26934.6) };

  \draw [mark=none,black] (5e6,600) node [right] {power-law, $s=2$};
  \end{groupplot}

  \makeatletter
  \pgfplotsset{
    every groupplot y label/.style={
      rotate=90,
      at={($({\pgfplots@group@name\space c2r1.north}-|{\pgfplots@group@name\space c2r1.outer
east})!0.5!({\pgfplots@group@name\space c2r2.south}-|{\pgfplots@group@name\space c2r2.outer east})$)},
      anchor=center %
    }
  }
  \makeatother
  \begin{groupplot}[
    group style={y descriptions at=edge right},
    groupplot ylabel={Output size / sample size}
  ]
  \pgfplotsset{
    xmode=log,
    log basis x=10,
    width=7cm, %
    height=5cm,
    xtick=\empty,
    axis y line* = right,
    ymin=0.42,
    ymax=1.1,
    xmin=5e5,
    xmax=2e9,
    every axis plot post/.append style={dashed, mark options={solid,fill=white}},
    cycle list shift=1,
    cycle list name=query,
  }

  \nextgroupplot
  \addplot coordinates { (1000000,0.999231) (3162000,0.997945) (10000000,0.993278) (31623000,0.979326) (1e+08,0.936384) (3.16228e+08,0.818828) (1e+09,0.567658) };
  \addplot coordinates { (1000000,1.0) (3162000,1.0) (10000000,1.0) (31623000,1.0) (1e+08,1.0) (3.16228e+08,1.0) (1e+09,0.994201) };

  \nextgroupplot
  \addplot coordinates { (1000000,0.997561) (3162000,0.994316) (10000000,0.984592) (31623000,0.95999) (1e+08,0.901475) (3.16228e+08,0.777127) (1e+09,0.556818) };
  \addplot coordinates { (1000000,0.999779) (3162000,0.999536) (10000000,0.998685) (31623000,0.994773) (1e+08,0.983584) (3.16228e+08,0.95102) (1e+09,0.829871) };

  \nextgroupplot[ymode=log,ymin=1.5e-5,ymax=1,ytick={1e-5,1e-4,1e-3,1e-2,1e-1,1}]
  \addplot coordinates { (1000000,0.487117) (3162000,0.433014) (10000000,0.379651) (31623000,0.325766) (1e+08,0.271763) (3.16228e+08,0.217821) (1e+09,0.164543) };
  \addplot coordinates { (1000000,0.53185) (3162000,0.476974) (10000000,0.420139) (31623000,0.367765) (1e+08,0.314042) (3.16228e+08,0.257305) (1e+09,0.205627) };

  \nextgroupplot[ymode=log,ymin=1.5e-5,ymax=1,ytick={1e-5,1e-4,1e-3,1e-2,1e-1,1}]
  \addplot coordinates { (1000000,0.001379) (3162000,0.000802024) (10000000,0.0004438) (31623000,0.000248838) (1e+08,0.00013922) (3.16228e+08,7.71405e-05) (1e+09,4.346e-05) };
  \addplot coordinates { (1000000,0.001782) (3162000,0.00100285) (10000000,0.0005809) (31623000,0.00033093) (1e+08,0.00018935) (3.16228e+08,0.000104915) (1e+09,5.823e-05) };

  \end{groupplot}
\end{tikzpicture}
\vspace*{-2.5mm}
\caption{AMD machine (62 threads). Note the logarithmic $y$-axes for the bottom plots.}
\label{fig:tpiamd}
\end{subfigure}
\caption{Time per unique output item for the different methods for $n=10^9$
  using all available cores on
  the Intel (top) and AMD (bottom) machines. Top left:
  uniform inputs, top right: power-law with $s=0.5$, bottom left: power-law
  $s=1$, bottom right: power-law $s=2$.  Dashed lines on
  the right $y$-axis belong to the same-colored solid lines on the left $y$-axis
  and show fraction of output size over sample size for output-sensitive
  algorithms.}
\label{fig:timeperitem}
\end{figure}

\FloatBarrier

\paragraph*{Throughput}%
\cref{fig:query} shows the query throughput of the different approaches
(note the logarithmic $y$-axes in the lower plots).  We can
see that \twolvl suffers a significant slowdown compared to \psa on all inputs
since an additional query for a meta-item is needed (this is also clearly
visible in \cref{fig:queryscaling}).  Its throughput is around 40\,\% lower than
sampling from an alias table.  Nonetheless, both algorithms are limited by the
latency of random accesses to memory for the (bottom) alias tables on both machines.
As long as the sample contains few duplicates (\cf the dashed lines with the scale
on the right $y$-axis, which belong to the solid lines of the same color and
marker shape), the cost of base case deduplication in \osens exceeds the
benefits of increased memory locality.  On the AMD machine, where memory access
locality is less important, this results in
higher throughput for \twolvl than for \osens for small sample sizes when inputs
are not too skewed (\cref{fig:queryamd}). As expected, when there are many duplicates, the output
sensitive algorithms (\osens and \osnd) perform very well.  Omitting base case deduplication (\osnd)
doubles throughput for uniform inputs and does no harm for skewed inputs, making
\osnd the consistently fastest algorithm.  In comparison, adding sequential
deduplication to normal alias tables using a fast hash table (Google's
\texttt{dense\_hash\_map}%
\footnote{\url{https://github.com/sparsehash/sparsehash}, version 2.0.2}) takes
five to six times longer for uniform inputs ($n=10^8$, $10^7$ samples) compared to
simply storing samples in an array without deduplication.

Lastly, we observe that \twolvl and \psa throughput levels off after $10^{7.5}\approx 3\cdot 10^7$
samples on the AMD machine (\cref{fig:queryamd}), whereas it keeps increasing slightly on the Intel
machine.  This is likely due to the Intel machine's higher overall memory
bandwidth.

\paragraph*{Time per Unique Item} \cref{fig:timeperitem} shows the time per
unique item in the sample.  This metric penalizes algorithms for emitting an
item with multiplicity greater than one multiple times, which affects the methods
based on alias tables considerably when the number of unique samples is low.
We can see that the \twolvl and \psa approaches do well
as long as few items have multiplicity larger than one, \ie the dashed lines are
close to 1 (right scale).  In these cases, what
\osens gains from having higher locality of memory accesses is lost in base case
deduplication, especially on the AMD machine.  Because it may split items'
total multiplicity over several occurrences by omitting base case deduplication,
\osnd does not suffer from this and is the fastest algorithm.  The
same is true for the power-law inputs with $s=0.5$ and $s=1$ (observe that as in
\cref{fig:query}, the $y$-axes for the lower two plots are logarithmic).  For
power-law inputs with $s=2$, the running time of \osens and \osnd is nearly
constant regardless of sample size.
This is because the number of
unique items is very low for this input (measured in the low thousands), and
thus what little time is spent on sampling is dominated by thread
synchronization and scheduling overhead.  These overheads are particularly
problematic with the 158 threads on the Intel machine (\cref{fig:tpiintel}),
where they add up to several milliseconds, around ten times as much as on the AMD
machine (\cref{fig:tpiamd}).

\subsection{Reservoir Sampling}\label{exp:res}

\begin{figure}[bt]
\begin{tikzpicture}
  \pgfplotsset{every axis title/.append style={at={(0.5,-0.3)}}}
  \begin{groupplot}[
    group style = {group size = 3 by 2, horizontal sep = 1.5mm, vertical sep = 1.5mm},
    xmode=log,
    log basis x=2,
    log ticks with fixed point,
    legend pos=north west,
    width=5.6cm,
    height=5.4cm,
    domain=0.5:512,
    cycle list name=mylist,
    ymin=0,
    ymax=1,
    enlarge y limits=0.07, %
    xtick={1,2,4,8,16,32,64,128,256},
    xticklabels={1,,4,,16,,64,,256},
    ytick={0,0.5,1},
    minor ytick={0,0.1,0.2,0.3,0.4,0.5,0.6,0.7,0.8,0.9,1},
  ]
\nextgroupplot[xticklabels={}, ylabel={Rel. efficiency, $b=10^6$},
legend style = {
  column sep=3pt,
  legend columns = 3,
  legend to name = legendweak,
  /tikz/column 2/.style={column sep=12pt},
  /tikz/column 4/.style={column sep=12pt}
}]
  \addplot coordinates { (1,1.0) (2,0.985991) (4,0.950526) (8,0.940587) (16,0.928867) (32,0.88253) (64,0.85218) (128,0.842504) (256,0.81955) };
  \addlegendentry{ours};
  \addplot coordinates { (1,0.999468) (2,0.986394) (4,0.953691) (8,0.942245) (16,0.931659) (32,0.884781) (64,0.85834) (128,0.843592) (256,0.820089) };
  \addlegendentry{ours-8};
  \addplot coordinates { (1,0.99834) (2,0.990779) (4,0.96015) (8,0.950484) (16,0.941406) (32,0.894651) (64,0.869798) (128,0.84709) (256,0.823745) };
  \addlegendentry{gather};

  \addplot[overlay,gray,dashed,update limits=false]{1} node[left,pos=0.95] {ideal scaling};
\nextgroupplot[xticklabels={}, yticklabels={}]
  \addplot coordinates { (1,1.0) (2,0.982923) (4,0.944764) (8,0.926281) (16,0.915781) (32,0.866708) (64,0.818923) (128,0.821059) (256,0.780626) };
  \addlegendentry{ours};
  \addplot coordinates { (1,1.00002) (2,0.987998) (4,0.957425) (8,0.941338) (16,0.933978) (32,0.887116) (64,0.852154) (128,0.839073) (256,0.813771) };
  \addlegendentry{ours-8};
  \addplot coordinates { (1,0.987101) (2,0.981745) (4,0.955967) (8,0.936846) (16,0.850853) (32,0.814467) (64,0.766596) (128,0.713254) (256,0.563409) };
  \addlegendentry{gather};

  \legend{};
  \addplot[overlay,gray,dashed,update limits=false]{1} node[left,pos=0.95] {ideal scaling};
\nextgroupplot[xticklabels={}, yticklabels={}]
  \addplot coordinates { (1,1.0) (2,0.968062) (4,0.937348) (8,0.913221) (16,0.889287) (32,0.811545) (64,0.793322) (128,0.790034) (256,0.735468) };
  \addlegendentry{ours};
  \addplot coordinates { (1,1.00217) (2,0.977719) (4,0.959595) (8,0.938095) (16,0.921525) (32,0.864176) (64,0.847218) (128,0.830635) (256,0.793863) };
  \addlegendentry{ours-8};
  \addplot coordinates { (1,0.849084) (2,0.832245) (4,0.810451) (8,0.777305) (16,0.727868) (32,0.638372) (64,0.522743) (128,0.370263) (256,0.11592) };
  \addlegendentry{gather};

  \legend{};
  \addplot[overlay,gray,dashed,update limits=false]{1} node[left,pos=0.95] {ideal scaling};
  \nextgroupplot[ylabel={Rel. efficiency, $b=10^5$},xlabel={Nodes ($p/20$), $k=10^3$}]
  \addplot coordinates { (1,1.0) (2,0.892154) (4,0.719446) (8,0.653672) (16,0.576039) (32,0.456795) (64,0.3462) (128,0.339914) (256,0.27327) };
  \addlegendentry{ours};
  \addplot coordinates { (1,0.953833) (2,0.836863) (4,0.686118) (8,0.619913) (16,0.549203) (32,0.433297) (64,0.336545) (128,0.327744) (256,0.261581) };
  \addlegendentry{ours-8};
  \addplot coordinates { (1,1.07382) (2,1.04557) (4,0.908646) (8,0.820152) (16,0.73861) (32,0.591463) (64,0.492267) (128,0.454325) (256,0.38588) };
  \addlegendentry{gather};

  \legend{};
  \addplot[overlay,gray,dashed,update limits=false]{1} node[below left,pos=0.95] {ideal scaling};
  \nextgroupplot[xlabel={Nodes ($p/20$), $k=10^4$},yticklabels={}]
  \addplot coordinates { (1,1.0) (2,0.859135) (4,0.712258) (8,0.655194) (16,0.57311) (32,0.453512) (64,0.33581) (128,0.327844) (256,0.249446) };
  \addlegendentry{ours};
  \addplot coordinates { (1,0.957578) (2,0.822165) (4,0.705546) (8,0.648258) (16,0.573681) (32,0.445837) (64,0.354894) (128,0.34037) (256,0.26907) };
  \addlegendentry{ours-8};
  \addplot coordinates { (1,0.889663) (2,0.83523) (4,0.763866) (8,0.723702) (16,0.596291) (32,0.499612) (64,0.415783) (128,0.379837) (256,0.259124) };
  \addlegendentry{gather};

  \legend{};
  \addplot[overlay,gray,dashed,update limits=false]{1} node[below left,pos=0.95] {ideal scaling};
  \nextgroupplot[xlabel={Nodes ($p/20$), $k=10^5$},yticklabels={}]
  \addplot coordinates { (1,1.0) (2,0.793803) (4,0.630756) (8,0.555611) (16,0.456583) (32,0.29245) (64,0.287843) (128,0.253938) (256,0.174047) };
  \addlegendentry{ours};
  \addplot coordinates { (1,1.00228) (2,0.82912) (4,0.716015) (8,0.645291) (16,0.560559) (32,0.391281) (64,0.370456) (128,0.334912) (256,0.249597) };
  \addlegendentry{ours-8};
  \addplot coordinates { (1,0.193007) (2,0.189177) (4,0.183561) (8,0.179373) (16,0.168911) (32,0.154608) (64,0.13468) (128,0.106281) (256,0.0406049) };
  \addlegendentry{gather};

  \legend{};
  \addplot[overlay,gray,dashed,update limits=false]{1} node[below left,pos=0.95] {ideal scaling};
\end{groupplot}

\hidelinks
\node at ($(group c2r1) + (0,2.3)$) {\ref{legendweak}};
\restorelinks

\end{tikzpicture}
\vspace*{-8mm}
\caption{Reservoir sampling, weak scaling with different batch (rows) and sample sizes
  (columns).  Efficiency (speedup divided by number of PEs) is measured relative
  to our algorithm with single-pivot selection (\emph{ours}) on 1 node (20 cores).}
\label{fig:resweak}
\end{figure}

The results of a scaling experiment are shown in \cref{fig:resweak}, with plots
in two rows for per-PE mini-batch sizes $b=10^6$ (top) and $10^5$ (bottom),
and three columns with sample sizes $k$ of $10^3$, $10^4$,
and $10^5$ items, from left to right.\footnote{The single-node times per batch
  with regard to which the speedups are reported are, by row from top left to
  bottom right, %
  $3.10\,ms$, $3.14\,ms$, $3.24\,ms$; %
  $116\,\mu s$, $122\,\mu s$, $131\,\mu s$.}  %
The figure shows the efficiency of the algorithms, measured relative to our
algorithm with single-pivot selection (\emph{ours}) for the same batch and
sample sizes on a single node ($p=20$ cores/PEs).

We can see that our algorithm shows good scaling, especially for larger
mini-batch sizes.  Using multiple pivots for selection (\emph{ours-8}, in blue)
is especially beneficial for larger sample sizes (center and right columns),
where it reduces average recursion depth by a factor of around 2.5 -- from 7.3
to 2.7 for $k=10^5$ and from 4.3 to 1.8 for $k=10^4$ -- compared to a much
smaller improvement from 1.9 to 1.1 for $k=10^3$ (left column), where the
average recursion depth is already very low when using a single pivot.  This
results in selection running time improvements of up to $35\,\%$ for $k=10^5$
and around $20\,\%$ for $k=10^4$, with no significant improvement for $k=10^3$.
Because local processing is a significant part of overall processing time, the
actual overall running time improvement is only around $8\,\%$ ($k=10^5$ and
$b=10^6$). As expected, smaller samples (left column) achieve slightly better
speedups than larger ones (right column).  The causes for this lie in the
$\Ohsmall{\log k \log p}$ latency of the selection algorithm and increased local
processing time due to larger local reservoirs.

We also see clearly that the centralized algorithm (\emph{gather}, in green)
performs well only when the sample size is very small ($k=10^3$, left column) or
both sample size and batch size are rather small (bottom center).  The common
denominator of these settings is that few candidates need to be gathered per
batch.  It begins struggling even with $k=10^4$ for large batch sizes.  For
larger sample sizes ($k=10^5$, right column), it performs very poorly.  All
algorithms show better -- and in the case of our algorithm, near-optimal --
scaling for large batch sizes, as communication overhead is much less noticeable
than for small batches, where local processing is fast.

Overall, \emph{ours-8} is the most consistently fast algorithm, and using
multiple pivots (\emph{ours-8}) should be preferred over the version with
single-pivot selection (\emph{ours}).

Ref.~\cite[Chapter 3.8.4]{lorenzdiss} presents additional experimental results.  A running time composition analysis
confirms that for $k=10^5$, local processing dominates our algorithms' running
time for larger batch sizes, whereas the centralized algorithm spends most of
its time on selection and, as $p$ grows, gathering the candidate items.  Strong
scaling experiments in Ref.~\cite{lorenzdiss} confirm consistent scaling of our
method as long as per-PE batch sizes do not drop below around~$10^4$.

\section{Conclusions and Future Work}\label{concl}

We have presented parallel algorithms for a wide spectrum of weighted
sampling problems running on a variety of machine models. The
algorithms are at the same time efficient in theory and sufficiently
simple to be practically useful.
Our experiments show that alias table computation can be parallelized
efficiently.  For random inputs, the variant \psag combines the high
construction speedup of \twolvl with the query performance of a normal,
single-level alias table.  For skewed inputs, \osens and even more so its
variant \osnd show excellent query performance.

Since the preparation of this paper, we have supervised a Master's thesis
exploring alias table construction and querying on GPUs \cite{mahpl}.  With
careful consideration of GPUs' memory systems, significant speedups over CPUs
can be achieved, but table size is limited by the available memory.

Future work could consider further
implementations, such as the block-wise algorithm of
Section~\ref{se:block}.  For implementing \PWithout one
could perhaps obtain constant factor improvements by studying more
accurate estimators for the output size of a sample with replacement.
For \PPerm it would be interesting to see whether our transformation
to integer sorting in a small range is worthwhile in practice compared to
direct radix sorting or state-of-the-art comparison-based sorting
algorithms.  %

Our algorithms are formulated for
explicitly parallel models of computing. It could also be interesting
to adapt them to the more abstract, implicitly parallel models used in
data bases or big data tools such as MapReduce \cite{DeanGhemawat08},
Spark \cite{Spark16}, or Thrill \cite{sanders2016_1}.
Let us discuss this for problem \POne.
Thrill seems
particularly suited for data structures like alias tables since it
natively supports arrays. Thrill also supports a prefix sum operation
that could easily be used to build a rejection sampling based table
similar to the one used by \citet{BLK13} and outlined in
Section~\ref{ss:succinct}.  Other systems that are based on sets or relations
could emulate arrays by explicitly storing the array index as the first
component of a tuple.
Batched queries can be supported by join or
merge operations with a set of random indices. However, for this to be
efficient, the batch needs to have size $\Om{n}$ or the join operation
must keep the larger join partner in-memory allowing sublinear time
access to the data.
We can also consider a more coarse-grained approach emulating the
distributed-memory approach of Section~\ref{dist} with $p=\sqrt{n}$
PEs.  The table of $\sqrt{n}$ meta-items would be replicated while
each distributed object would store a small alias table for
$\sqrt{n}$ objects.

\bibliographystyle{plainnats}
\bibliography{diss}

\end{document}